%% file: IPM_JMLR_Index.tex
\documentclass[twoside,11pt]{article}

\usepackage{blindtext}

\usepackage[utf8]{inputenc}
\usepackage[T1]{fontenc}
\usepackage{lmodern}
\usepackage{hyperref}
\usepackage{url}            
\usepackage{booktabs}       
\usepackage{amsfonts}       
\usepackage[numbers]{natbib}
\usepackage{nicefrac}       
\usepackage{algorithm,algpseudocode}
\usepackage{amsfonts}
\usepackage{multicol}
\usepackage{amsmath}
\usepackage{amsthm}
\usepackage{amssymb}
\usepackage[title]{appendix}
\usepackage{bm}
\usepackage{courier}
\usepackage[usenames,dvipsnames]{color}
\usepackage{enumitem}
\usepackage{graphicx}
\usepackage[margin=0.9in]{geometry}
\usepackage{url}
\usepackage{subfigure}
\usepackage{multirow}
\usepackage[dvipsnames]{xcolor}
\usepackage{epsfig}

\hypersetup{
	colorlinks,
	linkcolor={red!80!black},
	citecolor={blue!50!black},
	urlcolor={blue!80!black}
}

\title{
On the Convergence of Inexact Predictor-Corrector Methods for Linear Programming
}

\input{./Definitions}

\graphicspath{ {./figures/} }
\DeclareGraphicsExtensions{.pdf, .jpeg, .png}

\allowdisplaybreaks

\begin{document}

\title{Faster Randomized Interior Point Methods for Tall/Wide Linear Programs
}
\date{}
\author{Agniva Chowdhury\footnote{\scriptsize Computer Science and Mathematics Division, Oak Ridge National Laboratory, TN, USA,\,\url{chowdhurya@ornl.gov}.}
      \and
Gregory Dexter\footnote{\scriptsize Department of Computer Science, Purdue University,
West Lafayette, IN, USA, \url{{gdexter, pdrineas}@purdue.edu}.}
      \and
Palma London\footnote{\scriptsize Operations Research and Information Engineering, Cornell University, Ithaca, NY, USA, \url{plondon@cornell.edu}.}
  \and
  Haim Avron\footnote{\scriptsize School of Mathematical Sciences, Tel Aviv University, Tel Aviv, Israel, \url{haimav@tauex.tau.ac.il}.}
\and
  Petros Drineas\footnotemark[2]}

\maketitle

\input{0_Abstract}

\input{1_Introduction}
\input{2_Background}
\input{3_PCG_solver}
\input{4_Algorithm}
\input{X_extensions}
\input{5_Experiments}

\input{6_Conclusion}

\section*{Acknowledgements}
AC, PD, and GD were partially supported by NSF 1760353 and 1814041 and DOE SC0022085. HA was partially supported by BSF 2017698. PL was supported by an Amazon Graduate Fellowship in Artificial Intelligence. This work was done when AC was a graduate student in the Department of Statistics, Purdue University.  Oak Ridge National Laboratory is operated by UT-Battelle LLC for the U.S. Department of Energy under contract number DEAC05-00OR22725.


\newpage

\setlength{\bibsep}{6pt}
\bibliographystyle{abbrvnat}
{
\bibliography{bibliography}
}

\newpage
\begin{appendices}
	
	\newpage

\input{A_Appendix}

\end{appendices}

\end{document}

%% file: 0_Abstract.tex





\vspace{-5mm}
\begin{abstract}
	 Linear programming (LP) is an extremely useful tool which has been successfully applied to solve various problems in a wide range of areas, including operations research, engineering, economics, or even more abstract mathematical areas such as combinatorics. It is also used in many machine learning applications, such as $\ell_1$-regularized SVMs, basis pursuit, nonnegative matrix factorization, etc.  Interior Point Methods (IPMs) are one of the most popular methods to solve LPs both in theory and in practice. Their underlying complexity is dominated by the cost of solving a system of linear equations at each iteration.  In this paper, we consider both \emph{feasible} and \emph{infeasible} IPMs for the special case where the number of variables is much larger than the number of constraints. Using tools from Randomized Linear Algebra, we present a preconditioning technique that, when combined with the iterative solvers such as Conjugate Gradient or Chebyshev Iteration, provably guarantees that IPM algorithms (suitably modified to account for the error incurred by the approximate solver), converge to a feasible, approximately optimal solution, without increasing their iteration complexity. Our empirical evaluations verify our theoretical results on both real-world and synthetic data.

\end{abstract}

%% file: 1_Introduction.tex
\section{Introduction}\label{sec:intro}

Linear programming (LP) is one of the most useful tools available to theoreticians and practitioners throughout science and engineering. It has been extensively used to solve various problems in a wide range of areas, including operations research, engineering, economics, or even in more abstract mathematical areas such as combinatorics. In machine learning and numerical optimization, LP appears in numerous settings, including $\ell_1$-regularized SVMs~\citep{zhu20041}, basis pursuit (BP)~\citep{yang2011alternating}, sparse inverse covariance matrix estimation (SICE)~\citep{yuan2010high}, the nonnegative matrix factorization (NMF)~\citep{recht2012factoring}, MAP inference~\citep{meshi2011alternating},  adversarial deep learning~\citep{weng2018towards,wong2018provable} etc. Not surprisingly, designing and analyzing LP algorithms is a topic of paramount importance in computer science and applied mathematics.

\textcolor{black}{The first algorithm for general-purpose LPs was the famous \textit{simplex algorithm}, proposed by~\citep{dantzig1951maximization}. It worked well in practice, but was shown to have exponential worst-case running times~\citep{klee1972good}. The first polynomial time algorithm for general LPs was the \emph{ellipsoid method}~\citep{khachiyan1979polynomial}, which is rather slow in practice compared to the simplex algorithm. This motivated further research on LP algorithms which are efficient in both theory and practice.} One of the most successful paradigms for solving LPs is the family of Interior Point Methods (IPMs), pioneered by Karmarkar in the mid 1980s~\citep{karmarkar84}. Path-following IPMs (also called central-path algorithms) and, in particular, long-step path following IPMs, are among the most practical approaches for solving linear programs. \textcolor{black}{See Section~\ref{sxn:comparison} for a detailed overview of recent work on path-following IPMs.}

Consider the standard form of the primal LP problem:
\begin{flalign}
	\min\,\cbb^\ts\xb\,,\text{ subject to }\Ab\xb=\bb\,,\xb\ge \zero\,,\label{eq:primal}
\end{flalign}
where $\Ab\in\RR{\dimone}{\dimtwo}$, $\bb\in\R{\dimone}$, and $\cbb\in\R{\dimtwo}$ are the inputs, and $\xb\in\R{\dimtwo}$ is the vector of the primal variables. The associated dual problem is
\begin{flalign}
	\max\,\bb^\ts\yb\,,\text{ subject to }\Ab^\ts\yb+\sbb=\cbb\,,\sbb\ge\zero\,,\label{eq:dual}
\end{flalign}
where $\yb\in\R{\dimone}$ and $\sbb\in\R{\dimtwo}$ are the vectors of the dual and slack variables respectively.
Triplets $(\xb, \yb, \sbb)$ that uphold both eqns. \eqref{eq:primal} and \eqref{eq:dual} are called \emph{primal-dual solutions}.  Path-following IPMs typically
converge towards a primal-dual solution by operating as follows: given the current iterate $(\xb^{k},\yb^{k},\sbb^{k})$, they compute the Newton search direction $(\Delta\xb,\Delta\yb,\Delta\sbb)$ and update the current iterate by making a step towards the search direction. 

To compute the search direction, one standard approach~\citep{NW06} involves solving the
\emph{normal equations}\footnote{Another widely used approach is to solve the augmented system~\citep{NW06}. This approach is less relevant for this paper.}:
\begin{flalign}
	\Ab\Db^2\Ab^\ts\Delta\yb=~&\pb.\label{eq:normal}
\end{flalign}
Here, $\Db = \Xb^{1/2}\Sb^{-1/2}$ is a diagonal matrix, $\Xb,\Sb\in\RR{\dimtwo}{\dimtwo}$ are diagonal matrices whose $i$-th diagonal entries are equal to $\xb_i$ and $\sbb_i$, respectively, and $\pb \in \R{\dimone}$ is a vector whose exact definition is given in eqn.~(\ref{eqn:pdef})\footnote{The superscript $k$ in eqn.~(\ref{eqn:pdef}) simply indicates iteration count and is omitted here for notational simplicity.}. Given $\Delta\yb$, computing $\Delta \sbb$ and $\Delta \xb$ only involves matrix-vector products.

The core computational bottleneck in IPMs is the need to solve the linear system of eqn.~(\ref{eq:normal}) at each iteration. This leads to two key challenges: first, for high-dimensional matrices $\Ab$, solving the linear system is computationally prohibitive. Most implementations of IPMs use a \emph{direct solver}; see Chapter 6 of~\citep{NW06}. However, if $\Ab\Db^2\Ab^\ts$ is large and dense, direct solvers are computationally impractical. If $\Ab\Db^2\Ab^\ts$ is sparse, specialized direct solvers have been developed, but these do not apply to many LP problems, 
especially those arising in machine learning applications,
due to irregular sparsity patterns. 
%
%
Second, an alternative to direct solvers is the use of iterative solvers, but the situation is further complicated since $\Ab\Db^2\Ab^\ts$ is typically ill-conditioned. Indeed, as IPM algorithms approach the optimal primal-dual solution, the diagonal matrix $\Db$ becomes ill-conditioned, which also results in the matrix $\Ab\Db^2\Ab^\ts$ becoming ill-conditioned. Additionally, using approximate solutions for the linear system of eqn.~(\ref{eq:normal}) causes certain invariants, which are crucial for guaranteeing the convergence of IPMs, to be violated; see Section~\ref{sxn:contrib} for details.

In this paper, we address the aforementioned challenges, for the special case where $m \ll n$, 
i.e., the number of constraints is much smaller than the number of variables; see Section~\ref{sxn:extensions} for a generalization. This is a common setting in 
many applications 
of LP solvers. For example, in machine learning, 
$\ell_1$-SVMs and basis pursuit problems often exhibit such structure when the number of available features ($n$)  is larger than the number of objects ($m$). Indeed, this setting has been of interest in recent work on LPs \citep{Donoho05,Bienstock06,LondonAAAI2018}.  

For simplicity of exposition, we also assume that the constraint matrix $\Ab$ has full rank, equal to $m$. First, we propose and analyze two Krylov subspace-based solvers, namely,  preconditioned Conjugate Gradient (CG) and preconditioned Chebyshev iteration for the normal equations of eqn.~(\ref{eq:normal}), using matrix sketching constructions from the Randomized Linear Algebra (RLA) literature. We develop a preconditioner for $\Ab\Db^2\Ab^\ts$
%
%
using matrix sketching which allows us to prove strong convergence guarantees for the \textit{residual}
of both the CG and the Chebyshev iteration. 
%
Second, building upon the work of~\citet{Mon03}, we propose and analyze a provably accurate long-step IPM algorithm. Our framework works for both \emph{feasible} and \emph{infeasible} starting points. The proposed IPM solves the normal equations using iterative solvers. 
We note that a non-trivial concern is that the use of iterative solvers and matrix sketching tools implies that the normal equations at each iteration will be solved only approximately. In our proposed IPM framework, we develop a novel way to \textit{correct} for the error induced by the approximate solution in order to guarantee convergence. Importantly, this correction step is relatively computationally light, \textcolor{black}{unlike a similar step previously proposed by~\citet{Mon03}, which works similarly, but is computationally inefficient}. Third, we empirically show that our algorithm performs well in practice. We consider solving LPs that arise from $\ell_1$-regularized SVMs and test them on a variety of synthetic and real-world data sets. Several extensions of our work are discussed in Section~\ref{sxn:extensions}.

\subsection{Our contributions}\label{sxn:contrib}

Our point of departure in this work is the introduction of preconditioned, iterative solvers for solving eqn.~(\ref{eq:normal}). Preconditioning is used to address the ill-conditioning of the matrix $\Ab\Db^2\Ab^\ts$. Iterative solvers allow the computation of approximate solutions using only matrix-vector products while avoiding matrix inversion, Cholesky or LU factorizations, etc. A preconditioned formulation of eqn.~(\ref{eq:normal}) is:
\begin{flalign}
	\Qb^{-1}\Ab\Db^2\Ab^\ts\Delta\yb=\Qb^{-1}\pb,
	\label{eq:precond}
\end{flalign}
where $\Qb \in \mathbb{R}^{m \times m}$ is the preconditioning matrix; $\Qb$ should be easily invertible (see~\citep{axelsson1984finite,golub2013matrix} for background). An alternative yet equivalent formulation of eqn.~(\ref{eq:precond}), which is more amenable to theoretical analysis, is
\begin{flalign}
	\Qb^{-\nicefrac{1}{2}}\Ab\Db^2\Ab^\ts\Qb^{-\nicefrac{1}{2}}\zb=~\Qb^{-\nicefrac{1}{2}}\pb,\label{eq:precond_alt}
\end{flalign}
where $\zb\in\R{\dimone}$ is a vector such that $\Delta\yb=\Qb^{-\nicefrac{1}{2}}\zb$. Note that the matrix in the left-hand side of the above equation is always symmetric, which is not necessarily the case for eqn.~\eqref{eq:precond}. We do emphasize that one can use eqn.~\eqref{eq:precond} in the actual implementation of the preconditioned solver; eqn.~(\ref{eq:precond_alt}) is much more useful in theoretical analyses. 

Recall that we focus on the special case where $\Ab \in \mathbb{R}^{m \times n}$ has $m \ll n$, i.e., it is a short-and-fat matrix. Our first contribution starts with the design and analysis of a preconditioner for the Conjugate Gradient solver. The preconditioner satisfies, with high probability, the following bound:
\begin{flalign}\label{eq:pdcond1}
	\frac{2}{2+\zeta} \leq \sigma^2_{\min}(\Qb^{-\frac{1}{2}}\Ab\Db) \leq \sigma^2_{\max}(\Qb^{-\frac{1}{2}}\Ab\Db) \leq \frac{2}{2-\zeta},
\end{flalign}
for some error parameter $\zeta \in  [0,1]$. In the above, $\sigma_{\min}(\cdot)$ and $\sigma_{\max}(\cdot)$ correspond to the smallest and largest singular value of the matrix in parentheses.
The above condition says that the preconditioner effectively reduces the condition number of $\Ab\Db$ to a constant. We note that the particular form of the lower and upper bounds in eqn.~(\ref{eq:pdcond1}) was chosen to simplify our derivations.

RLA matrix-sketching techniques allow us to construct preconditioners for all short-and-fat matrices that satisfy the above inequality \textit{and} can be inverted efficiently. Such constructions go back to the work of~\citet{Avron2010}; see Section~\ref{sxn:PCG} for details on the construction of $\Qb$ and its inverse. Importantly, given such a preconditioner, we then prove that the resulting CG iterative solver satisfies
\begin{flalign}
	\|\Qb^{-\nicefrac{1}{2}}\Ab\Db^2\Ab^\ts\Qb^{-\nicefrac{1}{2}}\tilde{\zb}^t-\Qb^{-\nicefrac{1}{2}}\pb\|_2\leq
	\zeta^t \|\Qb^{-\nicefrac{1}{2}}\pb\|_2. \label{eq:pdcond2}
\end{flalign}
Here $\tilde{\zb}^t$ is the approximate solution returned by the CG iterative solver after $t$ iterations. In words, the above inequality states that the \textit{residual}
achieved after $t$ iterations of the CG iterative solver drops exponentially fast. 
\textcolor{black}{Given eqn.~\eqref{eq:pdcond1}, we derive eqn.~\eqref{eq:pdcond2} using the monotonic decrease of the preconditioned residual norms of CG.
}
To the best of our knowledge, \textcolor{black}{such monotonicity of the residual error} is not known in the CG literature: indeed, it is actually well-known that the residual error of CG may oscillate\,\citep{fong2012cg}, even in cases where the energy norm of the solution error decreases monotonically. However, we prove that if the preconditioner is sufficiently good, i.e., it satisfies the constraint of eqn.~\eqref{eq:pdcond1}, then the residual error decreases monotonically as well, \textcolor{black}{resulting in eqn.~\eqref{eq:pdcond2}. It is slightly better than the residual-norm bound derived directly from the energy norm of the solution-error, which is the standard bound for CG. In the later case, the bound in eqn.~\eqref{eq:pdcond2} would have another constant factor involving $\kappa(\Qb^{-1/2}\Ab\Db)$, the condition number of the preconditioned matrix. Using the aforementioned monotonicity of the residual norms, we are able to get rid of that constant factor. In addition, such monotonic decrease of the residual norms can also be of independent interest in CG literature. See Section~\ref{sxn:cg} for details. }

\vspace{0.4mm}
In addition, we also analyze another popular Krylov subspace-based solver, namely, \emph{Chebyshev iteration}\,\citep{barrett1994templates,gutknecht2002chebyshev}. This method avoids the computation of the inner products which is typically needed for CG or other non-stationary methods. Inner products are communication intensive in parallel or distributed settings, and, as such, are detrimental to the performance in such setups. However, there is a trade-off; in order to avoid the computation of the inner products, it requires adequate knowledge about spectrum of the coefficient matrix $\Qb^{-\nicefrac{1}{2}}\Ab\Db^2\Ab^\ts\Qb^{-\nicefrac{1}{2}}$, which, in our context, is nothing but the condition in eqn.~\eqref{eq:pdcond1}. Therefore, given eqn.~\eqref{eq:pdcond1}, we prove that Chebyshev iteration also satisfies eqn.~\eqref{eq:pdcond2}.

Our second contribution is the analysis of a novel variant of a long-step IPM algorithm proposed by~\citet{Mon03}. First, we analyze the feasible version of it, which is the one that starts from a strictly feasible point and stays feasible across all the iterations -- namely, any point $(\xb^k,\yb^k,\sbb^k)$ with $(\xb^k,\sbb^k)>\zero$ such that such that they are both primal and dual feasible \ie, $\Ab\xb^k=\bb$ and $\Ab^\ts\yb^k+\sbb^k=\cbb$. However, the use of an approximate solver prevents the iterates $(\xb^k,\yb^k,\sbb^k)$ from satisfying the above primal and dual feasibility constraints \emph{exactly} right after the first iteration of the IPM. In order to account for the error caused by the CG solver and to push the iterates back to the feasibility, \citet{Mon03} introduced a perturbation vector $\vb$ which needs to satisfy a certain linear invariant exactly. Again, we use RLA matrix sketching principles to propose an efficient construction for $\vb$ that provably satisfies the invariant. 

Finally, we combine the above two primitives to prove that Algorithm~\ref{algo:iipm} in Section~\ref{sxn:IIPM} satisfies the following theorem.
\begin{theorem}\label{thm:2f}
	Let $0 \leq \epsilon \leq 1$ be an accuracy parameter. Consider the long-step feasible IPM Algorithm~\ref{algo:iipm} (Section~\ref{sxn:IIPM}) that solves eqn.~(\ref{eq:precond_alt}) using the iterative solver of Algorithm~\ref{algo:PCG} (Section~\ref{sxn:PCG}). Assume that the iterative solver runs with accuracy parameter $\zeta = \nicefrac{1}{2}$ and iteration count
	$t = \Ocal (\log n)$. 
	%
	Then, with probability at least 0.9, the long-step feasible IPM converges after $\Ocal(n \log \nicefrac{1}{\epsilon})$ iterations.
\end{theorem}

We note that the constant success probability above is for simplicity of exposition and can be easily amplified using standard techniques. Also, at each iteration of our long-step feasible IPM algorithm, the running time is $\Ocal((\nnz{\Ab}+m^3)\log n)$. See Section~\ref{sxn:IIPM} for a detailed discussion of the overall running time. However, finding a strictly feasible initial point is a non-trivial task. Therefore, we also briefly discuss the infeasible version of the long-step IPM algorithm in Section~\ref{sxn:ipm_i}.  

Our empirical evaluation demonstrates that our algorithm requires an order of magnitude much fewer inner CG iterations than a standard IPM using CG, while producing a comparably accurate solution (see Section~\ref{sec:exp}). In practice, our empirical evaluation also indicates that using a CG solver with our sketching-based preconditioner does not increase the number of (outer) iterations of the infeasible IPM, compared to unpreconditioned CG or  a direct linear solver. Furthermore, there are instances where our solver performs much better than non-preconditioned CG in terms of (outer) iteration count.

\subsection{Comparison with Related Work}\label{sxn:comparison}

There is a large body of literature on solving LPs using IPMs, thus we only review literature that is immediately relevant to our work. Recall that we solve the normal equations inexactly at each iteration, and develop a method to \emph{correct} for the error incurred. We focus on IPMs that start with a strictly feasible initial point and discuss papers that present related ideas.

The use of an approximate iterative solver for eqn.~(\ref{eq:normal}), followed by a correction step to ``fix'' the approximate solution was proposed by~\citet{Mon03} (see our discussion in Section~\ref{sxn:contrib}). We propose efficient, RLA-based approaches to precondition and solve eqn.~(\ref{eq:normal}), as well as a novel approach to correct for the approximation error in order to guarantee the convergence of the IPM algorithm. Specifically,~\citet{Mon03} propose to solve eqn.~\eqref{eq:normal} using the so-called \emph{maximum weight basis} preconditioner \citep{RV93}.
However, computing such a preconditioner needs access to a maximal linearly independent set of columns of $\Ab\Db$ in each iteration, which is costly, taking $\Ocal(\dimone^2\dimtwo)$ time in the worst-case.
More importantly, while~\citep{Mon04} provides a bound on the condition number of the preconditioned matrix that depends only on properties of $\Ab$, and is independent of $\Db$, this bound might, in general, be very large. In contrast, our bound is a constant and it does not depend on properties of $\Ab$ or its dimension. In addition, \citet{Mon03} assume a bound on the two-norm of the residual of the preconditioned system, but it is unclear how the proposed preconditioner guarantees such a bound. Similar concerns exist for the construction of the correction vector $\vb$ proposed by~\citet{Mon03}, which our work alleviates. \textcolor{black}{In this context, there is a long list of works on designing efficient preconditioners for solving the linear system at each iteration of IPM. For a more detailed discussion,  we refer the interested readers to the survey of \citep{Gondzio12}.}

{\color{black} In the Theoretical Computer Science community, following the lines of \citep{karmarkar84}, there has been a series of efforts to design faster LP solvers with improved worst-case time complexity. The running time of \citep{karmarkar84} was improved by~\citep{renegar1988polynomial} and~\citep{vaidya1989speeding} which proposed an algorithm that takes $\widetilde{\Ocal}\big(n^{2.5+o(1)}\big)$ time.
Current state-of-the-art running times involve fast matrix multiplication (a theoretically appealing yet impractical approach). For example,~\citep{CLS19} proposed an algorithm that runs in $\widetilde{\Ocal}\big((n^{\omega}+n^{2.5-\alpha / 2+o(1)}+n^{2+1 / 6})\big)$ time, where $\omega$ is the exponent of matrix multiplication and $\alpha$ is the dual exponent of matrix multiplication. For $\omega \approx 2.38$ and $\alpha \approx 0.31$ this time complexity boils down to $\widetilde{\Ocal}\big(n^{\omega+o(1)}\big)$. More recently,~\citep{jiang2021faster} reduced the running time of~\citep{CLS19} to $\widetilde{\Ocal}\big((n^{\omega}+n^{2.5-\alpha / 2+o(1)}+n^{2+1 /18})\big)$, which further reduces the gap between  matrix multiplication and solving LPs.

Work by~\citep{lee2019solvingERM,van2020deterministic, song2021oblivious} achieved the same running time as~\citep{CLS19}, using the same so-called \emph{lazy update} framework of~\citep{CLS19}. However, there are subtle differences between these works in terms of the underlying sketching and sampling techniques, as well as on approaches that achieve fast queries for the so-called \emph{projection maintenance} data structure that handles infeasibilities over iterations. For example, while the work of~\citep{CLS19} involves a non-oblivious sampling scheme whose sampling set and size changes over iterations,~\citep{song2021oblivious} utilizes oblivious  sketching through an iterative framework to approximate the central path. On the other hand, while~\citep{CLS19, van2020deterministic, song2021oblivious} solve the linear system exactly, other works only maintain infeasible updates in each iteration. Similarly, while~\citep{song2021oblivious} can leverage sparse embeddings, most of the aforementioned works (including~\citep{lee2019solvingERM}) require the usage of dense sketching matrices, which could hinder the sparsity structure of the original linear program. 

All the aforementioned solvers are designed for the $m\approx n$ setting.
If $\Ab$ is dense and rectangular with $n\gg m$,~\citep{brand2020solving} provides the theoretically fastest solver, with an iteration complexity of $\widetilde{\mathcal{O}}(\sqrt{m})$ and a total time complexity equal to $\widetilde{\mathcal{O}}\left(m n+m^{3}\right)$. For sparse rectangular matrices, the best per iteration complexities are given by $\Ocaltil\left(\operatorname{nnz}(\mathbf{A})+m^{\omega}\right)$~\citep{LS14} and $\Ocaltil(\nnz{\Ab} + m^2)$ (see~\citep{LS15}); both algorithms  have iteration complexity $\widetilde{\mathcal{O}}(\sqrt{m})$.
Overall, we note that
\citep{LS14,LS15,lee2019solvingERM,brand2020solving,cohen2021solving,song2021oblivious, jiang2021faster} proposed and analyzed theoretically ground-breaking algorithms for LPs based on novel tools such as the so-called \emph{projection maintenance}, \emph{inverse maintenance}, \emph{fast matrix multiplication}, etc. for accelerating the linear system solvers in IPMs. In contrast, our paper differs from all these aforementioned TCS works in at least the following three directions:}
\begin{itemize}
\color{black}
    \item First, all these aforementioned approaches are primarily focused on theoretically fast but practically inefficient short-step path following methods, where the iterates are constrained within a narrow, restrictive neighborhood of the central path in the interior of the feasible region. Therefore, the algorithms based on short-step IPMs do not have much room to maneuver and the progress that they make in each iteration is limited, at least in practice\footnote{Theoretically, the worst-case iteration complexity of short-step central path methods is given by $\Ocaltil(\sqrt{n})$, which is the best complexity bound known for IPMs.}. 
    {\color{black} Some of the recent TCS works (for example, \cite{CLS19, song2021oblivious} etc.) rely on a stochastic version of the short step central path method, in which the neighborhood of the central path is slightly wider than that of the traditional short-step IPMs. This is still considered to be quite restrictive, as it needs all the pairwise products $x_i s_i$ to be close to to the duality measure $\mu$ (more precisely, it needs $0.9\mu\le x_i s_i\le 1.1\mu$ for all $i$), which is not a very practical idea in general.}
    On the other hand, our algorithm is based on long-step path following IPMs which explore a much larger neighborhood around the central path and offer significant flexibility in each iteration. As a result, the iterates can take much longer steps towards optimality. Long-step IPMs are known to be more efficient in practice compared to short-step IPMs,  despite the fact that they exhibit worst-case iteration complexity $\Ocal(n)$.  

    \item Second, the theoretical benefits of the aforementioned TCS works rely on techniques such as \emph{projection maintenance} or \emph{inverse maintenance}, where one needs to preserve the orthogonal projection matrix $\Db \Ab^{\top}(\Ab\Db^2\Ab^{\top})^{-1} \Ab\Db\in \mathbb{R}^{n \times n}$ or $(\Ab\Db^2\Ab^{\top})^{-1}\in \mathbb{R}^{n \times n}$ in order to solve the linear system at each iteration of the IPM. The idea of maintaining such projection matrices depends on the assumption that the matrix $\Db$ does not change much from one iteration to the next. Thus, one only needs to compute the matrix $(\Ab\Db^2\Ab^{\top})^{-1}$ a few times, which is also known as lazy updating. In addition, if $\Db$ only changes in a few of its diagonal entries, one can further use the idea of low-rank updating to maintain the above projection matrix via the Sherman–Morrison–Woodbury (SMW) formula. However, due to large constant factors involved in the lazy updates framework and the numerical instability of the SMW formula, the notion of \emph{projection maintenance} is generally inefficient in practice. In contrast, our methods do not depend on any of the aforementioned techniques. Instead, we use randomized preconditioners combined with iterative solvers to approximately solve the linear system. We also propose a computationally efficient way to correct for the error caused by the solver. As a result, our method is much more relevant in practice and, when $n\gg m$, the per iteration cost of our algorithm is $\widetilde{\Ocal}(\nnz{\Ab}+m^3)$
    
    \item Third, all the aforementioned papers involve fast matrix multiplication, which is a theoretically relevant yet practically inefficient tool. To the best of our knowledge, there are no in- or out-of-core implementations of such algorithms that work better than traditional matrix multiplication approaches and it seems unlikely that such techniques will become practically relevant in the near future. Our methods leverage standard matrix multiplication routines and in Section~\ref{sec:exp} we show how an implementation of our approach offers advantages over existing, practically relevant approaches.

\end{itemize}


{\color{black} Another line of research in the Theoretical Computer Science literature that is very close to our work is~\citep{daitch2008faster}, who presented an IPM that uses an approximate solver in each iteration. However, their accuracy guarantee is in terms of the final objective value, which is different from ours.
More importantly,~\citep{daitch2008faster} focuses on \textit{short-step} IPMs, whereas our approach is a \emph{long-step} algorithm that works for both feasible and infeasible starting points. Finally, the approximate solver proposed by~\citep{daitch2008faster} works only for the special case of input matrices that correspond to graph Laplacians, following the lines of~\citep{spielman2004nearly,spielman2014nearly}. In this context, we note that there are also other relevant works including~\citep{madry2013navigating,madry2016computing,cohen2017negative} that used IPM-based algorithms with approximate solvers for various combinatorial optimization problems on graphs. However, similar to \citep{spielman2004nearly}, the aforementioned papers also focus on short-step IPMs and the linear systems associated with them are either Laplacian or symmetric diagonally dominant (SDD).}

Another relevant line of research is the work by~\citet{CMTH16}, which proposed solving eqn.~\eqref{eq:normal}  using preconditioned Krylov subspace methods, including variants of \emph{generalized minimum residual} (GMRES) or CG methods. Indeed, \citet{CMTH16} conducted extensive numerical experiments on LP problems taken from standard benchmark libraries, but did not provide any theoretical guarantees.
\textcolor{black}{Beyond experiments, the methods of \citep{CMTH16} primarily differ from ours in the way they used preconditioning. Instead of applying the preconditioner explicitly,  their empirical evaluations rely on an implicit preconditioning technique called  \emph{(stationary) inner-iterations preconditioning}
\citep{morikuni2013inner,morikuni2015convergence}. As the name suggests, the key idea is to use another iterative method as a preconditioner within the linear system that needs to be solved at each iteration of the IPM. From a theoretical perspective, it is not clear how big or small the resulting condition number of their preconditioned system could be; moreover, there are two hyperparameters associated with their preconditioner which need to be tuned optimally and there is no theoretical guideline on how that can be done.}

From a matrix-sketching perspective, our work was also partially motivated by~\citep{CYD18}, which presented an iterative, sketching-based algorithm to solve under-constrained ridge regression problems, but did not address how to make use of such approaches in an IPM-based framework, as we do here. We refer the reader to the surveys~\citep{Woodruff14, DM2018,Mahoney11,Drineas2016,martinsson2020randomized} for more background on Randomized Linear Algebra and regression solvers. In the context of deep neural networks (DNNs), \citep{van2021training} recently presented an iterative method to speed up the training of overparametrized DNNs by a similar type of randomized preconditioner as ours; but the algorithm of \citep{van2021training} neither exploited the sparsity in the data, nor they had to correct for the error caused by the inexact solver, as we do here. In another work,~\citep{avron2017faster} also proposed a similar sketching-based preconditioning technique. However, their efforts broadly revolved around speeding up and scaling \emph{kernel ridge regression}.~
\textcolor{black}{\citep{PW17}
proposed the so-called \emph{Newton sketch} to construct an approximate Hessian matrix  for more general convex problems, of which LP is a special case. Nevertheless, based on their local convergence guarantee for the sketched Newton updates, their paper only derived the underlying iteration complexity of the IPM (see, for example, Theorem 4.3 of~\citep{PW17}). It is not clear how to use their approach to bound the number of inner iterations and, as a result, deriving the per iteration cost of their algorithm is not straightforward. Moreover, their convergence guarantees for the IPM is with respect to the objective value and not in terms of the duality measure. Finally, the Newton-sketch of~\citep{PW17} does not exploit the sparsity of the data, whereas the running time of our algorithm depends on the sparsity of $\Ab$.} 
We also note that~\citep{VPL18} proposed a probabilistic algorithm to solve LPs approximately using a random projection reduced feature space. A possible drawback of this work is that the approximate solution is infeasible with respect to the original region. 

\textcolor{black}{On the empirical side, there are prior implementations of $\Ocaltil(\nnz{\Ab} + \mathrm{poly}(m))$ solvers for speeding various ML applications including ordinary least square regression~\citep{clarkson2017low,cormode2019iterative}, general $\ell_p$-regression~\citep{yang2017weighted}, ridge regression~\citep{chen2015fast,CYD18}, Fisher linear discriminant analysis~\citep{ye2017fast,CYD19}, more general Newton updates~\citep{dahiya2018empirical} and many more. However, to the best of our knowledge, there are no such implementations in the context of general linear programming problems, perhaps with the exception of $\ell_1$-regression. As discussed before, prior work on short-step IPM for LP that came up with per iteration cost $\Ocaltil(\nnz{\Ab} + \mathrm{poly}(m))$ was due to \citep{LS14,LS15}, but there is no empirical evaluation because of its heavy reliance on various theoretical tools such as inverse maintenance, fast matrix multiplication etc.}

Finally, in addition to IPMs, LPs can be solved using the Simplex method. In commercial LP packages like Gurobi, both methods are often used in conjunction. For example, multiple solvers are run on multiple threads simultaneously and the one that is finishes first is chosen. Alternatively, an IPM is used initially to get close to the optimal solution and then the simplex algorithm is used to improve the solution~\citep{bixby1992very,glavelis2018improving}. Thus, developing efficient IPMs is vital for solving LPs and provides a crucial building block in commerical packages like Gurobi.

%% file: 2_Background.tex
\section{Notation and Background}\label{sxn:background}

$\mathbf{A}, \mathbf{B}, \ldots$ denote
matrices and $\mathbf{a}, \mathbf{b}, \ldots$ denote vectors. For vector $\ab$, $\|\mathbf{a}\|_{2}$ denotes its Euclidean norm; for a matrix $\mathbf{A},\|\mathbf{A}\|_{2}$ denotes
its spectral norm and $\|\Ab\|_F$ denotes its Frobenius norm. We use $\zero$ to denote a null vector or null matrix, dependent upon context, and $\one$ to denote the all-ones vector.
For any matrix $\Xb\in\RR{\dimone}{\dimtwo}$ with $\dimone\leq\dimtwo$ of rank $\dimone$ a thin Singular Value Decomposition (SVD) is a product $\Ub\Sigmab\Vb^\ts$ , with $\mathbf{U} \in\mathbb{R}^{\dimone\times\dimone}$ (the matrix of the left singular vectors), $\mathbf{V} \in$ $\mathbb{R}^{\dimtwo \times \dimone}($ the matrix of the top-$\dimone$ right singular vectors), and $\Sigmab \in$
$\mathbb{R}^{\dimone \times \dimone}$ a diagonal matrix whose entries are equal to the singular values of $\Xb$. We use $\sigma_{i}(\cdot)$ to denote the $i$-th singular value of the matrix in parentheses. For any two vectors $\ab=(a_1,\dots,a_\ell)^\ts$ and $\bb=(b_1,\dots,b_\ell)^\ts$, we denote $\ab\circ\bb=(a_1b_1,\dots,a_\ell b_\ell)^\ts$. For any two symmetric positive semidefinite (resp. positive definite) matrices $\Ab_1$ and $\Ab_2$ of appropriate dimensions, $\Ab_1\preccurlyeq\Ab_2$ ($\Ab_1\prec\Ab_2$) denotes that $\Ab_2-\Ab_1$ is positive semidefinite (resp. positive definite). For any vector $\ab\in\R{n}$ its $\ell_\infty$ norm is defined as $\|\ab\|_\infty=\max_i\abs{a_i}$. We extensively use the following standard inequality to prove several results in the paper:
\begin{flalign}\label{eq:normineq}
\abs{\frac{\ab^\ts\one_n}{n}}\le\|\ab\|_\infty\le\|\ab\|_2.
\end{flalign}



We now briefly discuss a result on matrix sketching~\citep{Cohen2016,cohen2016nearly} that is particularly useful in our theoretical analyses. In our parlance,~\citet{Cohen2016} proved that,
for any matrix $\Zb\in\RR{m}{n}$, there exists a sketching matrix $\Wb\in\RR{n}{w}$ such that
\begin{flalign}\label{eqn:pdprec}
	\nbr{\Zb \Wb \Wb^\ts \Zb^\ts - \Zb \Zb^\ts}_2\le \frac{\zeta}{4}\Big(\nbr{\Zb}_2^2+\frac{\|\Zb\|_F^2}{r}\Big)
\end{flalign}
holds with probability at least $1-\delta$ for any $r\ge 1$. Here $\zeta \in [0,1]$ is a (constant) accuracy parameter. Ignoring constant terms, $w=\Ocal(r\log(\nicefrac{r}{\delta}))$; $\Wb$ has $\Ocal(\log(r/\delta))$ non-zero entries per row; and the product $\Zb\Wb$ can be computed in time
$\Ocal(\log(r/\delta)\cdot\nnz{\Zb})$.

%% file: 3_PCG_solver.tex
\section{Preconditioned Iterative Solver}
\label{sxn:PCG}
%
In this section, we discuss the computation of the preconditioner $\Qb$ (and its inverse), followed by a discussion on how such a preconditioner can be used to satisfy eqns.~\eqref{eq:pdcond1} and~\eqref{eq:pdcond2}.
\begin{algorithm}[H]
	\caption{Solving eqn.~\eqref{eq:precond_alt} via CG or Chebyshev iteration}\label{algo:PCG}
	\hspace*{\algorithmicindent}\textbf{Input:} $\Ab\Db\in\RR{\dimone}{\dimtwo}$, $\pb\in\R{\dimone}$, sketching matrix $\Wb \in \mathbb{R}^{n \times w}$,  iteration count $t$;
	\begin{algorithmic}[1]
		\State Compute $\Ab\Db\Wb$ and its SVD: let $\Ub_{\Qb} \in \mathbb{R}^{m \times m}$ be the matrix of its left singular vectors and let $\Sigmab_{\Qb}^{\nicefrac{1}{2}} \in \mathbb{R}^{m \times m}$ be the matrix of its singular values;
		\State Compute $\Qb^{-\nicefrac{1}{2}} = \Ub_{\Qb} \Sigmab_{\Qb}^{-\nicefrac{1}{2}}\Ub_{\Qb}^\ts$;
		\State Initialize $\tilde{\zb}^{0} \gets \zero_\dimone $ and run standard CG or Chebyshev iteration on the preconditioned system of eqn.~\eqref{eq:precond_alt} for $t$ iterations;
	\end{algorithmic}
	\hspace*{\algorithmicindent}\textbf{Output:} return $\tilde{\zb}^t$;
\end{algorithm}
\noindent Algorithm~\ref{algo:PCG} takes as input the sketching matrix $\Wb \in \mathbb{R}^{n \times w}$, which we construct as discussed in Section~\ref{sxn:background}. Our preconditioner $\Qb$ is equal to
\begin{flalign}\label{eqn:pdprecond}
	\Qb=\Ab\Db\Wb\Wb^\ts\Db\Ab^\ts.
\end{flalign}
Notice that we only need to compute $\Qb^{\nicefrac{-1}{2}}$ in order to use it to solve eqn.~(\ref{eq:precond_alt}). Towards that end, we first compute the sketched matrix $\Ab\Db\Wb \in \mathbb{R}^{\dimone \times w}$. Then, we compute the SVD of the matrix $\Ab\Db\Wb$: let $\Ub_{\Qb}$ be the matrix of its left singular vectors and let $\Sigmab_{\Qb}^{\nicefrac{1}{2}}$ be the matrix of its singular values. Notice that the left (and right) singular vectors of $\Qb^{\nicefrac{-1}{2}}$ are equal to $\Ub_{\Qb}$ and its singular values are equal to $\Sigmab_{\Qb}^{-\nicefrac{1}{2}}$. Therefore, $\Qb^{\nicefrac{-1}{2}} = \Ub_{\Qb} \Sigmab_{\Qb}^{-\nicefrac{1}{2}}\Ub_{\Qb}^\ts$.

Let $\Ab\Db = \Ub\Sigmab\Vb^\ts$ be the thin SVD representation of $\Ab\Db$. We apply the results of~\citep{Cohen2016} (see Section~\ref{sxn:background}) to the matrix $\Zb = \Vb^\ts \in\mathbb{R}^{m \times n}$ with $r=m$ to get that, with probability at least $1-\delta$,
\begin{flalign}
	\nbr{\Vb^\ts \Wb \Wb^\ts \Vb - \Ib_{\dimone}}_2\le \frac{\zeta}{4}\Big(\nbr{\Vb}_2^2+\frac{\|\Vb\|_F^2}{m}\Big) \le \frac{\zeta}{2}.\label{eq:cnd1}
\end{flalign}
In the above we used $\nbr{\Vb}_2=1$ and $\nbr{\Vb}_F^2=m$. The running time needed to compute the sketch $\Ab\Db\Wb$ is equal to (ignoring constant factors)
%
$\Ocal(\nnz{\Ab}\cdot \log(m/\delta))$.
%
Note that $\nnz{\Ab\Db}=\nnz{\Ab}$. The cost of computing the SVD of $\Ab\Db\Wb$ (and therefore $\Qb^{\nicefrac{-1}{2}}$) is $\Ocal(m^3\log(m/\delta))$. Overall, computing $\Qb^{-\nicefrac{1}{2}}$ can be done in time
\begin{flalign}\label{eqn:svdQ}
	\Ocal(\nnz{\Ab}\cdot \log(m/\delta)+m^3\log(m/\delta)).
\end{flalign}
Given these results, we now discuss how to satisfy eqns.~\eqref{eq:pdcond1} and \eqref{eq:pdcond2} using the sketching matrix $\Wb$. We start with the following bound, which is relatively straightforward given prior RLA work.
%
\begin{lemma}\label{lem:cond3}
	If the sketching matrix $\Wb$ satisfies eqn.~\eqref{eq:cnd1}, then, for all $i=1\ldots m$,
	\begin{flalign*}
		(1+\zeta/2)^{-1}\le\sigma_i^2(\Qb^{-\nicefrac{1}{2}}\Ab\Db)\le (1-\zeta/2)^{-1}.
	\end{flalign*}
\end{lemma}
\begin{proof}
	Consider the condition of eqn.~\eqref{eq:cnd1}:
	\begin{flalign}
	&~\|\Vb^\ts\Wb\Wb^\ts\Vb-\Ib_\dimone\|_2\le\frac{\zeta}{2}~\Leftrightarrow~ -\frac{\zeta}{2}\,\Ib_{\dimone}\preccurlyeq\Vb^\ts\Wb\Wb^\ts\Vb-\Ib_\dimone\preccurlyeq\frac{\zeta}{2}\,\Ib_{\dimone}\label{eq:full}\\
	\Leftrightarrow~&-\frac{\zeta}{2}\,\Ab\Db^2\Ab^\ts\preccurlyeq\Ab\Db\Wb\Wb^\ts\Db\Ab^\ts-\Ab\Db^2\Ab^\ts\preccurlyeq\frac{\zeta}{2}\,\Ab\Db^2\Ab^\ts\label{eq:rich_3}\\
	\Leftrightarrow~&\left(1-\frac{\zeta}{2}\right)\,\Ab\Db^2\Ab^\ts\preccurlyeq\underbrace{\Ab\Db\Wb\Wb^\ts\Db\Ab^\ts}_{\Qb}\preccurlyeq\left(1+\frac{\zeta}{2}\right)\,\Ab\Db^2\Ab^\ts\,.\label{eq:rich_4}
	\end{flalign}
	We obtain eqn.~\eqref{eq:rich_3} by pre- and post-multiplying the previous inequality by $\Ub\Sigmab$ and $\Sigmab\Ub^\ts$ respectively and using the facts that $\Ab\Db=\Ub\Sigmab\Vb^\ts$ and $\Ab\Db^2\Ab^\ts=\Ub\Sigmab^2\Ub^\ts$. Also, from eqn.~\eqref{eq:full}, note that all the eigenvalues of $\Vb^\ts\Wb\Wb^\ts\Vb$ lie between $(1-\frac{\zeta}{2})$ and $(1+\frac{\zeta}{2})$ and thus $\rank(\Vb^\ts\Wb)=m$. Therefore, $\rank(\Ab\Db\Wb)=\rank(\Ub\Sigmab\Vb^\ts\Wb)=m$, as $\Ub\Sigmab$ is non-singular and we know that the rank of a matrix remains unaltered by pre- or post-multiplying it by a non-singular matrix. So, we have $\rank(\Qb)=m$; in words $\Qb$ has full rank. Therefore, all the diagonal entries of $\Sigmab_\Qb$ are positive and $\Qb^{-\nicefrac{1}{2}}\Qb\Qb^{-\nicefrac{1}{2}}=\Ib_m$\,.

Using the above arguments, pre- and post- multiplying eqn.~\eqref{eq:rich_4} by $\Qb^{-1/2}$, we get
	\begin{flalign}
	&~\left(1-\frac{\zeta}{2}\right)\,\Qb^{-1/2}\Ab\Db^2\Ab^\ts\Qb^{-1/2}\preccurlyeq\Ib_{m}\preccurlyeq\left(1+\frac{\zeta}{2}\right)\,\Qb^{-1/2}\Ab\Db^2\Ab^\ts\Qb^{-1/2}\nonumber\\
	\Rightarrow&~\left(1+\frac{\zeta}{2}\right)^{-1}\Ib_{m}\preccurlyeq\Qb^{-1/2}\Ab\Db^2\Ab^\ts\Qb^{-1/2}\preccurlyeq\left(1-\frac{\zeta}{2}\right)^{-1}\Ib_{\dimone}\,.\label{eq:normbound}
	\end{flalign}
Eqn.~\eqref{eq:normbound} implies that all the eigenvalues of $\Qb^{-1/2}\Ab\Db^2\Ab^\ts\Qb^{-1/2}$ are bounded between $\left(1+\frac{\zeta}{2}\right)^{-1}$ and $\left(1-\frac{\zeta}{2}\right)^{-1}$, which concludes the proof of the lemma.
\end{proof}
\noindent The above lemma directly implies eqn.~\eqref{eq:pdcond1}. We now proceed to show that the above construction for $\Qb^{\nicefrac{-1}{2}}$, when combined with the conjugate gradient solver or Chebyshev iteration to solve eqn.~\eqref{eq:precond_alt}, indeed satisfies eqn.~\eqref{eq:pdcond2}. 

\subsection{Conjugate Gradient Solver}\label{sxn:cg}
\textcolor{black}{As already mentioned in Section~\ref{sxn:contrib}, we derive eqn.~\eqref{eq:pdcond2} using the monotonicity property of CG resudual norms. 
}
It is known that even if the energy norm of the error of the approximate solution decreases monotonically, the norms of the CG residuals may oscillate. Interestingly, we can combine a result on the residuals of CG from~\citep{bouyouli2009new} with Lemma~\ref{lem:cond3} to prove that in our setting the norms of the CG residuals also decrease monotonically\footnote{See Chapter 9 of~\citep{Luenberger15} for a detailed overview of CG.}. \textcolor{black}{We do note that in prior work most of the convergence guarantees for CG focus on the error of the approximate solution, from which, directly deriving eqn.~\eqref{eq:pdcond2} induces a constant factor in terms of the condition number of the preconditioned matrix. Here, we are able to avoid that constant factor by deriving and using the aforementioned monotonicity of the CG residuals.}

Let $\tilde{\fb}^{(j)}$ be the residual at the $j$-th iteration of the CG algorithm:
$$\tilde{\fb}^{(j)}=\Qb^{-\nicefrac{1}{2}}\Ab\Db^2\Ab^\ts\Qb^{-\nicefrac{1}{2}}\tilde{\zb}^j-\Qb^{-\nicefrac{1}{2}}\pb.$$
Recall from Algorithm~\ref{algo:PCG} that~$\tilde{\zb}^0=\zero$ and thus $\tilde{\fb}^{(0)}=-\Qb^{-\nicefrac{1}{2}}\pb$.
In our parlance, Theorem 8 of~\citep{bouyouli2009new} proved the following bound.
\begin{lemma}[Theorem 8 of \citep{bouyouli2009new}]\label{lem:prev1}%
	Let $\tilde{\fb}^{(j-1)}$ and $\tilde{\fb}^{(j)}$ be the residuals obtained by the CG solver at steps $j-1$ and $j$. Then,%
	\begin{flalign*}
	\|\tilde{\fb}^{(j)}\|_2\le~\frac{\kappa^2(\Qb^{-\nicefrac{1}{2}}\Ab\Db)-1}{2}\|\tilde{\fb}^{(j-1)}\|_2\,,
	\end{flalign*}
	where $\kappa(\Qb^{-\nicefrac{1}{2}}\Ab\Db)$ is the condition number of $\Qb^{-\nicefrac{1}{2}}\Ab\Db$.
\end{lemma}
\noindent\textbf{Satisfying eqn.~\eqref{eq:pdcond2}.} From Lemma~\ref{lem:cond3}, we get
\begin{flalign}\label{eq:condbd}
\kappa^2(\Qb^{-\nicefrac{1}{2}}\Ab\Db)=\frac{\sigma_{\max}^2(\Qb^{-\nicefrac{1}{2}}\Ab\Db)}{\sigma_{\min}^2(\Qb^{-\nicefrac{1}{2}}\Ab\Db)}\le\frac{1+\zeta/2}{1-\zeta/2}.
\end{flalign}
Combining eqn.~\eqref{eq:condbd} with Lemma~\ref{lem:prev1},
\begin{flalign}
\|\tilde{\fb}^{(j)}\|_2\le~\frac{\frac{1+\zeta/2}{1-\zeta/2}-1}{2}\|\tilde{\fb}^{(j-1)}\|_2
=~\frac{\zeta}{2-\zeta}\|\tilde{\fb}^{(j-1)}\|_2
\le~\zeta \|\tilde{\fb}^{(j-1)}\|_2\label{eq:rec}\,,
\end{flalign}
where the last inequality follows from $\zeta\le1$. Applying eqn.~\eqref{eq:rec} recursively, we get
\begin{flalign}
\|\tilde{\fb}^{(t)}\|_2\le\zeta\|\tilde{\fb}^{(t-1)}\|_2\le\dots\le\zeta^t\|\tilde{\fb}^{(0)}\|_2=\zeta^t\|\Qb^{-\nicefrac{1}{2}}\pb\|_2\nonumber\,,
\end{flalign}
which proves the condition of eqn.~\eqref{eq:pdcond2}.

We remark that one can consider using MINRES~\citep{paige1975solution} instead of CG. Our results hinges on bounding the two-norm of the residual. MINRES finds, at each iteration, the optimal vector with respect the two-norm of the residual inside the same Krylov subspace of CG for the corresponding iteration. Thus, the bound we prove for CG applies to MINRES as well.

\subsection{Chebyshev Iteration}
Now, we show that we could potentially replace CG with Chebyshev iteration (see Algorithm~1 of \citep{gutknecht2002chebyshev}) in the step~(3) of Algorithm~\ref{algo:PCG}. As already discussed in Section~\ref{sxn:contrib}, the only requirement Chebyshev iteration needs is to have an upper bound and a lower bound for the singular values of $\Qb^{-\nicefrac{1}{2}}\Ab\Db^2\Ab^\ts\Qb^{-\nicefrac{1}{2}}$,
which we already have in the form of Lemma~\ref{lem:cond3}. Therefore, all we need to show is that the sketching matrix $\Wb$ satisfies eqn.~\eqref{eq:pdcond2} using Chebyshev iteration. For this, we state the following result from~\citep{Gutknecht08}, which is instrumental in proving eqn.~\eqref{eq:pdcond2}.
\begin{lemma}[Theorem~1.6.2 of \citep{Gutknecht08}]
The residual norm reduction of the Chebyshev iteration, when applied to an symmetric positive definite (SPD) system whose condition number is upper bounded by $\mathcal{U}$, is bounded according to
\begin{flalign}
\frac{\|\tilde{\fb}^{(t)}\|_2}{\|\tilde{\fb}^{(0)}\|_2} \leq 2\left[\left(\frac{\sqrt{\mathcal{U}}+1}{\sqrt{\mathcal{U}}-1}\right)^{t}+\left(\frac{\sqrt{\mathcal{U}}-1}{\sqrt{\mathcal{U}}+1}\right)^{t}\right]^{-1}\label{eq:cs}
\end{flalign}
\end{lemma}

\noindent\textbf{Satisfying eqn.~\eqref{eq:pdcond2}.} From Lemma~\ref{lem:cond3}, we directly have  $\mathcal{U}=\frac{2+\zeta}{2-\zeta}$ 
. Note that for $t=0$ and starting from $\tilde{\zb}^0=\zero$ (\emph{i.e., $\tilde{\fb}^{(0)}=\Qb^{-\nicefrac{1}{2}}\pb$}), we directly get eqn.~\eqref{eq:pdcond2} from eqn.~\eqref{eq:cs}.  Therefore, we only show that eqn.~\eqref{eq:pdcond2} is satisfied for $t\ge 1$.
%
We already have $\tilde{\fb}^{(0)}=\Qb^{-\nicefrac{1}{2}}\pb$ and letting $a=\left(\frac{\sqrt{\mathcal{U}}-1}{\sqrt{\mathcal{U}}+1}\right)^{t}$, we rewrite eqn.~\eqref{eq:cs} as follows
\begin{flalign}
\|\tilde{\fb}^{(t)}\|_2\le\frac{2}{a+\frac{1}{a}}\|\Qb^{-\nicefrac{1}{2}}\pb\|_2=\frac{2\,a}{(a^2+1)}\|\Qb^{-\nicefrac{1}{2}}\pb\|_2\le 2\,a\,\|\Qb^{-\nicefrac{1}{2}}\pb\|_2\,,\label{eq:cs2}
\end{flalign}
where the last inequality in eqn.~\eqref{eq:cs2} holds as $a^2>0$. Now, we'll work on the bound in eqn.~\eqref{eq:cs2}. Putting back $a=\left(\frac{\sqrt{\mathcal{U}}-1}{\sqrt{\mathcal{U}}+1}\right)^{t}$ and $\mathcal{U}=\frac{2+\zeta}{2-\zeta}$, we rewrite eqn.~\eqref{eq:cs2} as
\begin{flalign}
\|\tilde{\fb}^{(t)}\|_2\le&~ 2\left(\frac{\sqrt{\mathcal{U}}-1}{\sqrt{\mathcal{U}}+1}\right)^{t}\,\|\Qb^{-\nicefrac{1}{2}}\pb\|_2=2\left(\frac{\sqrt{\frac{2+\zeta}{2-\zeta}}-1}{\sqrt{\frac{2+\zeta}{2-\zeta}}+1}\right)^{t}\,\|\Qb^{-\nicefrac{1}{2}}\pb\|_2\nonumber\\
=~& 2\left(\frac{\sqrt{2+\zeta}-\sqrt{2-\zeta}}{\sqrt{2+\zeta}+\sqrt{2-\zeta}}\right)^{t}\,\|\Qb^{-\nicefrac{1}{2}}\pb\|_2=2\left(\frac{2\zeta}{\big(\sqrt{2+\zeta}+\sqrt{2-\zeta}\big)^2}\right)^{t}\,\|\Qb^{-\nicefrac{1}{2}}\pb\|_2\nonumber\\
=~& 2\left(\frac{2\zeta}{4\,\big(1+\sqrt{1-(\nicefrac{\zeta}{2})^2}\big)}\right)^{t}\,\|\Qb^{-\nicefrac{1}{2}}\pb\|_2=\frac{\zeta^t}{2^{t-1}\big(1+\sqrt{1-(\nicefrac{\zeta}{2})^2}\big)^t}\,\|\Qb^{-\nicefrac{1}{2}}\pb\|_2\nonumber\\
\le~& \zeta^t\,\|\Qb^{-\nicefrac{1}{2}}\pb\|_2\nonumber\,,
\end{flalign}
where the last inequality holds as $t\ge 1$ and the denominator is greater than unity. This establishes eqn.~\eqref{eq:pdcond2}.

%% file: 4_Algorithm.tex
\section{The Feasible IPM algorithm}\label{sxn:IIPM}

In order to avoid spurious solutions, primal-dual path-following IPMs bias the search direction towards the \emph{central path} and restrict the iterates to a neighborhood of the central path. This search is controlled by the \emph{centering parameter} $\sigma\in[0,1]$.
%
At each iteration, given the current feasible solution $(\xb^{k},\yb^{k},\sbb^{k})$, a standard feasible IPM obtains the search direction $(\Delta\xb^k,\Delta\yb^k,\Delta\sbb^k)$ by solving the following system of linear equations:
\begin{subequations}\label{eq:system}
	\begin{flalign}
		\Ab\Db^2\Ab^\ts\Delta\yb^k=~&\pb^k\,,\label{eq:normal1}\\
		\Delta\sbb^k=~&-\Ab^\ts\Delta\yb^k\,,\label{eq:dels}\\
		\Delta\xb^k=~&-\xb^k+\sigma\mu_k\Sb^{-1}\one_\dimtwo-\Db^2\Delta\sbb^k.\label{eq:delx}
	\end{flalign}
\end{subequations}
Here $\Db$ and $\Sb$ are computed given the current iterate $(\xb^{k}$ and $\sbb^{k})$; we skip the indices on $\Db$ and $\Sb$ for notational simplicity. 
After solving the above system, the feasible IPM Algorithm~\ref{algo:iipm} proceeds by computing a step-size $\alphabar$ to return:
\begin{flalign}\label{eqn:update}
	(\xb^{k+1},\yb^{k+1},\sbb^{k+1}) = (\xb^{k},\yb^{k},\sbb^{k}) + \alphabar (\Delta \xb^k,\Delta \yb^k,\Delta \sbb^k).
\end{flalign}
%
In the above linear system in eqn.~\eqref{eq:system}, we also use \textit{duality measure} $\mu_k=\nicefrac{{\xb^k}^\ts\sbb^k}{n}$ and the vector
\begin{flalign}\label{eqn:pdef}
	\pb^k&=-\sigma\mu_k\Ab\Sb^{-1}\one_\dimtwo+\Ab\xb^k.
\end{flalign}
Given $\Delta\yb^k$ from eqn.~(\ref{eq:normal1}), $\Delta\sbb^k$ and $\Delta\xb^k$ are easy to compute from eqns.~\eqref{eq:dels} and \eqref{eq:delx}, as they only involve matrix-vector products.  However, since we  use Algorithm~\ref{algo:PCG} to solve eqn.~\eqref{eq:normal1} approximately using the sketching-based preconditioned solver, the iterates $(\xb^{k},\yb^{k},\sbb^{k})$ do not satisfy the primal and dual constraints exactly.

For notational simplicity, we now drop the dependency of vectors and scalars on the iteration counter $k$. Let $\hat{\Delta \yb}=\Qb^{\nicefrac{-1}{2}}\tilde{\zb}^t$ be the approximate solution to eqn.~(\ref{eq:normal1}). In order to account for the loss of accuracy due to the approximate solver, we compute $\hat{\Delta\xb}$ as follows:
\begin{flalign}
	\hat{\Delta\xb}=~-\xb+\sigma\mu\Sb^{-1}\one_\dimtwo-\Db^2\hat{\Delta\sbb}-\Sb^{-1}\vb\label{eq:delxhat}.
\end{flalign}
Here $\vb\in\R{n}$ is a perturbation vector that needs to exactly satisfy the following invariant at each iteration of the feasible IPM:
\begin{flalign}
	\Ab\Sb^{-1}\vb=\Ab\Db^2\Ab^\ts\hat{\Delta\yb}-\pb\,\label{eq:addl}.
\end{flalign}
We note that the computation of $\hat{ \Delta \sbb}$ is still done using, essentially, eqn.~\eqref{eq:dels}, namely
\begin{flalign}
\hat{\Delta\sbb}^k=~&-\Ab^\ts\hat{\Delta\yb}^k.\label{eq:delshat}
\end{flalign}
%
At each iteration of the IPM, if $\vb$ satisfies eqn.~\eqref{eq:addl}, then it can be shown that the primal and dual feasibility constraints are satisfied exactly.

\vspace{0.02in}\noindent\textbf{Construction of $\vb$.} There are many choices for $\vb$ satisfying eqn.~\eqref{eq:addl}. \textcolor{black}{Intuitively, we would expect the approximation error due to the solver to be reasonably small. Therefore,}
to prove convergence, it is desirable for $\vb$ to have a small norm and hence a natural choice is
$$\vb=(\Ab\Sb^{-1})^{\dagger}(\Ab\Db^2\Ab^\ts\hat{\Delta\yb}-\pb)\,.$$
\textcolor{black}{The aforementioned choice of $\vb$ has a clear geometric interpretation: it not only ensures that $\mathbf{A} \mathbf{S}^{-1} \mathbf{v}$ is a Euclidean projection of the infeasible solution onto the column space of $\Ab\Db$, but it is also the minimum norm least squares solution and satisfies the invariant in eqn.~\eqref{eq:addl} exactly. However, computing $\vb$ such a way is expensive, as it involves the evaluation of the pseudoinverse of  $\mathbf{A} \mathbf{S}^{-1}$, which is expensive, taking time $\Ocal(\dimone^2\dimtwo)$.}
Instead, we propose
to construct $\vb$ using the sketching matrix $\Wb$ of Section~\ref{sxn:background}. More precisely, we construct the perturbation vector
\begin{flalign}\label{eq:compv}
	\vb=(\Xb\Sb)^{\nicefrac{1}{2}}\Wb(\Ab\Db\Wb)^{\dagger}(\Ab\Db^2\Ab^\ts\hat{\Delta\yb}-\pb).
\end{flalign}
\textcolor{black}{Similar to the minimum-norm solution mentioned above, our sketching based solution in eqn.~\eqref{eq:compv} 
also guarantees that $\mathbf{A} \mathbf{S}^{-1} \mathbf{v}$ is a projection of the ``infeasibility'' vector $\Ab\Db^2\Ab^\ts\hat{\Delta\yb}-\pb$ onto the column space of $\Ab\Db\Wb$ (which is identical to the column space of $\Ab\Db$) and satisfies eqn.~\eqref{eq:addl} exactly (see Lemma~\ref{lem:fullrankR} below). The computation of our proposed $\vb$ is dominated by the cost of computing $(\Ab\Db\Wb)^{\dagger}$, which can be done much more efficiently as discussed in Section~\ref{sxn:PCG}. Actually, we do not even need to compute it while evaluating $\vb$, since we have already computed it during the construction of $\Qb^{-1/2}$ in Algorithm~\ref{algo:PCG}. 
Finally, in Lemma~\ref{lem:conouter}, we showed that using our choice of $\vb$, $\|\vb\|_2$  remains small enough (with a few iterations of the iterative solver), which essentially leads to the convergence of the IPM.}
%

%
\begin{lemma}\label{lem:fullrankR}
	Let $\Wb\in\RR{\dimtwo}{w}$ be the sketching matrix of Section~\ref{sxn:background} and $\vb$ be the perturbation vector of eqn.~(\ref{eq:compv}). Then, with probability at least $1-\delta$, $\rank(\Ab\Db\Wb)=\dimone$ and
	$\vb$ satisfies eqn.~\eqref{eq:addl}.
\end{lemma}
\begin{proof}
	Let $\Ab\Db=\Ub\Sigmab\Vb^\ts$ be the thin SVD representation of $\Ab\Db$.
	We use the exact same $\Wb$ as discussed in Section~\ref{sxn:PCG}.  Therefore, eqn.~\eqref{eq:cnd1} holds with probability $1-\delta$ and it directly follows from the proof of Lemma~\ref{lem:cond3} that $\rank(\Ab\Db\Wb)=\dimone$.
Recall that $\Ab\Db\Wb$ has full \emph{row-rank} and thus $\Ab\Db\Wb\,(\Ab\Db\Wb)^\dagger=\Ib_\dimone$. Therefore, taking $\vb=(\Xb\Sb)^{\nicefrac{1}{2}}\Wb(\Ab\Db\Wb)^{\dagger}(\Ab\Db^2\Ab^\ts\hat{\Delta\yb}-\pb)$, we get
	\begin{flalign*}
	\Ab\Sb^{-1}\,\vb=&~\Ab\Sb^{-1}(\Xb\Sb)^{\nicefrac{1}{2}}\Wb(\Ab\Db\Wb)^{\dagger}(\Ab\Db^2\Ab^\ts\hat{\Delta\yb}-\pb)\nonumber\\
	=&~\Ab\Db\Wb(\Ab\Db\Wb)^{\dagger}(\Ab\Db^2\Ab^\ts\hat{\Delta\yb}-\pb)\nonumber\\
	=&~\Ab\Db^2\Ab^\ts\hat{\Delta\yb}-\pb\,,
	\end{flalign*}
	where the second equality follows from $\Db = \Xb^{1/2}\Sb^{-1/2}$.
\end{proof}
 \noindent We emphasize here that we use the same exact sketching matrix $\Wb \in \mathbb{R}^{n \times w}$ to form the preconditioner used in the iterative solver of Section~\ref{sxn:PCG} \textit{as well as} the vector $\vb$ in eqn.~(\ref{eq:compv}). This allows us to sketch $\Ab \Db$  only once, thus saving time in practice. Next, we present a bound for the two-norm of the perturbation vector $\vb$ of eqn.~(\ref{eq:compv}).

\begin{lemma}\label{lem:v}
	With probability at least $1-\delta$, our perturbation vector $\vb$ in Lemma~\ref{lem:fullrankR} satisfies
	\begin{flalign}
		\|\vb\|_2\le\sqrt{3n\mu}\,\|\tilde{\fb}^{(t)}\|_2,
	\end{flalign}
	with $\tilde{\fb}^{(t)}=\Qb^{-\nicefrac{1}{2}}\Ab\Db^2\Ab^\ts\Qb^{-\nicefrac{1}{2}}\tilde{\zb}^t-\Qb^{-\nicefrac{1}{2}}\pb$.
\end{lemma}
\begin{proof}
Recall that $\Qb=\Ab\Db\Wb(\Ab\Db\Wb)^\ts=\Ub_\Qb\Sigmab_\Qb\Ub_\Qb^\ts$. Also, $\Ub_\Qb$ and $\Sigmab_\Qb^{\nicefrac{1}{2}}$ are (respectively) the matrices of the left singular vectors and the singular values of $\Ab\Db\Wb$. Now, let $\widehat{\Vb}$ be the right singular vector of $\Ab\Db\Wb$. Therefore, $\Ab\Db\Wb=\Ub_\Qb\Sigmab_\Qb^{\nicefrac{1}{2}}\widehat{\Vb}^\ts$ is the thin SVD representation of $\Ab\Db\Wb$. Also, from Lemma~\ref{lem:cond3}, we know that $\Qb$ has full rank. Therefore, $\Qb^{\nicefrac{1}{2}}\Qb^{\nicefrac{-1}{2}}=\Ib_m$. Next, we bound $\|\vb\|_2$:
\begin{flalign}
	\|\vb\|_2=&~\|(\Xb\Sb)^{\nicefrac{1}{2}}\Wb(\Ab\Db\Wb)^{\dagger}(\Ab\Db^2\Ab^\ts\hat{\Delta\yb}-\pb)\|_2\nonumber\\
	=&~\|(\Xb\Sb)^{\nicefrac{1}{2}}\Wb(\Ab\Db\Wb)^{\dagger}\Qb^{\nicefrac{1}{2}}\Qb^{\nicefrac{-1}{2}}(\Ab\Db^2\Ab^\ts\hat{\Delta\yb}-\pb)\|_2\nonumber\\
	\le&~\|(\Xb\Sb)^{\nicefrac{1}{2}}\Wb(\Ab\Db\Wb)^{\dagger}\Qb^{\nicefrac{1}{2}}\|_2\,\|\tilde{\fb}^{(t)}\|_2\label{eq:s1}.
\end{flalign}
In the above we used $\Qb^{\nicefrac{-1}{2}}(\Ab\Db^2\Ab^\ts\hat{\Delta\yb}-\pb)=\tilde{\fb}^{(t)}$. Using the SVD of $\Ab\Db\Wb$ and $\Qb$, we get $(\Ab\Db\Wb)^{\dagger}\Qb^{\nicefrac{1}{2}}=\widehat{\Vb}\Sigmab_\Qb^{-1/2}\Ub_\Qb^\ts\,\Ub_\Qb\Sigmab_\Qb^{1/2}\Ub_\Qb^\ts=\widehat{\Vb}\Ub_\Qb^\ts$. Now, note that $\Ub_\Qb\in\RR{m}{m}$ is an orthogonal matrix and $\|\widehat{\Vb}\|_2=1$. Therefore, combining with eqn.~\eqref{eq:s1} yields
	\begin{flalign}
	\|\vb\|_2\le&~\|(\Xb\Sb)^{\nicefrac{1}{2}}\Wb\widehat{\Vb}\Ub_\Qb^\ts\|_2\|\tilde{\fb}^{(t)}\|_2=\|(\Xb\Sb)^{\nicefrac{1}{2}}\Wb\widehat{\Vb}\|_2\|\tilde{\fb}^{(t)}\|_2\nonumber\\
	\le&\|(\Xb\Sb)^{\nicefrac{1}{2}}\Wb\|_2\|\tilde{\fb}^{(t)}\|_2.\label{eq:semi}
	\end{flalign}
The first equality follows from the unitary invariance property of the spectral norm and the second inequality follows from the sub-multiplicativity of the spectral norm and $\|\widehat{\Vb}\|_2=1$.
Our construction for $\Wb$ implies that eqn.~\eqref{eqn:pdprec} holds for any matrix $\Zb$ and, in particular, for $\Zb=(\Xb\Sb)^{\nicefrac{1}{2}}$. Eqn.~\eqref{eqn:pdprec} implies that
	\begin{flalign}\label{eq:w2}
	\nbr{(\Xb\Sb)^{\nicefrac{1}{2}} \Wb \Wb^\ts(\Xb\Sb)^{\nicefrac{1}{2}} - (\Xb\Sb)}_2\le \frac{\zeta}{4} \left(\|(\Xb\Sb)^{\nicefrac{1}{2}}\|_2^2+\frac{\|(\Xb\Sb)^{\nicefrac{1}{2}}\|_F^2}{m}\right)
	\end{flalign}
	holds with probability at least $1-\delta$. Applying Weyl's inequality on the left hand side of the  eqn.~\eqref{eq:w2}, we get
	\begin{flalign}\label{eq:semi2}
	\abs{\nbr{(\Xb\Sb)^{\nicefrac{1}{2}} \Wb}_2^2-\nbr{(\Xb\Sb)^{\nicefrac{1}{2}}}_2^2}\le  \frac{\zeta}{4} \left(\|(\Xb\Sb)^{\nicefrac{1}{2}}\|_2^2+\frac{\|(\Xb\Sb)^{\nicefrac{1}{2}}\|_F^2}{m}\right).
	\end{flalign}
	Using $\zeta\le 1$ and $\|(\Xb\Sb)^{\nicefrac{1}{2}}\|_2^2\le \|(\Xb\Sb)^{\nicefrac{1}{2}}\|_F^2 =\xb^\ts\sbb=n\mu$, we get\footnote{The constant three in eqn.~(\ref{eq:semi3}) could be slightly improved to \nicefrac{3}{2}; we chose to keep the suboptimal constant 
	as the better constant does not result in any significant improvements in the number of iterations of Algorithm~\ref{algo:iipm}.}
	\begin{flalign}\label{eq:semi3}
	\nbr{(\Xb\Sb)^{\nicefrac{1}{2}} \Wb}_2^2\le 3\|(\Xb\Sb)^{\nicefrac{1}{2}}\|_F^2=3n \mu.
	\end{flalign}
	Finally, combining eqns.~\eqref{eq:semi} and \eqref{eq:semi3}, we conclude
	\begin{flalign*}
	\|\vb\|_2\le\sqrt{3n\mu}\|\tilde{\fb}^{(t)}\|_2.
	\end{flalign*}
\end{proof}
%
Intuitively, the bound in Lemma~\ref{lem:v} implies that $\|\vb\|_2$ depends on how close the approximate solution $\hat{\Delta\yb}$ is to the exact solution. Lemma~\ref{lem:v} is particularly useful in proving the convergence of Algorithm~\ref{algo:iipm}, which needs $\|\vb\|_2$ to be a small quantity.  
Next, using the properties of our preconditioner $\Qb^{-1/2}$, we prove that
$\|\Qb^{\nicefrac{-1}{2}}\pb\|_2= \Ocal(\sqrt{n})\sqrt{\mu}$.
This bound allows us to further show that that if we run Algorithm~\ref{algo:PCG} for $\Ocal(\log n)$ iterations, then
$\ \|\vb\|_2\le\frac{\gamma\sigma}{4}\mu.$
%
This inequality is critical in the convergence analysis of Algorithm~\ref{algo:iipm} (see Section~\ref{sxn:conv} for details). Before presenting our feasible IPM algorithm, we first prove the above two inequalities using a couple of lemmas.

{\color{black}Let $\mathcal{F}^0$ be the set of strictly feasible points respectively \ie,
\begin{flalign*}
\mathcal{F}^0=&~\{(\xb,\yb,\sbb): ~(\xb,\sbb)>\zero,~\Ab\xb=\bb,~\Ab^\ts\yb+\sbb=\cbb\}.
\end{flalign*}
In addition, we will need the following definition for the neighborhood
\begin{flalign}
&\mathcal{N}(\gamma)=\Big\{(\xb,\yb,\sbb)\in\mathcal{F}^0:x_i s_i\ge(1-\gamma)\mu\Big\}.\label{eq:neigh}
\end{flalign}
%
%
Here $\gamma \in (0,1)$ and $\mu$ is the duality measure. Note that $\mathcal{N}(\gamma)\subseteq\mathcal{F}^0$ and we assume that $\mathcal{F}^0$ is non-empty.}
\begin{lemma}\label{thm:boundf_f}
	Let $(\xb,\yb,\sbb)\in\mathcal{N}(\gamma)$ and let the sketching matrix $\Wb\in\RR{\dimtwo}{w}$ satisfy the condition in eqn.~\eqref{eq:pdcond1}. Then,
	\begin{flalign}
	\|\Qb^{-\nicefrac{1}{2}}\pb\|_2\le~\left(1+ \frac{\sigma}{\sqrt{1-\gamma}}\right)\sqrt{2n\mu}\,.
	\end{flalign}
\end{lemma}

\begin{proof}	
	To bound $\|\Qb^{-\nicefrac{1}{2}}\pb\|_2$,  first we express $\pb$ as in eqn.~\eqref{eqn:pdef} and rewrite
	\begin{flalign}
	\Qb^{-\nicefrac{1}{2}}\pb
	=&~\Qb^{-\nicefrac{1}{2}}\left(-\sigma\mu\Ab\Sb^{-1}\one_\dimtwo+\Ab\xb\right)\label{eq:recur3_f}.
	\end{flalign}
Applying the triangle inequality on $\|\Qb^{-\nicefrac{1}{2}}\pb\|_2$ in eqn.~\eqref{eq:recur3_f}, we get
	\begin{flalign}
	\|\Qb^{-\nicefrac{1}{2}}\pb\|_2\le \Delta_1+\Delta_2\label{eq:recur4_f}\,,
	\end{flalign}
	where $\Delta_1=~\sigma\mu\|\Qb^{-\nicefrac{1}{2}}\Ab\Db(\Xb\Sb)^{-\nicefrac{1}{2}}\one_\dimtwo\|_2$ and
	$\Delta_2=~\|\Qb^{-\nicefrac{1}{2}}\Ab\Db\Db^{-1}\xb\|_2$.
	In order to bound  $\Delta_1$ and $\Delta_2$, we use the condition of eqn.~\eqref{eq:pdcond1}. In particular, eqn.~\eqref{eq:pdcond1} implies that $\|\Qb^{-\nicefrac{1}{2}}\Ab\Db\|_2\le\sqrt{2}$ as $\zeta\le 1$.
	
	\paragraph{Bounding $\Delta_1$.} Applying submultiplicativity, we get
	\begin{flalign}
	\Delta_1=&~\sigma\mu\,\|\Qb^{-\nicefrac{1}{2}}\,\Ab\Db\,(\Xb\Sb)^{-\nicefrac{1}{2}}\one_\dimtwo\|_2\nonumber\\
	\le&~\sigma\mu\,\|\Qb^{-\nicefrac{1}{2}}\,\Ab\Db\|_2\|(\Xb\Sb)^{-\nicefrac{1}{2}}\one_\dimtwo\|_2
	\le~\sqrt{2}\,\sigma\mu\,\|(\Xb\Sb)^{-\nicefrac{1}{2}}\one_\dimtwo\|_2\nonumber\\
	=&~\sqrt{2}\,\sigma\mu\,\sqrt{\sum_{i=1}^{\dimtwo}\frac{1}{x_i s_i}}
	\le~\sqrt{2}\,\sigma\mu\,\sqrt{\sum_{i=1}^{\dimtwo}\frac{1}{(1-\gamma)\mu}}
	=~\sqrt{2}\,\sigma\,\sqrt{\frac{\dimtwo\,\mu}{(1-\gamma)}}\label{eq:del2_f}\,,
	\end{flalign}
	where we used the fact that $(\xb,\yb,\sbb)\in\mathcal{N}(\gamma)$.
	
	\paragraph{Bounding $\Delta_2$.} Since $\Db=\Sb^{-\nicefrac{1}{2}}\Xb^{\nicefrac{1}{2}}$ and $\xb=\Xb\,\one_\dimtwo$, we get
	\begin{flalign}
	\Delta_2=&~\|\Qb^{-\nicefrac{1}{2}}\,\Ab\Db\,(\Sb^{\nicefrac{1}{2}}\Xb^{-\nicefrac{1}{2}})\,\Xb\,\one_\dimtwo\|_2
	=~\|\Qb^{-\nicefrac{1}{2}}\,\Ab\Db\,(\Sb\Xb)^{\nicefrac{1}{2}}\,\one_\dimtwo\|_2\nonumber\\
	\le&~\|\Qb^{-\nicefrac{1}{2}}\,\Ab\Db\|_2\|(\Sb\Xb)^{\nicefrac{1}{2}}\,\one_\dimtwo\|_2
	\le~\sqrt{2}\,\sqrt{\sum_{i=1}^{\dimtwo}x_i s_i}=~\sqrt{2\dimtwo\,\mu}\label{eq:del3_f}.
	\end{flalign}

	\paragraph{Final bound.} Combining eqns.~\eqref{eq:recur4_f},~\eqref{eq:del2_f}, and~\eqref{eq:del3_f}, we get
	\begin{flalign}
	\|\Qb^{-\nicefrac{1}{2}}\pb\|_2\le~\left(1+ \frac{\sigma}{\sqrt{1-\gamma}}\right)\sqrt{2n\mu}\,.
	\end{flalign}
	This concludes the proof of Lemma~\ref{thm:boundf_f}.
\end{proof}
\begin{lemma}\label{lem:conouter}
	Let $(\xb,\yb,\sbb)\in\mathcal{N}(\gamma)$ and let the sketching matrix $\Wb$ satisfy the conditions of eqns.~\eqref{eq:pdcond1} and \eqref{eq:pdcond2}. Then, after $t\ge\scriptstyle\frac{\log(n\,\psi)}{\log(\nicefrac{1}{\zeta})}$ iterations of the iterative solver in Algorithm~\ref{algo:PCG}, we have $\|\vb\|_2\le\frac{\gamma\sigma}{4}\mu.$
	Here $\psi=\scriptstyle\frac{4\sqrt{6}\left(1+ \nicefrac{\sigma}{\sqrt{1-\gamma}}\right)}{\gamma\sigma}$ and $\tilde{\fb}^{(t)}=\Qb^{-\nicefrac{1}{2}}\Ab\Db^2\Ab^\ts\Qb^{-\nicefrac{1}{2}}\tilde{\zb}^t-\Qb^{-\nicefrac{1}{2}}\pb$ is the residual of the solver.
\end{lemma}

\begin{proof}
	Combining Lemma~\ref{thm:boundf_f} and the condition in eqn.~\eqref{eq:pdcond2}, we get
	\begin{flalign}\label{eq:bd}
	\|\tilde{\fb}^{(t)}\|_2\le\zeta^t\left(1+ \frac{\sigma}{\sqrt{1-\gamma}}\right)\sqrt{2n\mu}.
	\end{flalign}
	%
	Next, combining Lemma~\ref{lem:v} and eqn.~\eqref{eq:bd} we get
	\begin{flalign*}
	\|\vb\|_2\le\sqrt{3n\mu}\,\|\tilde{\fb}^{(t)}\|_2\le\sqrt{6}n\,\zeta^t\left(1+ \frac{\sigma}{\sqrt{1-\gamma}}\right)\mu
	\end{flalign*}
	Therefore, $\|\vb\|_2\le\frac{\gamma\sigma\mu}{4}$ holds if
	$\sqrt{6}n\,\zeta^t\left(1+ \nicefrac{\sigma}{\sqrt{1-\gamma}}\right)\mu\le\frac{\gamma\sigma\mu}{4}$, which holds for our choice of $t$.
	Now, fixing $\gamma$, $\sigma$, and $\zeta$, after $t=\Ocal (\log\dimtwo)$ iterations of Algorithm~\ref{algo:PCG} the conclusions of the lemma hold.
\end{proof}
Now, we are ready to present the feasible IPM algorithm. Recall the definition the neighborhood $\mathcal{N}(\gamma)$ in eqn.~\eqref{eq:neigh}.
%
%

\begin{algorithm}[H]
	\caption{Feasible IPM}\label{algo:iipm}%
	\hspace*{\algorithmicindent} \textbf{Input:}
		$\Ab\in\RR{\dimone}{\dimtwo}$,  $\bb\in\R{\dimone}$, $\cbb\in\R{\dimtwo}$, $\gamma \in (0,1)$, tolerance $\epsilon> 0$, $\sig\in (0,\nicefrac{4}{5})$;

		\hspace*{\algorithmicindent} \textbf{Initialize:} $k\gets 0$; initial point $(\xb^{0},\yb^{0},\sbb^{0})\in\mathcal{F}^0$;
	\begin{algorithmic}[1]
		\While{$\mu_k > \epsilon$}
		
		\State Compute sketching matrix $\Wb \in \mathbb{R}^{n \times w}$ (Section~\ref{sxn:background}) with $\zeta=1/2$ and $\delta = O(n^{-1})$;
		\State Solve eqn.~\eqref{eq:precond_alt} for $\zb$ using Algorithm~\ref{algo:PCG} with $\Wb$ from step (2) and $t=\Ocal(\log n)$,
		
		\hspace*{-3mm} and then Compute $\hat{\Delta \yb}=\Qb^{\nicefrac{-1}{2}}{\zb}$;
		\State Compute $\vb$ using eqn.~\eqref{eq:compv} with $\Wb$ from step (2); $\hat{\Delta\sbb}$ using eqn.~\eqref{eq:dels}; $\hat{\Delta\xb}$ using 
		
		\hspace{-3mm} eqn.~\eqref{eq:delxhat};
		\State \label{stepalpha1} Compute $\alphamax = \argmax\{ \alpha \in [0,1] : (\xb^k,\yb^k,\sbb^k) + \alpha (\hat{\Delta \xb}^k,\hat{\Delta \yb}^k,\hat{\Delta \sbb}^k)  \in \neigh\}$.
		\State \label{stepalpha2}Compute $\alphaused = \argmin\{\alpha \in [0, \alphamax]: (\xb^k + \alpha \hat{\Delta \xb}^k)^\ts (\sbb^k + \alpha \hat{\Delta\sbb}^k)\}$.
		\State Compute $(\xb^{k+1}, \yb^{k+1}, \sbb^{k+1}) = (\xb^k,\yb^k,\sbb^k) + \alphabar (\hat{\Delta \xb}^k,\hat{\Delta \yb}^k,\hat{\Delta \sbb}^k)$; set $k \gets k + 1$;
		
		\EndWhile
	\end{algorithmic}
\end{algorithm}

\vspace{0.02in}\noindent\textbf{Running time.} We start by discussing the running time to compute $\vb$. As discussed in Section~\ref{sxn:PCG},
$(\Ab\Db\Wb)^{\dagger}$ can be computed in $\Ocal(\nnz{\Ab}\cdot \log(m/\delta)+m^3\log(m/\delta))$ time.  Now, as $\Wb$ has $\Ocal(\log(m/\delta))$ non-zero entries per row, pre-multiplying by $\Wb$ takes $\Ocal(\nnz{\Ab}\log(m/\delta))$ time (assuming $\nnz{A}\ge n$). Since $\Xb$ and $\Sb$ are diagonal matrices, computing $\vb$ takes  $\Ocal(\nnz{\Ab}\cdot \log(m/\delta)+m^3\log(m/\delta))$ time, which is asymptotically the same as computing $\Qb^{\nicefrac{-1}{2}}$ (see eqn.~(\ref{eqn:svdQ})).

We now discuss the overall running time of Algorithm~\ref{algo:iipm}. At each iteration, with failure probability $\delta$, the preconditioner $\Qb^{\nicefrac{-1}{2}}$ and the vector $\vb$ can be computed in
$\Ocal(\nnz{\Ab}\cdot \log(m/\delta)+m^3\log(m/\delta))$ time. In addition, for $t=\Ocal(\log n)$ iterations of Algorithm~\ref{algo:PCG}, all the matrix-vector products in the CG or Chebyshev iteration can be computed in $\Ocal(\nnz{\Ab}\cdot \log n)$ time. Therefore, the computational time for steps (2)-(4) is given by $\Ocal(\nnz{\Ab}\cdot(\log n+ \log(m/\delta))+m^3\log(m/\delta))$. 
\textcolor{black}{Finally, considering $\epsilon$ to be a constant, if we assume that the IPM needs $k=c\,n$ iterations to converge  and accordingly, if we fix the failure probability $\delta=\frac{0.1}{c\,n}$ for some suitable constant $c$, then taking a union bound over all the IPM iterations, our algorithm converges with probability  at least $1-c\,n\cdot\frac{0.1}{c\,n}=0.9$ and
the running time at each iteration is given by
$\Ocal((\nnz{\Ab}+m^3)\log n)$.}

\subsection{Convergence Analysis of Algorithm~\ref{algo:iipm}}\label{sxn:conv}
{\color{black} In this section, we prove a set of results that ultimately establish Theorem~\ref{thm:2f} and guarantee the convergence of Algorithm~\ref{algo:iipm}. Due to the use of an approximate solver, these proofs typically differ from the standard analysis of long-step feasible IPM~\citep{wright1997primal} in many aspects. For example, all the major results in this section rely on the condition that the error due to the linear solver is small \ie, $\|\vb\|_2$ is small, whereas the standard convergence analysis does not have this requirement as the linear system there is solved exactly \ie $\|\vb\|_2$ is always zero. This difference makes our case more intricate as we deal with an extra term involving $\|\vb\|_2$ which needs more care.

On the other hand, while the origin of the statements of our feasible IPM results is essentially \citep{Mon03}, the proofs are different from that of \citep{Mon03}. The exact same analysis of \citep{Mon03} (just by making the primal and dual residuals equal to zero) neither directly applies to the feasible case, nor matches the best iteration complexity of it, whereas our current analysis has the iteration complexity $\Ocal(n\log 1/\epsilon)$, which is the best known for feasible long-step path following IPM algorithms. The proofs that look similar to \citep{Mon03} also have differences. We discuss them individually before the respective lemmas. The only overlap we have with \citep{Mon03} is our Lemma~\ref{lemmaminalpha1_f} that works for both feasible and infeasible setting. Now, we proceed to prove our Theorem~\ref{thm:2f}.
} 

%
%

First, we can rewrite the linear system of eqns.~\eqref{eq:delxhat}, \eqref{eq:addl}, \eqref{eq:delshat} as follows:
\begin{subequations}\label{eq:iip_3_f}
	\begin{flalign}
	\Ab\hat{\Delta\xb}=&~\zero, \label{eq:iip_3_f_1}\\
	\Ab^\ts\hat{\Delta\yb}+\hat{\Delta\sbb}=&~\zero,  \label{eq:iip_3_f_2}\\
	\Xb\hat{\Delta\sbb}+\Sb\hat{\Delta\xb}=&-\Xb\Sb\,\one_\dimtwo+\sigma\mu\,\one_\dimtwo - \vb.  \label{eq:iip_3_f_3}
	\end{flalign}
\end{subequations}
Indeed, we now show how to derive eqns.~\eqref{eq:delxhat}, \eqref{eq:addl}, \eqref{eq:delshat} from eqn.~\eqref{eq:iip_3_f}. Pre-multiplying both sides of eqn.~\eqref{eq:iip_3_f_3} by $\Ab\Sb^{-1}$ and noting that $\Db^2=\Xb\Sb^{-1}$, we get
\begin{flalign}
&~\Ab\Db^2\hat{\Delta\sbb}+\Ab\hat{\Delx}=-\Ab\Xb\one_n+\sigma\mu\Ab\Sb^{-1}\one_n-\Ab\Sb^{-1}\vb\nonumber\\
\Rightarrow&~\Ab\Db^2\hat{\Delta\sbb}=-\Ab\xb+\sigma\mu\Ab\Sb^{-1}\one_n-\Ab\Sb^{-1}\vb.\label{eq:iip_3_f1}
\end{flalign}
Eqn.~\eqref{eq:iip_3_f1} holds as $\Ab\Xb\one_n=\Ab\xb$ and, from eqn.~\eqref{eq:iip_3_f_1}, $\Ab\hat{\Delx}=\zero$. Next, pre-multiplying eqn.~\eqref{eq:iip_3_f_2} by $\Ab\Db^2$, we get
\begin{flalign}
&~\Ab\Db^2\Ab^\ts\hat{\Dely}+\Ab\Db^2\hat{\Dels}=\zero\nonumber\\
\Rightarrow &~\Ab\Db^2\Ab^\ts\hat{\Dely}=\Ab\xb-\sigma\mu\Ab\Sb^{-1}\one_n+\Ab\Sb^{-1}\vb=\pb+\Ab\Sb^{-1}\vb.\label{eq:iip_3_f11}
\end{flalign}
The first equality in eqn.~\eqref{eq:iip_3_f11} follows from eqn.~\eqref{eq:iip_3_f1} and the definition of $\pb$. This establishes
eqn.~\eqref{eq:addl}. Eqn.~\eqref{eq:delshat} directly follows from eqn.~\eqref{eq:iip_3_f_2}. Finally, we get eqn.~\eqref{eq:delxhat} by pre-multiplying eqn.~\eqref{eq:iip_3_f_3} by $\Sb^{-1}$.
We will now use a slightly different notations. We define the next point traversed by the algorithm as $(\xnew,\ynew, \snew)$, where
\begin{flalign}
(\xnew,\ynew,\snew) &= (\xb, \yb, \sbb) + \alpha(\hat{\Delta\xb},\hat{\Delta \yb},\hat{\Delta\sbb}),\ \mbox{and} \\
\munew &= \left(\nicefrac{1}{n}\right)\xnew^\ts \snew.
\end{flalign}
Our goal is to bound the number of outer iterations required by the feasible IPM algorithm.
To do so, we bound the magnitude of the step size $\alpha$. First, we provide an upper bound on $\alpha$, which allows us to show that each new point $(\xnew,\ynew,\snew)$ traversed by the algorithm stays within the neighborhood $\neigh$.
Second, we provide a lower bound on $\alpha$, which allows us to bound the number of iterations required. 
The following Lemma will be used throughout the section.
\begin{lemma}\label{lemmamaxalpha0}
	Assume $(\hat{\Delx},\hat{\Dels},\hat{\Dely})$ satisfies eqns. $(\ref{eq:iip_3})$ for some $\sig \in \mathbb{R}$ and $\vb \in \mathbb{R}^\dimtwo$. Let $(\xb,\yb,\sbb)$ be any point such that $(\xb,\sbb)>0$. Then, for every $\alpha \in \mathbb{R}$, 
	\begin{flalign*}
	(a)&~~\xnew \circ \snew = (1- \alpha) \xcircs + \alpha \sig \mu \one_\dimtwo - \alpha \vb + \alpha^2 \hat{\Delta\xb}\circ\hat{\Delta \sbb}\,,\\
	(b)&~~\munew =  [1-\alpha(1-\sig)]\mu - \frac{\alpha\,\vb^\ts\one_\dimtwo}{\dimtwo}.
	\end{flalign*}
\end{lemma}
\begin{proof}
Proving $(a)$: 
	\begin{flalign}
	\xnew \circ \snew &= (\xb + \alpha \hat{\Delta\xb}) \circ (\sbb + \alpha \hat{\Delta\sbb})  \notag \\ 
	& = \xcircs + \alpha(\xb \circ \hat{\Delta\sbb} + \sbb\circ\hat{\Delta\xb}) + \alpha^2 \hat{\Delta\xb}\circ\hat{\Delta \sbb}  \notag\\
	& = \xcircs + \alpha(-\xcircs + \sig\mu \one_\dimtwo - \vb) + \alpha^2\hat{\Delta\xb}\circ\hat{\Delta \sbb}  \notag \\
	& = (1-\alpha)\xcircs + \alpha\sig\mu \one_\dimtwo - \alpha\vb + \alpha^2\hat{\Delta\xb}\circ\hat{\Delta \sbb}\,, \notag
	\end{flalign} 
	where the third equality follows from eqn.~\eqref{eq:iip_3_f_3}. Now, left-multiply the above equality by $\one_\dimtwo^\ts$ and divide by $\dimtwo$ to obtain $(b)$. (Notice that $\hat{\Delx}^\ts\hat{\Dels}=0$ from eqns.~\eqref{eq:iip_3_f_1} and \eqref{eq:iip_3_f_2}.)
\end{proof}

\noindent Next, we provide an upper bound on $\alpha$, ensuring that each new point $(\xnew,\ynew,\snew)$ traversed by the algorithm stays within the neighborhood $\neigh$. Note that the following result resembles Lemma~3.5 of \citep{Mon03}, but what makes Lemma~\ref{lemmamaxalpha1_f} different from it is the fact that here we need to additionally prove the strict feasibility of the new iterate \ie~$(\xnew,\ynew,\snew)\in\mathcal{F}^0$ (in order to show $(\xnew,\ynew,\snew) \in \neigh$), which was not proven in Lemma~3.5 of \citep{Mon03}.

\begin{lemma}\label{lemmamaxalpha1_f}
	Assume $(\hat{\Delx},\hat{\Dely},\hat{\Dels})$  satisfies eqns. $(\ref{eq:iip_3_f})$ for some $\sig > 0$, $(\xb, \yb,\sbb) \in \neigh$ for $\gamma \in (0,1)$, and $\|\vb\|_2 \leq \frac{\gamma \sig \mu}{4}$. Then, $(\xnew,\ynew,\snew) \in \neigh$ for every scalar $\alpha$ such that 
	\begin{flalign} \label{alphalemma1_f}
	0 \leq \alpha \leq \min \left\{ 1, \frac{\gamma \sig \mu}{4 \delxdelsnorm}\right\} .
	\end{flalign}
\end{lemma}

\begin{proof}
	First, we show that  $\xnew \circ \snew \geq (1-\gamma) \munew \one_\dimtwo$. From Lemma~\ref{lemmamaxalpha0}, we get
\begin{flalign}
&\xnew \circ \snew - (1-\gamma) \munew \one_\dimtwo \notag \\
&= (1-\alpha)\left(\xcircs - (1-\gamma)\mu\,\one_\dimtwo\right) + \alpha \threefac\, \one_\dimtwo - \alpha \left(\vb - (1-\gamma)\frac{\vb^\ts\one_\dimtwo}{\dimtwo}\one_\dimtwo\right)\nonumber\\
&~~~~~~~~~~~~~~~~~~~~~~~~~~~~~~~~~~~~~~~~~~~~~~~+ \alpha^2 \left(\delxdels \right)  \notag \\
& \geq \alpha \left( \threefac  - \nbr{\vb- (1-\gamma)\frac{\vb^\ts\one_\dimtwo}{\dimtwo}\,\one_\dimtwo}_\infty    - \alpha \nbr{\hat{\Delta\xb}\circ\hat{\Delta\sbb}}_\infty\right) \one_\dimtwo  \notag \\
& \geq \alpha \Bigg(\threefac - 2 \|\vb\|_\infty -  \alpha\delxdelsnorm \Bigg)\one_\dimtwo \notag \\ 
& \geq \alpha \Bigg(\threefac - \frac{\threefac}{2} -  \frac{\threefac}{4}\Bigg) \one_\dimtwo  = \, \alpha\frac{\threefac}{4}\one_n\ge\zero \notag.
\end{flalign} 
The first inequality follows from $\xcircs\ge(1-\gamma)\mu\,\one_\dimtwo$, because $(\xb,\yb,\sbb)\in\mathcal{N}(\gamma)$ and $\ab\le\|\ab\|_\infty\,\one_\dimtwo$ for any vector $\ab\in\R{\dimtwo}$. The second-to-last inequality follows from the fact that for any $\ub \in \mathbb{R}^\dimtwo$ and $\delta \in [0,\dimtwo]$, $\nbr{\ub - \delta \frac{\ub^\ts \one_\dimtwo}{\dimtwo}\,\one_\dimtwo}_\infty \leq (1+ \delta)\|\ub\|_\infty$. Thus, we prove that the point $(\xnew,\ynew,\snew)$ satisfies the proximity condition for $\neigh$.


Finally, we show that $(\xnew,\ynew,\snew)\in\mathcal{F}^0$ \ie~it satisfies the primal and dual constraints and $(\xnew,\snew)>\zero$. From eqn.~\eqref{eq:iip_3_f} and the fact that $(\xb,\yb,\sbb)\in\mathcal{F}^0$,  we get $\Ab\,\xnew=\Ab\xb+\alpha\Ab\hat{\Delta\xb}=\bb$. Similarly, $$\Ab^\ts\,\ynew+\snew=(\Ab^\ts\yb+\sbb)+\alpha(\Ab^\ts\hat{\Delta\yb}+\hat{\Delta\sbb})=\Ab^\ts\yb+\sbb = \cbb.$$  
We now show that $(\xnew,\snew)>\zero$. For $\alpha=0$, we trivially have$(\xnew,\snew)=(\xb,\sbb)>\zero$. To prove $(\xnew,\snew)>\zero$ for $0<\alpha\le1$, we first show $\munew>0$. Using $\gamma \in (0,1)$, the inequality $\abs{\frac{\vb^\ts\one_n}{n}}\le\|\vb\|_\infty\le\|\vb\|_2$, and the assumption $\|\vb\|_2\le\frac{\threefac}{4}$, we get $\frac{\vb^\ts\one_\dimtwo}{\dimtwo} < \frac{\sig \mu}{4}$. 
Thus, from Lemma~\ref{lemmamaxalpha0}(b), 
	\begin{flalign}
	\munew &= [1-\alpha(1-\sig)]\mu - \alpha \,\frac{\vb^\ts\one_\dimtwo}{\dimtwo}\nonumber\\
	&> [1-\alpha(1-\sig)]\mu -  \alpha\,\frac{\sig \mu}{4}\nonumber\\
	&= (1-\alpha)\mu + \alpha\,\frac{3\,\sig \mu}{4}>0.
	\label{last}
	\end{flalign} 
The last inequality holds because $\alpha\in (0,1]$, $\sigma\in (0,1)$, and $\mu>0$. 
We already have $\xnew\circ\snew\ge (1-\gamma)\munew\one_n$. Combining with $\munew>0$ and $\gamma\in (0,1)$ this implies that $\xnew\circ\snew>\zero$. Therefore, $x_i(\alpha)\,s_i(\alpha)>0$ for all $i=1\dots n$ which implies that, for each $i$, either both $x_i(\alpha)$ and $s_i(\alpha)$ are positive or both $x_i(\alpha)$ and $s_i(\alpha)$ are negative. We will use contradiction to prove that the second case is not possible.

Indeed, assume that $x_i(\alpha)<0$ and $s_i(\alpha)<0$ for some $i=1 \ldots n$. First, we rewrite $x_i(\alpha)$ and $s_i(\alpha)$ as follows\footnote{Here, $x_i(\alpha)$, $s_i(\alpha)$, $x_i$, $s_i$, $\hat{\Delta x_i}$, and $\hat{\Delta s_i}$ are the $i$-th elements of $\xnew$, $\snew$, $\xb$, $\sbb$, $\hat{\Delx}$, and $\hat{\Dels}$, respectively.}
\begin{subequations}
\begin{flalign}
x_i(\alpha)=&~x_i+\alpha\hat{\Delta x_i}<0\label{eq:x_i},\\
s_i(\alpha)=&~s_i+\alpha\hat{\Delta s_i}<0\label{eq:s_i}.
\end{flalign}
\end{subequations}
Recall that both $x_i$ and $s_i$ are positive. Therefore, pre-multiplying eqn.~\eqref{eq:x_i} by $s_i$ and eqn.~\eqref{eq:s_i} by $x_i$ we get
\begin{subequations}
\begin{flalign}
x_is_i+\alpha s_i\hat{\Delta x_i}<&~0\label{eq:x_i1},\\
x_is_i+\alpha x_i\hat{\Delta s_i}<&~0\label{eq:s_i1}.
\end{flalign}
\end{subequations}
Adding eqns.~\eqref{eq:x_i1} and \eqref{eq:s_i1} and applying eqn.~\eqref{eq:iip_3_f_3} (element-wise), we get
\begin{flalign}
&~2x_is_i+\alpha (s_i\hat{\Delta x_i}+x_i\hat{\Delta s_i})<0\nonumber\\
\Rightarrow&~2x_is_i+ \alpha (-x_is_i+\sigma\mu-v_i)<0\nonumber\\
\Rightarrow&~ (2-\alpha)x_is_i+\alpha\sigma\mu-\alpha v_i<0\nonumber\\
\Rightarrow&~ v_i>\frac{2-\alpha}{\alpha}x_is_i+\sigma\mu > \sigma\mu\label{eq:xs}.
\end{flalign}
In the above $v_i$ is the $i$-th element of $\vb$; the first inequality in eqn.~\eqref{eq:xs} holds because $\alpha>0$; the second inequality in eqn.~\eqref{eq:xs} because $\frac{2-\alpha}{\alpha}x_is_i>0$ ($x_i,s_i>0$ and $0 < \alpha \leq 1$). Using $\|\vb\|_2\le\frac{\threefac}{4}$ we get
$$v_i\le\|\vb\|_\infty\le\|\vb\|_2\le\frac{\threefac}{4}<\sigma\mu\,,$$
for all $i=1\ldots n$. This contradicts the inequality of eqn.~\eqref{eq:xs}; thus, both $x_i(\alpha)>0$ and $s_i(\alpha)>0$ for all $i=1\ldots n$ and all $\alpha\in[0,1]$.
\end{proof}

\noindent We now cite a result from \citep{Mon03}  that provides a lower bound on $\alphaused$ and the corresponding $\mu(\alphaused)$. Note that \citep{Mon03} presented it in context of infeasible IPM; however, it holds for the feasible case as well, as long as the perturbation vector $\vb$ satisfies $\|\vb\|_2\le\frac{\gamma\sigma\mu}{4}$ at each iteration.
\begin{lemma}[Lemma 3.6 of \citep{Mon03}]\label{lemmaminalpha1_f}
	At each iteration of the Algorithm~\ref{algo:iipm}, if $\|\vb\|_2\le\frac{\gamma\sigma\mu}{4}$, then the step size $\alphaused$ satisfies 
	\begin{flalign} \label{alphalemma2_f}
	\alphaused \geq \min \left\{ 1,  \frac{ \min   \{ \gamma \sigma, (1 - \frac{5}{4} \sigma) \} \mu } {4  \|\hat{\Delta x} \circ \hat{\Delta s}\|_\infty}  \right\} 
	\end{flalign}
	and 
	\begin{flalign} \label{mulemma2_f}
	\mu(\alphaused)= \Big[ 1 - \frac{\alphaused}{2} (1-\frac{5}{4}\sig) \Big] \mu. 
	\end{flalign}
\end{lemma}
\noindent At this point, we have provided a lower bound (eqn.~(\ref{alphalemma2_f})) for the allowed values of the step size $\alphaused$. Next, we will show that this lower bound is bounded away from zero. From eqn.~(\ref{alphalemma2_f}), is suffices to show that $\delxdelsnorm$ is bounded. First, we state the following inequality that will be instrumental in proving Lemma~\ref{lemmaminalpha2_f}.

\begin{lemma}\label{prop:wri}
	Let $\ab,\bb\in\R{n}$ be any two vectors such that $\ab^\ts\bb\ge 0$. Then
	\begin{flalign*}
	\|\ab\circ\bb\|_2\le \|\ab+\bb\|_2^2\,.
	\end{flalign*}
\end{lemma}
\noindent See~\citep{wright1997primal} for a proof of Lemma~\ref{prop:wri}; as a matter of fact,~\citep{wright1997primal} proved $\|\ab\circ\bb\|_2\le2^{-3/2}\|\ab+\bb\|_2^2$, which is tighter. This tighter bound is not needed in our proof. 

\begin{lemma}\label{lemmaminalpha2_f}
	Let $(\xb,\yb,\sbb) \in \neigh$ and $\|\vb\|_2\le\frac{\gamma\sigma\mu}{4}$. Then   $(\hat{\Delta \xb},\hat{\Delta \yb},\hat{\Delta \sbb})$ satisfies 
	
\begin{flalign} \label{normsmax_f}
\|\hat{\Delx}\circ\hat{\Dels}\|_\infty\le\left(1+\frac{\sigma^2}{1-\gamma}\right)n\mu + \frac{\gamma^2\sigma^2}{16(1-\gamma)}\mu+\frac{\gamma\sigma^2}{2}\mu\,.
\end{flalign}
\end{lemma}

\begin{proof}
First, we multiply eqn.~(\ref{eq:iip_3_f_3}) on the left by $(\Xb \Sb)^{-1/2}$ to get
\begin{flalign}\label{eq:Dint}
\Db^{-1} \hat{\Delta \xb} + \Db \hat{\Delta \sbb} = -(\Xb \Sb)^{1/2}\one_\dimtwo + 
\sig \mu (\Xb \Sb)^{-1/2} \one_\dimtwo - (\Xb \Sb)^{-1/2} \vb\,.
\end{flalign} 
Next, pre-multiplying eqn.~\eqref{eq:iip_3_f_2} by $\hat{ \Delta \xb}^\ts$ and applying eqn.~\eqref{eq:iip_3_f_1}, we have $\hat{\Delx}^\ts\hat{ \Delta \sbb}=0$. This also implies that $\hat{\Delx}^\ts\hat{ \Delta \sbb}=(\Db^{-1}\hat{\Delx})^\ts(\Db\hat{ \Delta \sbb})$. Applying Lemma~\ref{prop:wri} with $\ab=\Db^{-1}\hat{\Delx}$ and $\bb=\Db\hat{\Dels}$, and using eqn.~\eqref{eq:Dint} we get
\begin{flalign}
\|\hat{ \Delta \xb}\circ\hat{ \Delta \sbb}\|_2=&~\|(\Db^{-1}\hat{ \Delta \xb})\circ(\Db\hat{ \Delta \sbb})\|_2\le\|\Db^{-1} \hat{\Delta \xb} + \Db \hat{\Delta \sbb}\|_2^2\nonumber\\
=&~\|-(\Xb \Sb)^{1/2}\one_\dimtwo + 
\sig \mu (\Xb \Sb)^{-1/2} \one_\dimtwo - (\Xb \Sb)^{-1/2} \vb\|_2^2\nonumber\\
=&~\|(\Xb \Sb)^{1/2}\one_\dimtwo+(\Xb \Sb)^{-1/2}(\vb-\sigma\mu\one_\dimtwo)\|_2^2\nonumber\\
=&~\|(\Xb \Sb)^{1/2}\one_\dimtwo\|_2^2+ 2\cdot\one_\dimtwo^\ts(\vb-\sigma\mu\one_\dimtwo)+\|(\Xb \Sb)^{-1/2}(\vb-\sigma\mu\one_\dimtwo)\|_2^2\nonumber\\
\le&~n\mu+2n (\|\vb\|_2-\sigma\mu)+\|(\Xb \Sb)^{-1/2}(\vb-\sigma\mu\one_\dimtwo)\|_2^2\label{eq:dineq}.
\end{flalign}
The inequality in eqn.~\eqref{eq:dineq} follows from $\|(\Xb \Sb)^{1/2}\one_\dimtwo\|_2^2=n\mu$ and $\abs{\one_\dimtwo^\ts\vb}\le n\,\|\vb\|_\infty\le n\,\|\vb\|_2$. Next, consider the last term on the right hand side of eqn.~\eqref{eq:dineq}:
\begin{flalign}
\|(\Xb \Sb)^{-1/2}(\vb-\sigma\mu\one_\dimtwo)\|_2^2=&~\|(\Xb \Sb)^{-1/2}\vb\|_2^2+\sigma^2\mu^2\|(\Xb \Sb)^{-1/2}\one_\dimtwo\|_2^2-2\sigma\mu\,\one_n^\ts(\Xb\Sb)^{-1}\vb\nonumber\\
\le&~\frac{\|\vb\|_2^2}{\min_i x_is_i}+\sigma^2\mu^2\sum_{i=1}^{n}\frac{1}{x_is_i}-2\sigma\mu\sum_{i=1}^{n}\frac{v_i}{x_is_i}.\label{eq:dint2}
\end{flalign}
Eqn.~\eqref{eq:dint2} follows from $\|(\Xb \Sb)^{-1/2}\vb\|_2^2\le\|(\Xb \Sb)^{-1/2}\|_2^2\|\vb\|_2^2$ and $\|(\Xb \Sb)^{-1/2}\|_2^2=\nicefrac{1}{\min_i x_is_i}$. Moreover, it is easy to verify that $\|(\Xb \Sb)^{-1/2}\one_\dimtwo\|_2^2=\sum_{i=1}^{n}\nicefrac{1}{x_is_i}$ and $\one_n^\ts(\Xb\Sb)^{-1}\vb=\sum_{i=1}^{n}\nicefrac{v_i}{x_is_i}$. Now, we have $x_is_i\ge(1-\gamma)\mu$ for $i=1 \ldots n$ as $(\xb,\yb,\sbb)\in\mathcal{N}(\gamma)$. Also, $x_is_i\le\sum_{i=1}^{n}x_is_i=n\mu$. Using the above we rewrite eqn.~\eqref{eq:dint2} as 
\begin{flalign}
\|(\Xb \Sb)^{-1/2}(\vb-\sigma\mu\one_\dimtwo)\|_2^2\le\frac{\|\vb\|_2^2}{(1-\gamma)\mu}+\frac{n\sigma^2\mu^2}{(1-\gamma)\mu}-2\sigma\mu\frac{\sum_{i=1}^{n}v_i}{n\mu}.\label{eq:dint3}
\end{flalign}
Combining eqns.~\eqref{eq:dineq} and~\eqref{eq:dint3} we get
\begin{flalign}
\|\hat{ \Delta \xb}\circ\hat{ \Delta \sbb}\|_2\le&~ n\mu+2n (\|\vb\|_2-\sigma\mu)+\frac{\|\vb\|_2^2}{(1-\gamma)\mu}+\frac{n\sigma^2\mu}{(1-\gamma)}-2\sigma\frac{\sum_{i=1}^{n}v_i}{n}\nonumber\\
=&~\left(1+\frac{\sigma^2}{1-\gamma}\right)n\mu+2n (\|\vb\|_2-\sigma\mu)+\frac{\|\vb\|_2^2}{(1-\gamma)\mu}-2\sigma\frac{\sum_{i=1}^{n}v_i}{n}\,.\label{eq:dint4}
\end{flalign}
Using  $\|\vb\|_2\le\frac{\gamma\sigma\mu}{4}$ and the fact that $\abs{\nicefrac{\vb^\ts\one_n}{n}}\le\|\vb\|_\infty\le\|\vb\|_2$, we get
\begin{flalign}
\|\hat{ \Delta \xb}\circ\hat{ \Delta \sbb}\|_2\le\left(1+\frac{\sigma^2}{1-\gamma}\right)n\mu+\frac{\gamma^2\sigma^2}{16 (1-\gamma)}\mu+\frac{\gamma\sigma^2}{2}\mu.
\end{flalign}
Finally, we conclude the proof using $\|\hat{ \Delta \xb}\circ\hat{ \Delta \sbb}\|_\infty\le\|\hat{ \Delta \xb}\circ\hat{ \Delta \sbb}\|_2$.
\end{proof}

The next result guarantees the convergence of Algorithm~\ref{algo:iipm}
.
\begin{lemma}\label{theoremOuter_f}
Assume that the constants $\gamma$ and $\sigma$ are such that $\max\{\gamma^{-1},(1-\gamma)^{-1},\sigma^{-1},(1-\frac{5}{4}\sigma)^{-1}\}=\Ocal(1)$. At each iteration of Algorithm~\ref{algo:iipm}, if $\|\vb\|_2\le\frac{\gamma\sigma\mu}{4}$, then after 
	$k=\mathcal{O}(\dimtwo \log{\nicefrac{1}{\epsilon}})$ iterations, $(\xb^{k}, \sbb^{k}, \yb^{k})$ satisfies $$\mu_k \leq \epsilon \mu_0.$$
\end{lemma}

\vspace{-2mm}
\begin{proof}
From Lemma~\ref{lemmaminalpha2_f}, 
\begin{flalign}
&~\|\hat{\Delx}\circ\hat{\Dels}\|_\infty\le \left(1+\frac{\sigma^2}{1-\gamma} + \frac{\gamma^2\sigma^2}{16(1-\gamma)}+\frac{\gamma\sigma^2}{2}\right)n\mu\nonumber\\
\Rightarrow&~\mu\|\hat{\Delx}\circ\hat{\Dels}\|_\infty^{-1}\ge n^{-1} \left(1+\frac{\sigma^2}{1-\gamma} + \frac{\gamma^2\sigma^2}{16(1-\gamma)}+\frac{\gamma\sigma^2}{2}\right)^{-1}.\label{eq:step1}
\end{flalign}
Combining eqns.~\eqref{alphalemma2_f} and \eqref{eq:step1} we get
\begin{flalign}
	\alphaused \geq \min \left\{ 1,  \frac{ \min   \{ \gamma \sigma, (1 - \frac{5}{4} \sigma) \} } {4  \,n \left(1+\frac{\sigma^2}{1-\gamma} + \frac{\gamma^2\sigma^2}{16(1-\gamma)}+\frac{\gamma\sigma^2}{2}\right)}  \right\}.\label{eq:step2}
\end{flalign}
Let $\max\{\gamma^{-1},(1-\gamma)^{-1},\sigma^{-1},(1-\frac{5}{4}\sigma)^{-1}\}\le \lambda$ for some constant $\lambda> 1$. Therefore, $\gamma\sigma\ge\frac{1}{\lambda^2}$ and $(1-\frac{5}{4}\sigma)\ge\frac{1}{\lambda}$, which further implies that $\min\{ \gamma \sigma, (1 - \frac{5}{4} \sigma)\}\ge\min\{\frac{1}{\lambda},\frac{1}{\lambda^2}\}=\frac{1}{\lambda^2}$. Also, $\frac{1}{1-\gamma}\le\lambda$. Combining these with eqn.~\eqref{eq:step2}, we get
\begin{flalign}
\alphaused\ge\min\left\{1,\frac{\frac{1}{\lambda^2}}{4n (1+\lambda\sigma^2+\frac{\lambda\gamma^2\sigma^2}{16}+\frac{\gamma\sigma^2}{2})}\right\}=\min\left\{1,\frac{1}{4n(\lambda^2+\lambda^3\sigma^2+\frac{\lambda^3\gamma^2\sigma^2}{16}+\frac{\lambda^2\gamma\sigma^2}{2})}\right\}.\label{eq:step3}
\end{flalign}
Note that in eqn.\eqref{eq:step3}, $4n(\lambda^2+\lambda^3\sigma^2+\frac{\lambda^3\gamma^2\sigma^2}{16}+\frac{\lambda^2\gamma\sigma^2}{2})> 1$. Thus,
\begin{flalign}
&~\alphaused\ge \frac{1}{4n(\lambda^2+\lambda^3\sigma^2+\frac{\lambda^3\gamma^2\sigma^2}{16}+\frac{\lambda^2\gamma\sigma^2}{2})}\nonumber\\
\Rightarrow&~ \frac{\alphaused}{2}\left(1-\frac{5}{4}\sigma\right)\ge \frac{1-\frac{5}{4}\sigma}{8n(\lambda^2+\lambda^3\sigma^2+\frac{\lambda^3\gamma^2\sigma^2}{16}+\frac{\lambda^2\gamma\sigma^2}{2})}\ge \frac{1}{8n\lambda^3(1+\lambda\sigma^2+\frac{\lambda\gamma^2\sigma^2}{16}+\frac{\gamma\sigma^2}{2})}=\frac{\beta}{n}\,,\label{eq:step4}
\end{flalign}
where 
\begin{equation}\label{eqn:pdbeta1}
\beta=\frac{1}{8\lambda^3(1+\lambda\sigma^2+\frac{\lambda\gamma^2\sigma^2}{16}+\frac{\gamma\sigma^2}{2})}.
\end{equation}
We also note that the second inequality in eqn.~\eqref{eq:step4} holds because $1-\frac{5}{4}\sigma\ge\frac{1}{\lambda}$. Let $\mu=\mu_k$ and $\mu(\alphaused)=\mu_{k+1}$ in eqn.~\eqref{mulemma2_f};  applying eqn.~\eqref{eq:step4}, we get, $\mu_{k+1} \leq  \left(1 - \nicefrac{\beta}{\dimtwo}\right) \mu_k, \forall k \geq 0$.
%
%
Applying the above inequality recursively,  we get $\mu_{k} \leq  \left(1 - \nicefrac{\beta}{\dimtwo}\right)^k \mu_0$. Therefore, for any accuracy parameter $\epsilon\in(0,1)$, $\mu_{k}\le\epsilon\mu_0$ holds, if $\left(1 - \nicefrac{\beta}{\dimtwo}\right)^k\le \epsilon$ holds. Thus, it suffices for $k$ to be at least
\begin{flalign}
k\ge \frac{n}{\beta}\log (\nicefrac{1}{\epsilon}).
\end{flalign}
Therefore, since $\beta$, as defined in eqn.~(\ref{eqn:pdbeta1}), is a constant,  we need $\Ocal(n\log\frac{1}{\epsilon})$ iterations to satisfy $\mu_k\le \epsilon\mu_0$.
\end{proof}

\vspace{-4mm}
\textbf{Proof of Theorem~\ref{thm:2f}.} Finally, the proof of our Theorem~\ref{thm:2f} directly follows from combining Lemma~\ref{lem:conouter} and Lemma~\ref{theoremOuter_f}.

\vspace{-3mm}
\section{Infeasible IPM}\label{sxn:ipm_i}

In this section, we briefly discuss the long-step infeasible IPM using approximate solver with our sketching-based preconditioner. Recall that such algorithms can, in general, start with an initial point that is not necessarily feasible, but the initial point does need to satisfy some, more relaxed, constraints. Following the lines of~\citep{Zh94,Mon03}, let $\mathcal{S}$ be the set of feasible and optimal solutions
of the form $(\xb^*,\yb^*,\sbb^*)$ for the primal and dual problems of eqns.~\eqref{eq:primal} and~\eqref{eq:dual} and assume that $\mathcal{S}$ is not empty. Then, long-step infeasible IPMs can start with any initial point $(\xb^{0},\yb^{0},\sbb^{0})$ that satisfies $(\xb^{0},\sbb^{0}) > 0$ \textit{and} $(\xb^{0},\sbb^{0}) \geq (\xb^{*},\sbb^{*})$, for some feasible and optimal solution
$(\xb^{*},\sbb^{*})\in \mathcal{S}$. In words, the starting primal and slack variables must be strictly positive \textit{and} larger (element-wise) when compared to some feasible, optimal primal-dual solution. See Chapter 6 of \citep{wright1997primal}
for a discussion regarding why such choices of starting points 
are
are relevant to computational practice and can be identified more efficiently than feasible points.

The flexibility of infeasible IPMs comes at a cost: long-step \textit{feasible} IPMs converge in $\Ocal(n\log\nicefrac{1}{\epsilon})$  iterations, while long-step \textit{infeasible} IPMs need $\Ocal(n^2 \log\nicefrac{1}{\epsilon})$ iterations to converge~\citep{Zh94,Mon03}. Here $\epsilon$ is the accuracy of the approximate LP solution returned by the IPM.
Let
\begin{subequations}\label{eq:ipm_residuals}
\begin{flalign}
	\Ab\xb^k-\bb&= \rb_p^k, \label{eq:primalres}\\
	\Ab^\ts\yb^k+\sbb^k-\cbb &= \rb_d^k,\label{eq:dualres}
\end{flalign}
\end{subequations}
%
where $\rb_p^k \in \mathbb{R}^n$ and $\rb_d^k \in \mathbb{R}^m$ are the \textit{primal} and \textit{dual} residuals, respectively that characterize how far the iterate $(\xb^k,\yb^k,\sbb^k)$ is from being feasible. 

As long-step infeasible IPM algorithms iterate and update the primal and dual solutions, the residuals are updated as well. In case of convergence, these residuals $\rb_p^k$ and $\rb_d^k$ are reduced at the same rate as the duality measure $\mu_k$ and eventually converge to zero~\citep{wright1997primal}. Let $\rb^k = (\rb_p^k,\rb_d^k) \in \mathbb{R}^{n+m}$ be the primal and dual residual at the $k$-th iteration: it is well-known that the convergence analysis of infeasible long-step IPMs critically depends on $\rb^k$ lying on the line segment between 0 and $\rb^0$ \ie, the initial residual. Unfortunately, using approximate solvers for the normal equations violates this invariant. Similar to the feasible case, a simple solution to fix this problem by adding a perturbation vector $\vb$ to the current primal-dual solution that guarantees that the invariant is satisfied is proposed in~\citep{Mon03}. In this case, the following modified system is slightly different from eqns.~\eqref{eq:delxhat}-\eqref{eq:delshat}, as it now involves the residuals\footnote{For notational simplicity, we drop the index $k$.} $\rb_p^k$ and $\rb_d^k$: 

\begin{subequations}\label{eq:ipm_inf}
\begin{flalign}
    \Ab\Db^2\Ab^\ts\hat{\Delta\yb} &=\Ab\Sb^{-1}\vb+ \pb\,\label{eq:addl_i}\\
	\hat{\Delta\xb} &=~-\xb+\sigma\mu\Sb^{-1}\one_\dimtwo-\Db^2\hat{\Delta\sbb}-\Sb^{-1}\vb\label{eq:delxhat_i}\\
\hat{\Delta\sbb}&=~-\rb_d-\Ab^\ts\hat{\Delta\yb}\,,\label{eq:delshat_i}	
\end{flalign}
\end{subequations}
where the expression of $\pb$ is now slightly different from eqn.~\eqref{eqn:pdef} due to infeasibility and is given by, $\pb=-\rb_p-\sigma\mu\Ab\Sb^{-1}\one_\dimtwo+\Ab\xb-\Ab\Db^2\rb_d.$
%
%
Again, we use the exact same sketching-based construction of $\vb$ that provably satisfies the invariant. Next, we present our main theorem for long-step infeasible IPM:
\begin{theorem}\label{thm:1}
	Let $0 \leq \epsilon \leq 1$ be an accuracy parameter. Consider the long-step infeasible IPM Algorithm~\ref{algo:ipm_i} that solves eqn.~(\ref{eq:precond_alt}) using the CG or Chebyshev iteration of Algorithm~\ref{algo:PCG} (Section~\ref{sxn:PCG}). Assume that the iterative solver runs with accuracy parameter $\zeta = \nicefrac{1}{2}$ and iteration count
	$t = \Ocal (\log n)$. 
	%
	Then, with probability at least 0.9, the long-step infeasible IPM converges after $\Ocal(n^2 \log \nicefrac{1}{\epsilon})$ iterations.
\end{theorem}

Before presenting the infeasible IPM algorithm, we will need the following definition for the neighborhood.  The involvement of the residuals $\rb^k$ makes it different from the one in the feasible case:
$$\mathcal{N}(\gamma)=\left\{(\xb,\yb,\sbb): (\xb,\sbb)>\zero, x_i s_i\ge(1-\gamma)\mu\ \text{and}\  \frac{\|\rb\|_2}{\|\rb^0\|_2} \leq \frac{\mu}{\mu_0}\right\}.$$

Notice that Lemma~\ref{lem:v} also holds for our infeasible IPM. The only difference is the expression of the vector $\pb$ which now contains the residuals. Combining a result from~\citep{Mon03} with our preconditioner $\Qb^{\nicefrac{-1}{2}}$, we can prove that
$\|\Qb^{\nicefrac{-1}{2}}\pb\|_2= \Ocal(n)\sqrt{\mu}$. Again, it is to be noted that the above bound is worse than  Lemma~\ref{thm:boundf_f} by a factor of $\sqrt{n}$. This bound allows us to prove that if we run Algorithm~\ref{algo:PCG} for $\Ocal(\log n)$ iterations, then
$\|\vb\|_2\le\frac{\gamma\sigma}{4}\mu$. However, the extra $\sqrt{n}$ factor essentially contributes to the $\Ocal(n^2)$ iteration complexity of Algorithm~\ref{algo:ipm_i}. See Appendix~\ref{app:convergence} for details.

\begin{algorithm}[H]
	\caption{Infeasible IPM}\label{algo:ipm_i}%
			\hspace*{7mm}\textbf{Input:}
		$\Ab\in\RR{\dimone}{\dimtwo}$,  $\bb\in\R{\dimone}$, $\cbb\in\R{\dimtwo}$, $\gamma \in (0,1)$, tolerance $\epsilon> 0$, centering parameter 
		
		\hspace*{6mm}
		$\sig\in (0,\nicefrac{4}{5})$;
		
		\vspace{1mm}
		\hspace*{\algorithmicindent} \textbf{Initialize:} $k\gets 0$; initial point $(\xb^{0},\yb^{0},\sbb^{0})$;
	\begin{algorithmic}[1]

		\vspace{1mm}
		\While{$\mu_k > \epsilon$}
		
		\State Compute sketching matrix $\Wb \in \mathbb{R}^{n \times w}$ (Section~\ref{sxn:background}) with $\zeta=1/2$ and $\delta = O(n^{-2})$;
		\State Compute $\rb_p^k=\Ab\xb^k-\bb$; $\rb_d^k=\Ab^\ts\yb^k+\sbb^k-\cbb$; and $\pb^k$ from eqn.~(\ref{eq:ipm_residuals});
		\State Solve the linear system of eqn.~\eqref{eq:precond_alt} for $\zb$ using Algorithm~\ref{algo:PCG} with $\Wb$ from step (2) 
		
		\hspace*{-3mm}
		and $t=\Ocal(\log n)$. Compute $\hat{\Delta \yb}=\Qb^{\nicefrac{-1}{2}}{\zb}$;
		\State Compute $\vb$ using eqn.~\eqref{eq:compv} with $\Wb$ from step (2); $\hat{\Delta\sbb}$ using eqn.~\eqref{eq:delshat_i}; $\hat{\Delta\xb}$ using 
		
		\hspace*{-3mm}
		eqn.~\eqref{eq:delxhat_i};
		\State Compute $\alphamax = \argmax\{ \alpha \in [0,1] : (\xb^k,\yb^k,\sbb^k) + \alpha (\hat{\Delta \xb}^k,\hat{\Delta \yb}^k,\hat{\Delta \sbb}^k)  \in \neigh\}$.
		\State Compute $\alphaused = \argmin\{\alpha \in [0, \alphamax]: (\xb^k + \alpha \hat{\Delta \xb}^k)^\ts (\sbb^k + \alpha \hat{\Delta\sbb}^k)\}$.
		\State Compute $(\xb^{k+1}, \yb^{k+1}, \sbb^{k+1}) = (\xb^k,\yb^k,\sbb^k) + \alphabar (\hat{\Delta \xb}^k,\hat{\Delta \yb}^k,\hat{\Delta \sbb}^k)$; set $k \gets k + 1$;
		
		\EndWhile
	\end{algorithmic}
\end{algorithm}

Notice that as compared to the feasible IPM \ie~Algorithm~\ref{algo:iipm},  Algorithm~\ref{algo:ipm_i} needs an additional step to compute the primal and dual residuals, namely, $\rb_p$ and $\rb_d$ respectively (see Step~(3)).
However, per iteration cost of Algorithm~\ref{algo:ipm_i} is asymptotically the same as that of Algorithm~\ref{algo:iipm} (see Section~\ref{sxn:IIPM}) since computing $\rb_p$ and $\rb_d$ only involve a matrix-vector product and therefore, are dominated by the SVD of $\Ab\Db\Wb$ and the computation of the perturbation vector $\vb$. See Appendix~\ref{app:convergence} for the convergence analysis of Algorithm~\ref{algo:ipm_i}.

%% file: X_extensions.tex
\section{Extensions}\label{sxn:extensions}

We briefly discuss extensions of our work. Note that we focus only on analyzing preconditioned CG and preconditioned Chebyshev iteration due to their practical advantages over other solvers. In addition,  Chebyshev iteration also offers several advantages in a parallel environment as it does not need to evaluate communication-intensive inner products for computing the recurrence parameters. However, from a theoretical perspective,
%
in \citep{CLAD20}, we analyzed two more solvers, namely, preconditioned Richardson Iteration and the preconditioned Steepest Descent  that could replace the proposed CG or Chebyshev iteration without any loss in accuracy or any increase in the number of iterations for the long-step feasible IPM Algorithm~\ref{algo:iipm} of Section~\ref{sxn:IIPM}. 
%

Second, recall that our approach focused on full rank input matrices $\Ab \in \mathbb{R}^{m \times n}$ with $m \ll n$. Our overall approach still works if $\Ab$ is any $m \times n$ matrix that is low-rank, e.g., $\rank(\Ab)=k\ll \min\{\dimone,\dimtwo\}$. In that case, using the thin SVD of $\Ab$, we can rewrite the linear constraints as follows
%
$\Ub_\Ab\Sigmab_\Ab\Vb_\Ab^\ts\xb=\bb$,
%
where $\Ub_\Ab\in\RR{\dimone}{k}$ and $\Vb_\Ab\in\RR{\dimtwo}{k}$ are the matrices of left and right singular vectors of $\Ab$ respectively; $\Sigmab_\Ab\in\RR{k}{k}$ is the diagonal matrix with the $k$ non-zero singular values of $\Ab$ as its diagonal elements. The LP of eqn.~\eqref{eq:primal} can be restated as
\begin{flalign}
\min\,\cbb^\ts\xb\,,\text{ subject to }\Vb_\Ab^\ts\xb=\widetilde{\bb}\,,\xb\ge \zero\,,\label{eq:primal2}
\end{flalign}
where $\widetilde{\bb}=\Sigmab_\Ab^{-1}\Ub_\Ab^\ts\bb$. Note that, $\rank(\Vb_\Ab)=k\ll \dimtwo$ and therefore eqn.~\eqref{eq:primal2} can be solved using our framework. The matrices $\Ub_\Ab$, $\Vb_\Ab$, and $\Sigmab_\Ab$ can be approximately recovered using the fast SVD algorithms of \citep{Halko2011,BouDriMag14,clarkson2017low}. However, the accuracy of the final solution will depend on the accuracy of the approximate SVD and we defer this analysis to future work.

Third, even though we chose to use the Count-Min sketch and its analysis from~\citep{Cohen2016} (Section~\ref{sxn:background}), there are many other alternative sketching matrix constructions that would lead to similar results. A particularly simple one is the Gaussian sketching matrix $\Wb_G \in \mathbb{R}^{n \times w}$, where every entry is a $\mathcal{N}(0,1)$ random variable. Setting $w=\Ocal\left(\nicefrac{m+\log (1/ \delta)}{\zeta^{2}}\right)$
would result in the same accuracy guarantees as the sketching matrix of Section~\ref{sxn:background}. However, the (theoretical) running time needed to compute $\Ab \Db \Wb$ increases to $\Ocal (m \cdot\nnz{\Ab} )$. In practice, at least for relatively small matrices, using Gaussian sketching matrices is a reasonable alternative; see the discussion in~\citep{Meng2014SISC} which argued that the Gaussian matrix sketching-based solvers are considerably better than direct solvers. We also opted to use Gaussian matrices in our empirical evaluation, since we primarily interested in measuring the accuracy of the final solution as a function of the number of iterations of the solver and the IPM algorithm. Other known constructions of sketching matrices that are also applicable in our setting include (any) sub-gaussian sketching matrix; the Subsampled Randomized Hadamard transform (SRHT); and any of the Sparse Subspace Embeddings of~\citep{clarkson2017low,nelson2013osnap,meng2013low,cohen2016nearly}.

%% file: 5_Experiments.tex
\begin{figure}[t]
	\centering
	\subfigure{
		\label{fig:ARCENE_v1} 
		\includegraphics[width=2.8in]{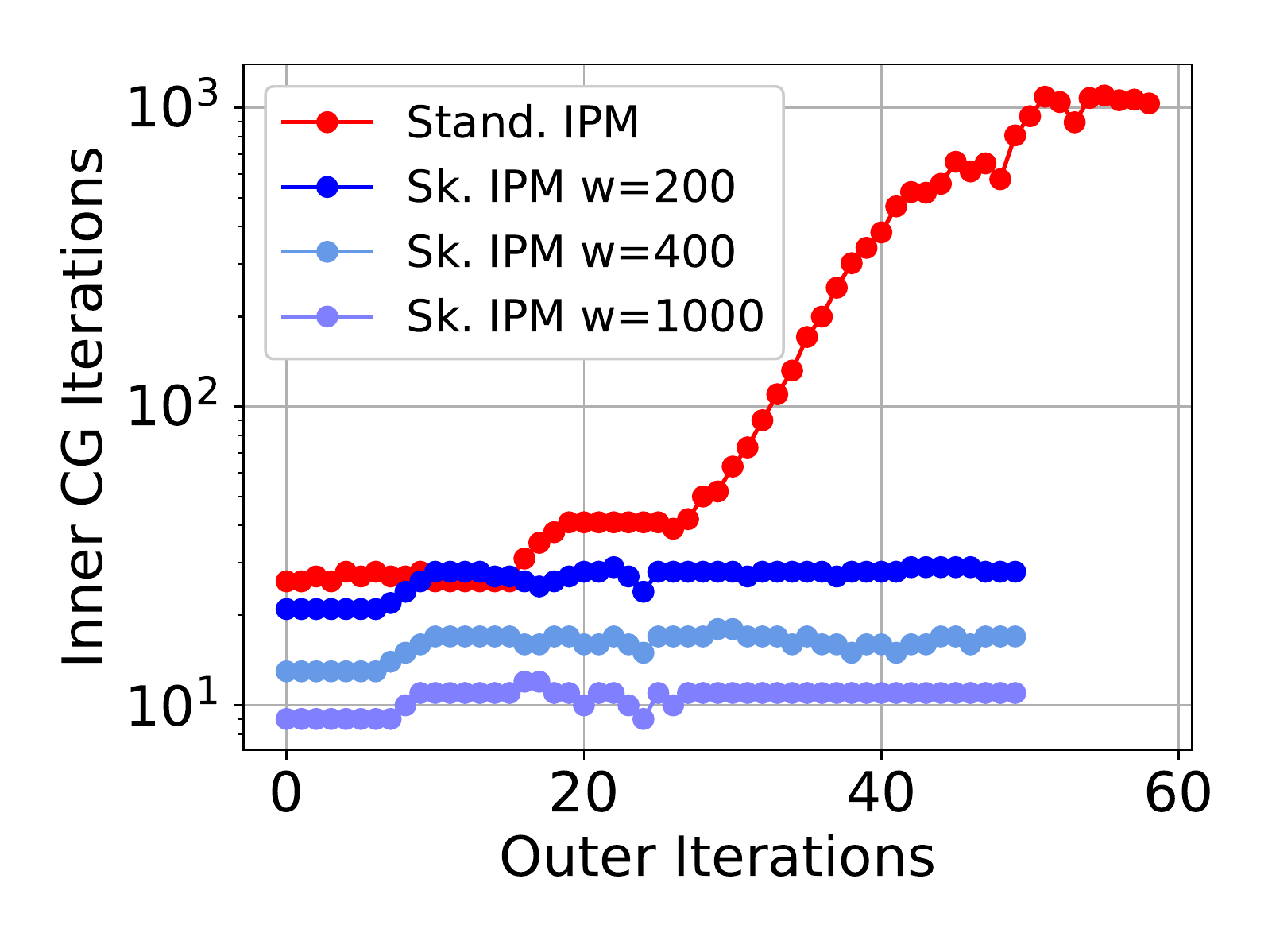}}
	\subfigure{
		\label{fig:ARCENE_v2} 
		\includegraphics[width=2.8in]{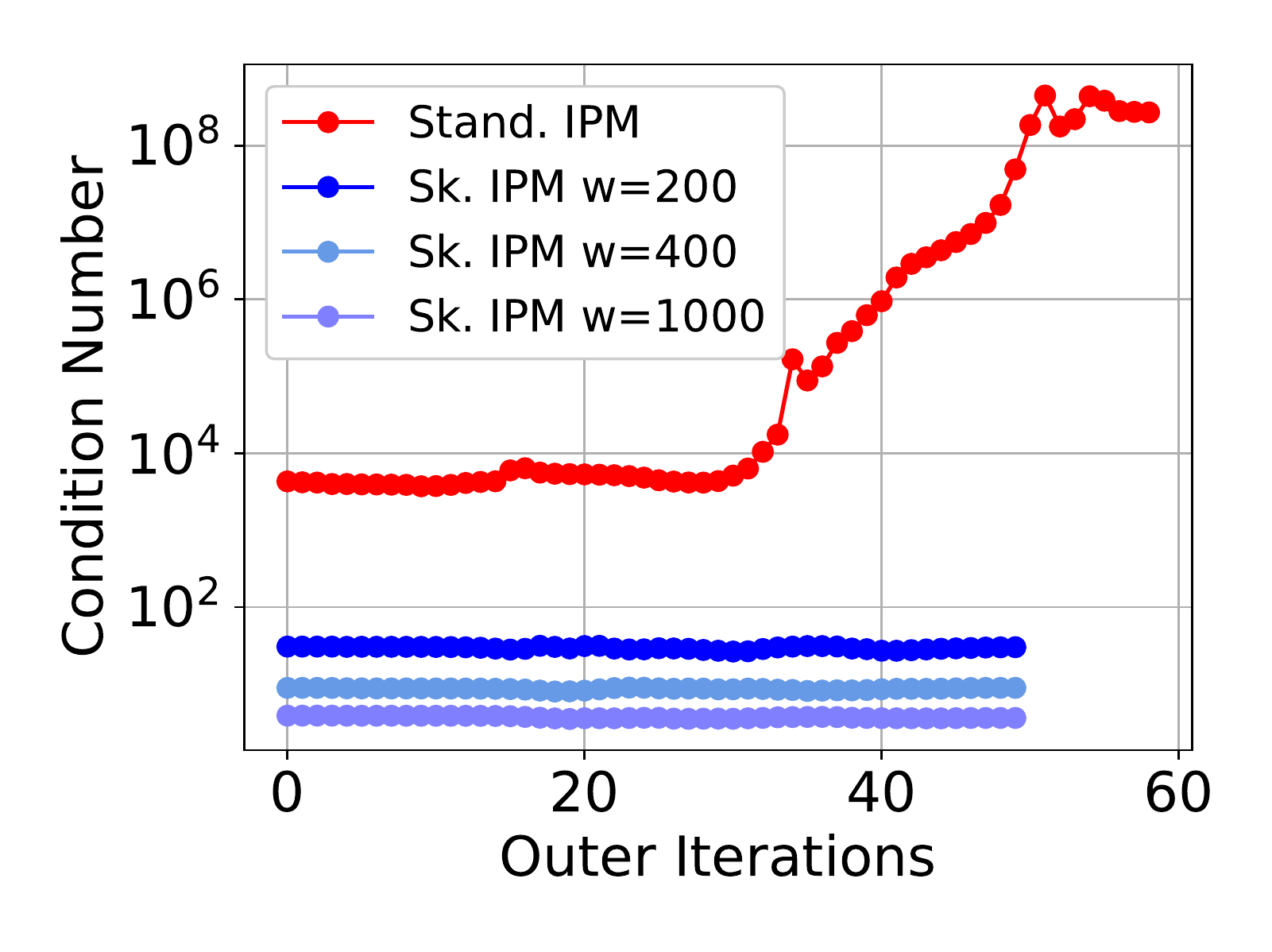}}
	\label{fig:iter_ARCENE_sub}
\vspace{-5pt}
\\
\addtocounter{subfigure}{-2}
\subfigure[ ]{
		\label{fig:ARCENE_v3} 
		\includegraphics[width=2.8in]{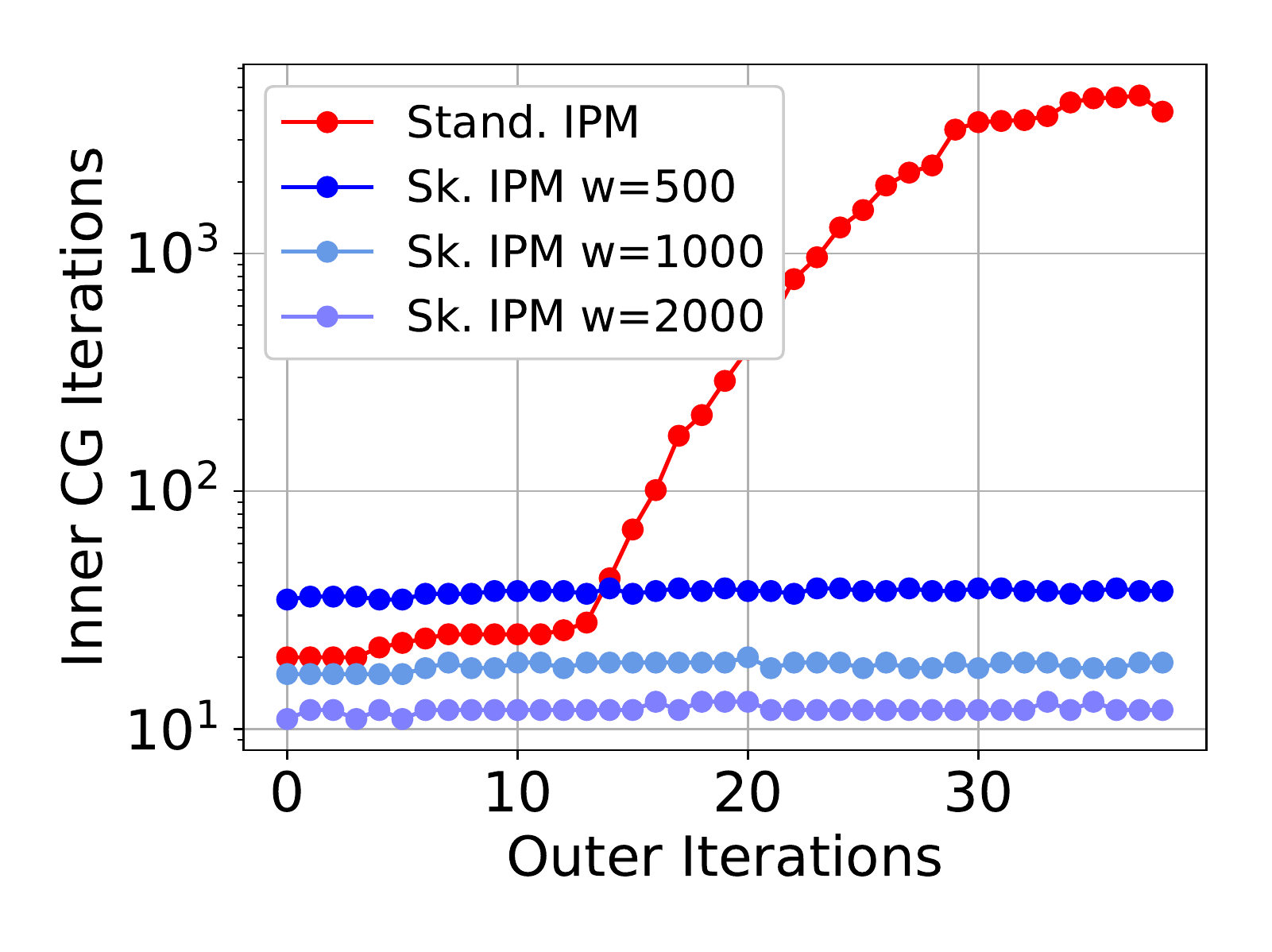}}
	\subfigure[]{
		\label{fig:ARCENE_v4} 
		\includegraphics[width=2.8in]{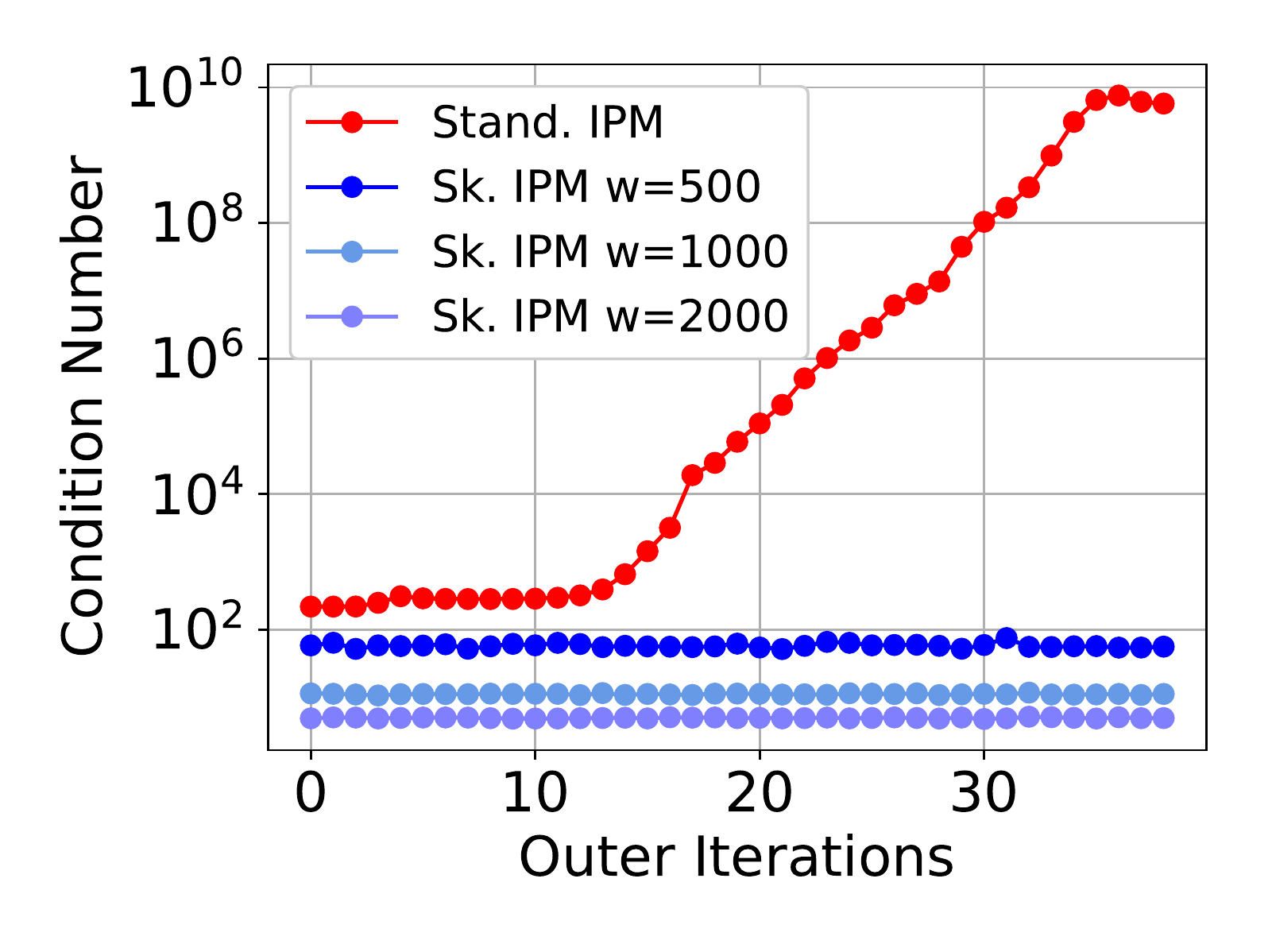}}
	\caption{\emph{ARCENE (top row) and DEXTER (bottom row) data sets}: Our algorithm (Sk. IPM) requires an order of magnitude fewer inner iterations than the Standard IPM with CG at each outer iteration, as demonstrated in (a). This is possibly due to the improved conditioning of \precNormal~compared to $\Ab \Db^2 \Ab^T$, as shown in (b).
		For all experiments \tolCG~$ = 10^{-5}$ and \tolOuterRes~$ = 10^{-9}$.}
	\label{fig:iter_ARCENE}

\end{figure}

\begin{figure}[t]
	\centering
	\subfigure[Max. Inner CG Iterations.]{
		\label{fig:ARCENE_iter} 
		\includegraphics[width=2.57in]{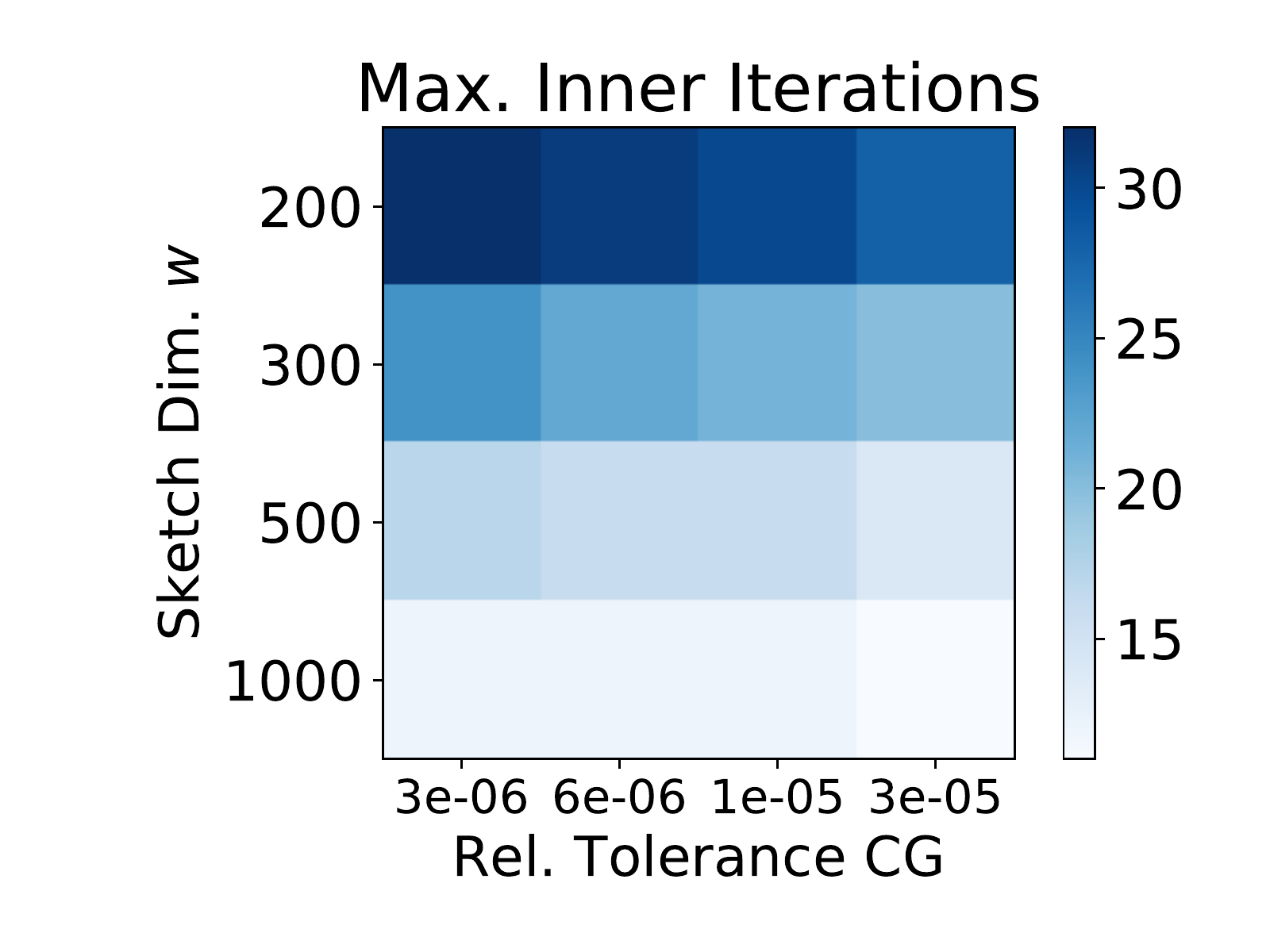}}
	\subfigure[Max. Condition Number.]{
		\label{fig:ARCENE_iter2} 
		\includegraphics[width=2.57in]{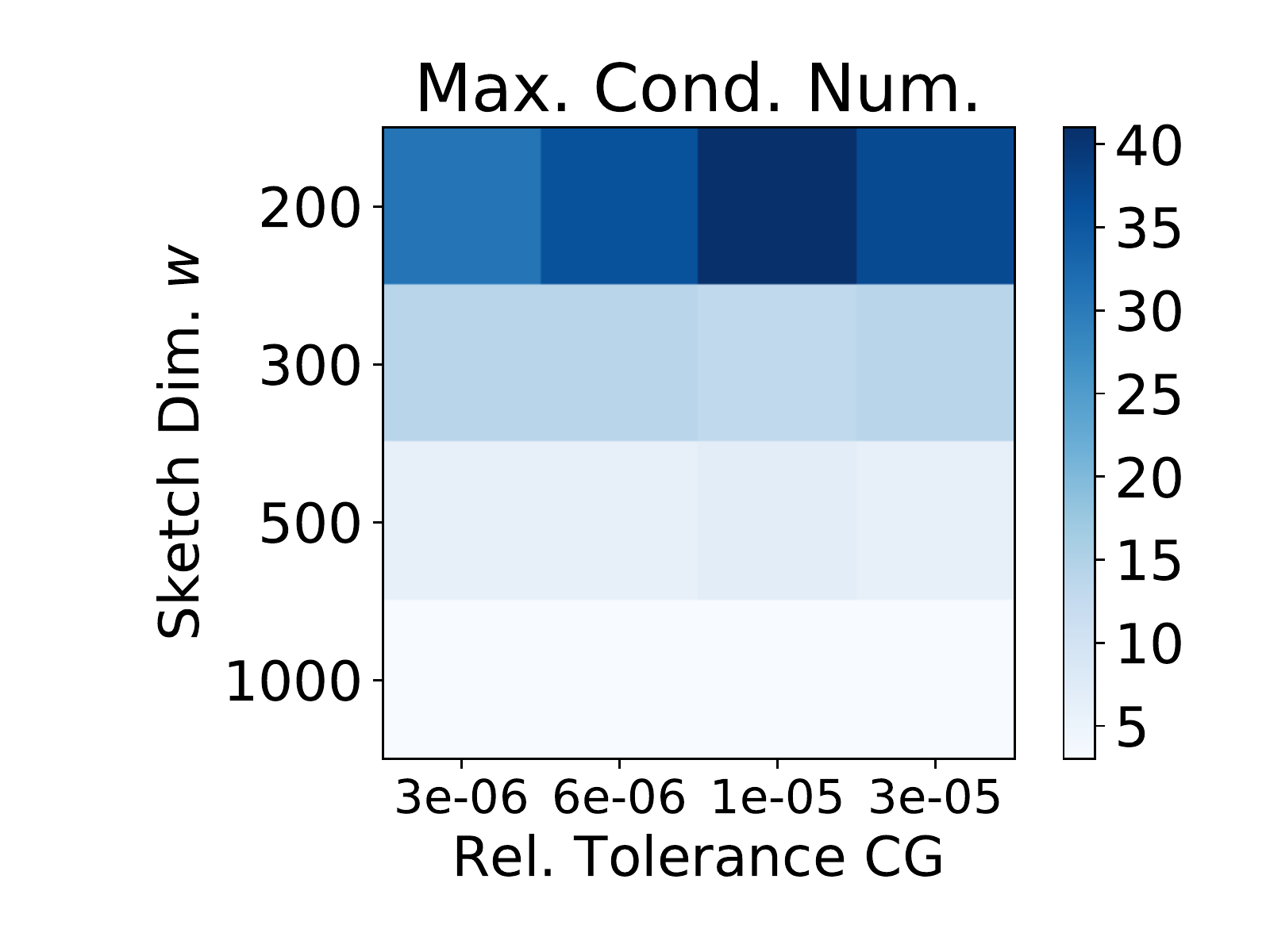}}
	\caption{\emph{ARCENE data set}: for various ($w$, \tolCG) settings,
		(a)~the maximum number of inner iterations used by our algorithm and (b)~the maximum condition number of \precNormal, across outer iterations. The standard IPM, across all settings, needed on the order of 1,000 iterations and  $\kappa(\Ab \Db^2\Ab^T)$ was on the order of $10^{8}$.
		The relative error was fixed to $0.04\%$. }
	\label{fig:ARCENE} 
\end{figure}

\section{Experiments}\label{sec:exp}

Here we demonstrate the empirical performance of our algorithm on a variety of real-world data sets from the UCI ML Repository~\citep{Dua2019}. More specifically, we consider two problems that were part of the NeurIPS 2003 feature selection challenge: ARCENE and DEXTER~\citep{guyon2005result}. For the ARCENE data set, the task is to distinguish between cancer and normal patterns from mass-spectrometric data and DEXTER data set is for a text classification problem. Further, we consider DrivFace~\citep{diaz2016reduced}, a problem concerned with identifying the gaze direction in photos of human subjects taken while driving,
and a gene expression cancer RNA-Sequencing data set, accessible on the UCI ML Repository, which is part of the RNA-Seq (HiSeq) PANCAN data set \citep{Weinstein2013}. It is a random extraction of gene expressions from patients who have different types of tumors: BRCA, KIRC, COAD, LUAD and PRAD. We considered the binary classification task of identifying BRCA versus other types.
We also perform experiments on synthetic data sets (see Appendix~\ref{app:rand} for details).
The experiments were implemented in Python and we observed that the results for both synthetic data (generated as described in Appendix~\ref{app:rand}) and real-world data were qualitatively similar. 
Thus, we highlight results on several representative real-world datasets. 
The experiments were implemented in Python and run on a server with Intel E5-2623V3@3.0GHz 8 cores and 64GB RAM. 

As an application, we consider $\ell_1$-regularized SVMs. All of the data sets are concerned with binary classification with $m \ll n$, where $n$ is the number of features.
The SVM problem is a core model in machine learning that is crucial for applications in both regression and classification.
While there are many variations of SVMs, we use the classical version of SVMs with an $\ell_1$ regularizer to illustrate the application of our algorithm.
In Appendix~\ref{app:svm}, we describe the $\ell_1$-SVM problem and how it can be formulated as an LP. Here, $m$ is the number of training points, $n$ is the feature dimension, and the size of the constraint matrix in the LP becomes $m \times (2n +1)$.
%


\vspace{0.02in}\noindent\textbf{Comparisons and Metrics}. \textcolor{black}{Our empirical evaluations serve as a proof-of-concept verifying our theoretical findings, by evaluating the effectiveness of our randomized preconditioner combined with an approximate solver. State-of-the-art implementations of LP solvers are highly optimized; therefore, it is unlikely to get a fair time-comparison between our algorithm and industrial-grade solvers, since the true algorithmic efficiency of commercial solvers is confounded by the built-in optimization strategies. We do not report running times to avoid such direct comparisons with heavily optimized benchmark LP solvers.}

In most of our evaluations, we use the infeasible case \ie,~Algorithm~\ref{algo:ipm_i} as finding a strictly feasible staring point is a non-trivial task. In addition, we focus on CG iterative solver to compute the approximate search directions. We compare Algorithm~\ref{algo:ipm_i} with a standard IPM (see \citep[Ch.~10]{NumericalRecipes}) using CG, and a standard IPM using a direct solver.
We also use CVXPY as a benchmark to compare the accuracy of the solutions; we define the \emph{relative error} $\nicefrac{\|\hat{\xb} - \xb^\star\|_2}{\|\xb^\star\|_2}$, where $\hat{\xb}$ is our solution and $\xb^\star$ is the solution generated by CVXPY. \textcolor{black}{In addition, in some of the experiments, we also consider the primal-dual error $\cbb^\ts\xb-\bb^\ts\yb$ as a key metric to evaluate the quality of the solution returned by our algorithm.} We also consider the number of \emph{outer iterations}, namely the number of iterations of the IIPM algorithm, as well as the  number of \emph{inner iterations}, namely the number of iterations of the CG solver.
The \emph{inner iterations} is highly dependent on the condition number of  the matrices in the normal equations (eqn.~\eqref{eq:normal} or \eqref{eq:precond}), which we also report.
We denote the relative stopping tolerance for CG by \tolCG~ and we denote the
outer iteration residual error by \tolOuterRes. If not specified: \tolOuterRes~$= 10^{-9}$, \tolCG~$ = 10^{-5}$, and $\sigma = 0.5$.
We evaluated a Gaussian sketching matrix, and the initial triplet $(\xb, \yb, \sbb)$ for all IPM algorithms was set to be all ones.

\begin{table*}[ht]
	\caption{Comparison of (our) sketched IPM with CG, standard IPM with CG, and Standard IPM with a direct solver, for the $\ell_1$-SVM problem on UCI Machine Learning Repository~\citep{Dua2019} data sets. Across all, $\tau = 10^{-9}$ and a relative error of $10^{-3}$ or less was achieved. We define $\kappa_{\text{Sk}} = $ \kprecNormal~and $\kappa_{\text{Stan}} = \kappa(\Ab \Db^2\Ab^T)$.} \label{tableSVM}
	\begin{center}
		\resizebox{\textwidth}{!}{\begin{tabular}{|c|c|c|c|c|c|c|c|c|c|}
			\hline
			\multicolumn{1}{|c|}{\textbf{Problem}} &
			\multicolumn{1}{|c|}{\textbf{Size}} &
			\multicolumn{4}{|c|}{\textbf{Sketch IPM w/ Precond. CG}} &
			\multicolumn{3}{|c|}{\textbf{Stand. IPM w/ Unprec. CG}} &
			\multicolumn{1}{|c|}{\textbf{IPM w/ Dir.}} \\
			& \multicolumn{1}{c}{$(m \times N )$} & \multicolumn{1}{|c}{$w$} & \multicolumn{1}{c}{In. It.}  & \multicolumn{1}{c}{Out. It.} & \multicolumn{1}{c|}{$\kappa_{\text{Sk}}$}
			& \multicolumn{1}{c}{In. It.}  & \multicolumn{1}{c}{Out. It.} & \multicolumn{1}{c|}{$\kappa_{\text{Stan}}$}  & \multicolumn{1}{c|}{Out. It.}\\
			\hline
			ARCENE &$(100 \times 10K)$ & $200$ & $\mathbf{30}$ & $50$ & $38.09$
			& $\mathbf{1.1 K}$ & $59$ & $4.4 \times 10^{8}$ & $50$\\
			DEXTER &$(300 \times 20K)$ & $500$ & $\mathbf{39}$ & $39$ & $75.42$
			& $\mathbf{4.6K}$ & $39$ & $7.6 \times 10^{9}$ & $39$\\
			
			DrivFace &$(606 \times 6400)$ & $1000$ & $\mathbf{50}$ & $42$ & $68.87$
			& $\mathbf{139K}$ & $43$ & $17 \times 10^{12}$ & $42$\\
			Gene RNA &$(801 \times 20531)$ & $2000$ & $\mathbf{27}$ & $44$ & $20.03$
			& $\mathbf{101K}$ & $208$ & $4.7 \times 10^{12}$ & $44$\\
			\hline
		\end{tabular}}
	\end{center}
	\label{tab:multicol}
\end{table*}



\vspace{0.02in}\noindent\textbf{Experimental Results}.
Figure~\ref{fig:iter_ARCENE}(a) shows that our Algorithm~\ref{algo:iipm} uses an order of magnitude fewer \textit{inner} iterations than the un-preconditioned standard solver. This is due to the improved conditioning of the respective matrices in the normal equations, as demonstrated in Figure~\ref{fig:iter_ARCENE}(b).  Across various real-world and synthetic data sets, the results were qualitatively similar to those shown in Figure~\ref{fig:iter_ARCENE}. Results for several real-world data sets are summarized in Table~\ref{tableSVM}.

In general, our preconditioned CG solver used in Algorithm~\ref{algo:iipm} does not increase the total number of \textit{outer} iterations as compared to the standard IPM with CG, and the standard IPM with a direct linear solver (denoted IPM w/Dir), as seen in Table~\ref{tableSVM}. Actually, for unpreconditioned CG there is clearly more outer iterations, especially for Gene RNA, which has x5 outer iterations.
Figure~\ref{fig:iter_ARCENE} also demonstrates the relative insensitivity to the choice of $w$ (the sketching dimension, i.e., the number of columns of the sketching matrix $\Wb$ of Section~\ref{sxn:background}). For smaller values of $w$, our algorithm requires more inner iterations. However, across various choices of $w$, the number of inner iterations is always an order of magnitude smaller than the number required by the standard solver.

Figure~\ref{fig:ARCENE} shows the performance of our algorithm for a range of ($w$, \tolCG) pairs.
Figure~\ref{fig:ARCENE}(a) demonstrates that the number of the inner iterations is robust to the choice of \tolCG~ and $w$. The number of inner iterations varies between $15$ and $35$ for the ARCENE data set, while the standard IPM took on the order of $1,000$ iterations across all parameter settings. Across all settings, the relative error was fixed at $0.04\%$. In general, our sketched IPM is able to produce an extremely high accuracy solution across parameter settings. Thus we do not report additional numerical results for the relative error, which was consistently $10^{-3}$ or less.
Figure~\ref{fig:ARCENE}(b) demonstrates a trade-off of our approach: as both \tolCG~ and $w$ are increased, the condition number \kprecNormal~decreases, corresponding to better conditioned systems. As a result,  fewer inner iterations are required. 
In this context, Figure~\ref{fig:ARCENE_more} shows that how the number of inner CG iterations (Figure~\ref{fig:ARCENE_v_more1}) or the condition number of $\Qb^{-1/2}\Ab\Db^2\Ab^\ts\Qb^{-1/2}$ (Figure~\ref{fig:ARCENE_v_more2}) decreases with the increase in sketching dimension $w$ for various \tolCG\,.

\begin{figure}[t]
	\centering
	\subfigure[ ]{
		\label{fig:ARCENE_v_more1} 
		\includegraphics[width=2.8in]{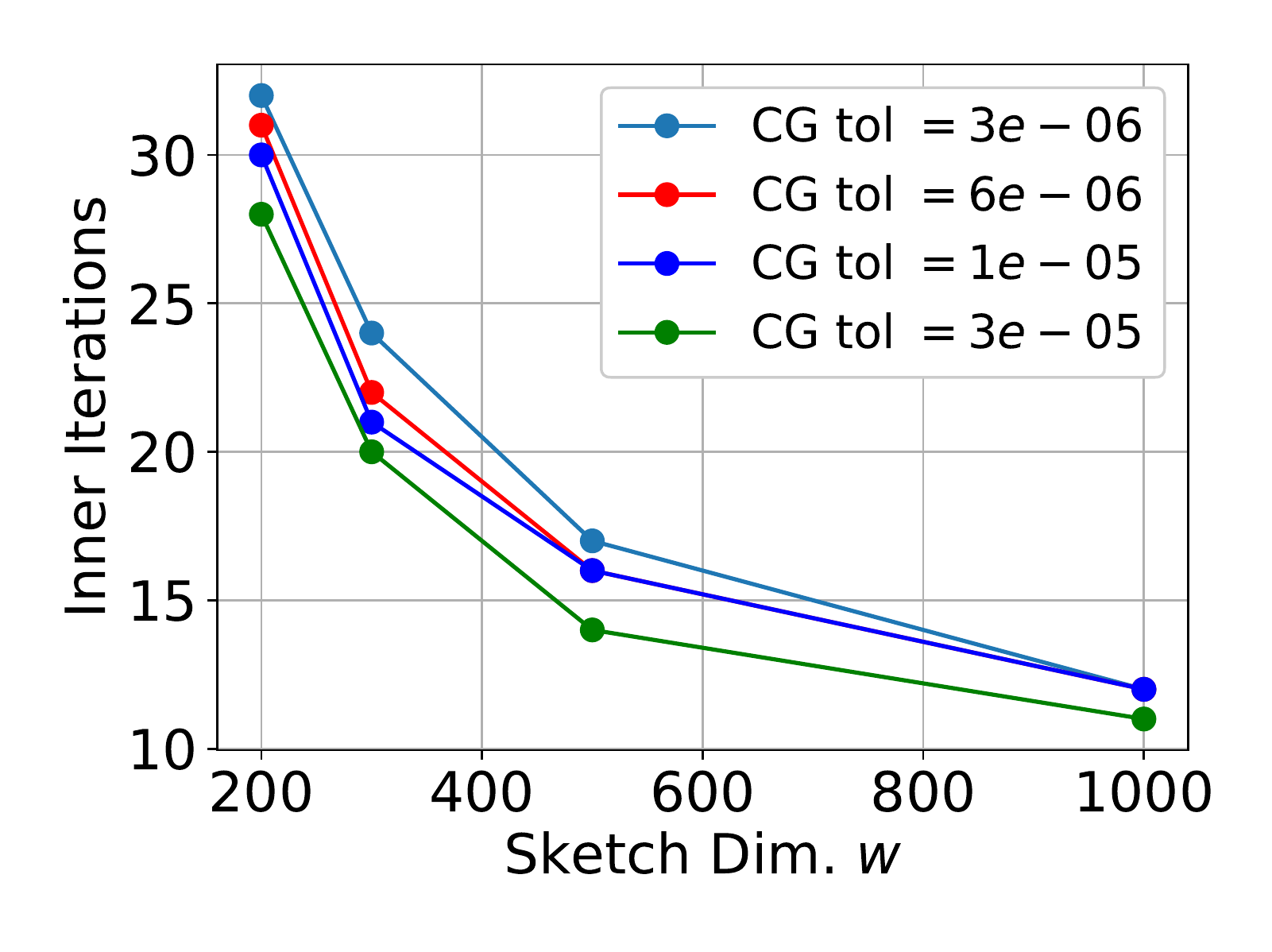}}
	\subfigure[]{
		\label{fig:ARCENE_v_more2} 
		\includegraphics[width=2.8in]{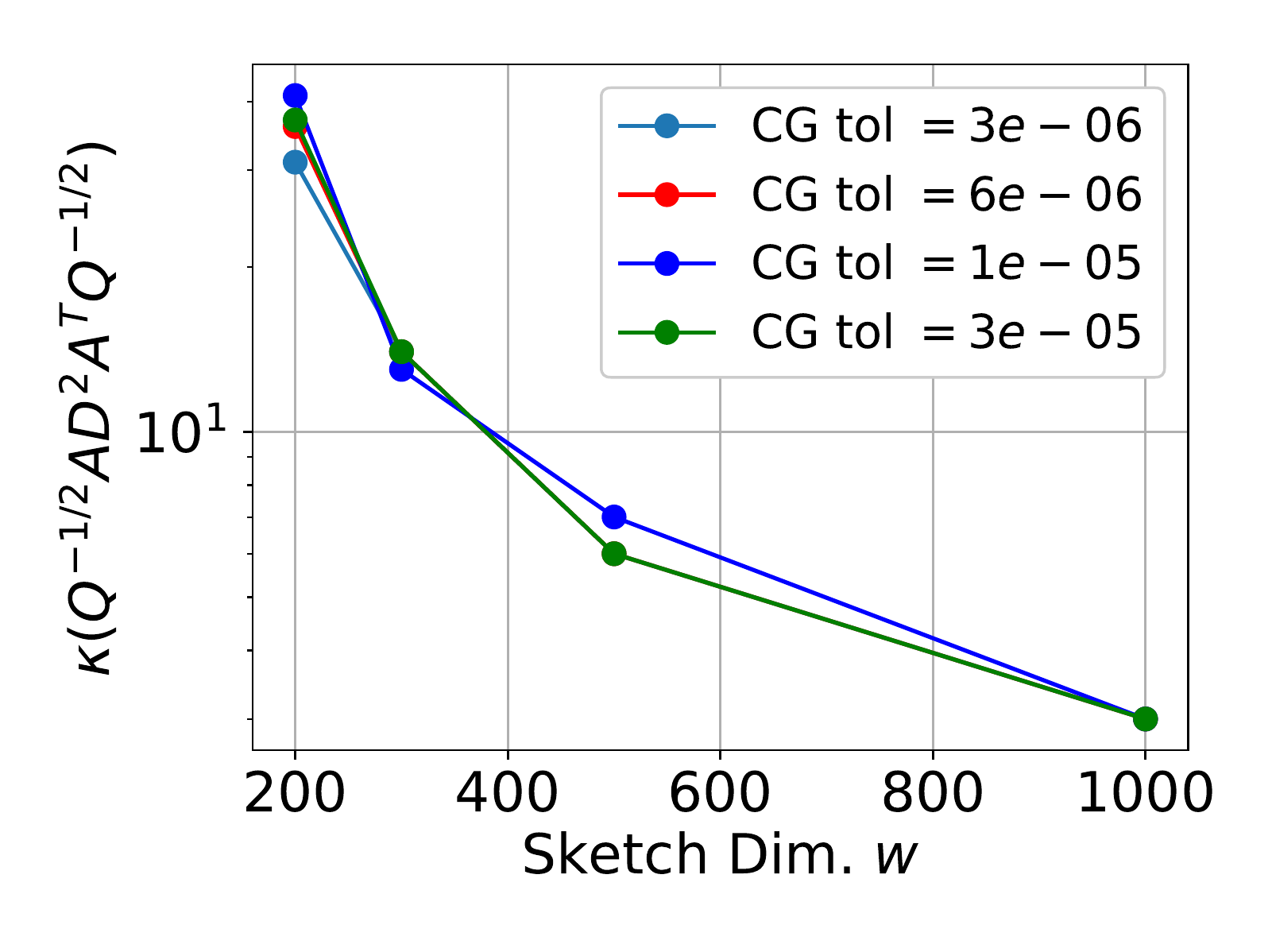}}
	\caption{\emph{ARCENE data set}: As $w$ increases, (a) the number of inner iterations decreases, and is relatively robust to \tolCG and (b) the condition number decreases as well.
	}
	\label{fig:ARCENE_more}
\end{figure}

\begin{figure}[t]
	\centering
    \includegraphics[width=15cm,height=15cm,keepaspectratio]{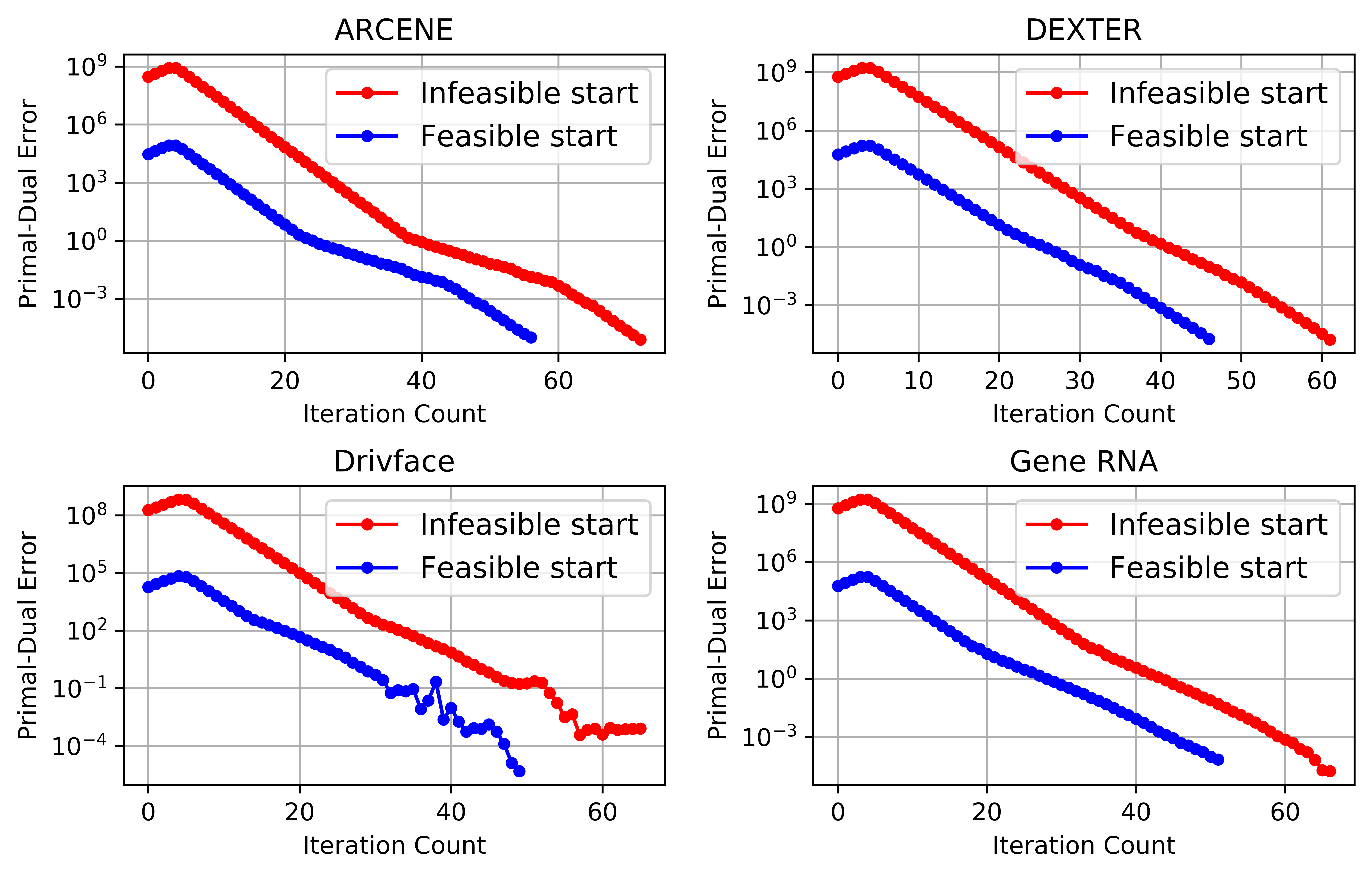}
	\caption{\emph{Feasible vs. infeasible start}: For all four data sets, we see that the algorithm takes much fewer number of outer iteration if one can start from a feasible point. Here, we take $w=1000$. Primal-dual errors are in log scale.
	}
	\label{fig:starting_point}
\end{figure}

\textcolor{black}{
Next, we further evaluate the performance of our algorithm in terms of the total number of \emph{outer iterations} when starting from an infeasible point vs. starting from a strictly feasible point. The objective is to study the effect of feasibility in convergence of our algorithms. For the feasible IPM, we assume that we already have a strictly feasible starting point. See Appendix~\ref{app:feasible_start} on how to find a strictly feasible point of an LP. Figure~\ref{fig:starting_point} shows that if we already know a feasible starting point beforehand, then, across all the data sets, our algorithm indeed takes much fewer number of outer iterations as compared to that of infeasible-start IPM. Additionally, starting from a feasible or an infeasible point does not seem to affect the rate with which the primal-dual error decreases
}

Finally, we run our IPM solver without a \emph{v-correction} \ie, without using the perturbation vector $\vb$ in step~(5) of Algorithm~\ref{algo:ipm_i} and notice that our algorithm still converges without significantly changing the inner or outer iteration counts (see Table \ref{tableNoCorrection}). We leave the corresponding theoretical analysis for future work.  We use the same tolerance parameters and sketching dimension as in Table~\ref{tab:multicol}.

\begin{table*}[ht]
	\caption{Comparison of (our) sketched IPM with and without correction, for the $\ell_1$-SVM problem on UCI Machine Learning Repository~\citep{Dua2019} data sets. Across all, $\tau = 10^{-9}$ and a relative error of $10^{-3}$ or less was achieved.} \label{tableNoCorrection}
	\begin{center}
		\resizebox{\textwidth}{!}{\begin{tabular}{|c|c|c|c|c|c|c|c|c|}
			\hline
			\multicolumn{1}{|c|}{\textbf{Problem}} &
			\multicolumn{1}{|c|}{\textbf{Size}} &
			\multicolumn{1}{|c|}{\textbf{Sketch Size}} &
			\multicolumn{3}{|c|}{\textbf{Precond. Sketch IPM without Correction}} &
			\multicolumn{3}{|c|}{\textbf{Precond. Sketch IPM with Correction}}  \\
			& \multicolumn{1}{c}{$(m \times N )$} & \multicolumn{1}{|c|}{$w$} & \multicolumn{1}{c}{Max In. It.}  & \multicolumn{1}{c}{Sum In. It.} & \multicolumn{1}{c|}{Out. It.}
			& \multicolumn{1}{c}{Max In. It.}  & \multicolumn{1}{c}{Sum In. It.} & \multicolumn{1}{c|}{Out. It.}  \\
			\hline
			ARCENE &$(100 \times 10K)$ & $200$ & $29$ & $1868$ & $73$
			& $29$ & $1873$ & $73$\\
			DEXTER &$(300 \times 20K)$ & $500$ & $40$ & $2307$ & $62$
			& $40$ & $2271$ & $62$\\
			DrivFace &$(606 \times 6400)$ & $1000$ & $52$ & $2820$ & $66$
			& $50$ & $2804$ & $66$\\
			Gene RNA &$(801 \times 20531)$ & $2000$ & $27$ & $1445$ & $67$
			& $27$ & $1434$ & $67$\\
			\hline
		\end{tabular}}
	\end{center}
	\label{tab:multicol2}
\end{table*}

%% file: 6_Conclusion.tex
\section{Conclusions and Open Problems}
We proposed and analyzed a long-step IPM algorithm (both feasible and infeasible) using a preconditioned conjugate gradient solver for the normal equations and a novel perturbation vector to correct for the error due to the approximate solver. Thus, we speed up each iteration of the IPM algorithm, without increasing the overall number of iterations.  We demonstrate empirically that our IPM requires an order of magnitude fewer inner iterations within each linear solve than standard IPMs. 

Several important questions remain open. First of all, from a theoretical perspective, using the vector $\vb$ to correct for infeasibility was necessary for our theoretical analysis, but, from an empirical perspective, we observed that the correction was not needed. A theoretical analysis of long-step IPMs without a correction vector would be of interest. Second, it would be interesting to explore whether there are other ways to use the preconditioner to design a feasible step instead of the $\vb$-correction. \textcolor{black}{Third}, a thorough empirical evaluation of the effect of preconditioning and approximate solvers, with or without the $\vb$-correction, would be a significant undertaking in future work. \textcolor{black}{Finally, it would be interesting to investigate what other theoretically impactful ideas could be used to efficiently solve linear program in practice. There exists barriers to using methods such as \textit{inverse maintenance} and \textit{lazy updates} in practice, as discussed in Section \ref{sxn:comparison}. However, it is unknown whether these issues are fundamental or avoidable.}



%% file: A_Appendix.tex
\section{Convergence analysis of Algorithm~\ref{algo:ipm_i}}\label{app:convergence}
\textcolor{black}{The proofs of long-step feasible IPM and long-step infeasible IPM are different from each other, since the latter needs additional assumptions on the initial iterate $(\xb^0,\yb^0,\sbb^0)$. For a detailed comparison between proofs of these two variants, we refer the readers to Chapters~5 and 6 of ~\citep{wright1997primal}.
In the context of an approximate solver, most of the proofs related to the convergence of the long-step infeasible IPMs followed from~\citep{Mon03}, except for the fact that we used our sketching-based preconditioner $\Qb^{-1/2}$, as well as our choice of the vector $\vb$ that corrects for the error caused by the inexact solver. Here, we only prove results that are different from~\citep{Mon03}. For our feasible IPM proofs, there is no prior work that analyzed the theoretical aspects of long-step feasible IPMs with an approximate solver. Therefore, compared to the prototypical, long-step feasible IPM of~\citep{wright1997primal}, our proofs needed extra care in bounding the duality gap decrease in each iteration when the linear system is only approximately solved.}



\subsection{Number of iterations for the iterative solver}
In this section, most of the proofs follow \citep{Mon03} except for the fact that we used our sketching based preconditioner $\Qb^{\nicefrac{-1}{2}}$. Recall that $\mathcal{S}$ is the set of optimal and feasible solutions for the proposed LP.

\begin{lemma}\label{lem:ineq}
	Let $(\xb^{0},\yb^{0},\sbb^{0})$ be the initial point with $(\xb^{0},\sbb^{0})>\zero$ and $(\xb^{*},\yb^{*},\sbb^{*})\in\mathcal{S}$ such that $(\xb^{*},\sbb^{*})\le(\xb^{0},\sbb^{0})$ with $\sbb^{0}\ge|\Ab^\ts\yb^{0}-\cbb|$. Then, for any point $(\xb,\yb,\sbb)\in\mathcal{N}(\gamma)$ such that $\rb=\eta\,\rb^{0}$ and $0\le\eta\le\min\left\{1,\frac{\sbb^\ts\xb}{\sbb^{0\ts}\xb^{0}}\right\}$, we get
	\begin{subequations}
		\begin{flalign}
		&(i)~~\eta\,(\xb^\ts\sbb^{0}+\sbb^\ts\xb^{0})\le\,3\dimtwo\mu\label{eq:ineq}\,,\\
		&(ii)~~\eta\,\|\Sb(\xb^{*}-\xb^{0})\|_2\le\eta\,\|\Sb\xb^{0}\|_2\le\eta\sbb^{\ts}\xb^{0}\le\,3\dimtwo\mu\label{eq:ineq2}\,,\\
		&(iii)~~\eta\,\|\Xb(\sbb^{0}+\Ab^\ts\yb^{0}-\cbb)\|_2\le~2\eta\,\|\Xb\sbb^{0}\|_2\le~2\eta\,\xb^\ts\sbb^{0}\le~6\dimtwo\mu\label{eq:ineq3}\,.
		\end{flalign}
	\end{subequations}
\end{lemma}
\begin{proof}
	We prove eqns.~\eqref{eq:ineq}--\eqref{eq:ineq3} below.
	
\noindent\textbf{Proof of eqn.~\eqref{eq:ineq}.}
	For completeness, we provide a proof of eqn.~\eqref{eq:ineq} following~\citep{Mon03}. Since $(\xb^{*}, \sbb^{*}, \yb^{*}) \in \mathcal{S}$, the following equalities hold:
	\begin{subequations}
		\begin{flalign}
		\Ab\xb^{*} &=\bb\label{eq:cond1}\\
		\Ab^\ts\yb^{*}+\sbb^{*} &=\cbb.\label{eq:cond2}
		\end{flalign}
	\end{subequations}
	
	\noindent Furthermore, $\rb=\eta \rb^{0}$ implies
	\begin{subequations}
		\begin{flalign}
		\Ab\xb-\bb &=\eta(\Ab\xb^{0}-b)\label{eq:cond3}\\
		\Ab^\ts\yb+\sbb-\cbb &=\eta(\Ab^\ts\yb^{0}+\sbb^{0}-\cbb).\label{eq:cond4}
		\end{flalign}
	\end{subequations}
	
	\noindent Combining eqn.~\eqref{eq:cond1} with eqn.~\eqref{eq:cond3} and eqn.~\eqref{eq:cond2} with eqn.~\eqref{eq:cond4}, we get
	\begin{subequations}
		\begin{flalign}
		\Ab\big(\xb-\eta\xb^{0}-(1-\eta)\xb^{*}\big) &=\zero\label{eq:cond5}\\
		\Ab^\ts(\yb-\eta\yb^{0}-(1-\eta)\yb^{*})+(\sbb-\eta\sbb^{0}-(1-\eta)\sbb^{*}) &=\zero.\label{eq:cond6}
		\end{flalign}
	\end{subequations}
	Multiplying eqn.~\eqref{eq:cond6} by $\left(\xb-\eta \xb^{0}-(1-\eta)\xb^{*}\right)^\ts$ on the left and using eqn.~\eqref{eq:cond5},
	we get
	$$
	\left(\xb-\eta \xb^{0}-(1-\eta)\xb^{*}\right)^\ts\left(\sbb-\eta\sbb^{0}-(1-\eta)\sbb^{*}\right)=0.
	$$
	Expanding we get
	\begin{flalign}
	&\eta\left(\xb^{0^{\ts}}\sbb+\xb^{\ts}\sbb^{0}\right)=\eta^{2} \xb^{0^{\ts}}\sbb^{0}+(1-\eta)^{2} (\xb^{*})^{\ts} \sbb^{*}+\xb^{\ts}\sbb\nonumber\\
	&~~~~~~~~~~~~~~~~~~~~~~~~~~~~~+\eta(1-\eta)\left(\xb^{0^{\ts}}\sbb^{*}+(\xb^{*})^{\ts} \sbb^{0}\right)-(1-\eta)\left((\xb^{*})^{\ts}\sbb+\xb^{\ts} \sbb^{*}\right).\label{eq:cond7}
	\end{flalign}
	Next, we use the given conditions and rewrite eqn.~\eqref{eq:cond7} as
	\begin{flalign}
	\eta\left(\xb^{0^{\ts}} \sbb+\sbb^{0^{\ts}} \xb\right) & \leq \eta^{2} \xb^{0^{\ts}} \sbb^{0}+\xb^{\ts} \sbb+\eta(1-\eta)\left(\xb^{0^{\ts}} \sbb^{*}+\sbb^{0^{\ts}} \xb^{*}\right) \nonumber\\
	& \leq \eta^{2} \xb^{0^{\ts}} \sbb^{0}+\xb^{\ts} \sbb+2 \eta(1-\eta) \xb^{0^{\ts}} \sbb^{0} \nonumber\\
	& \leq 2 \eta \xb^{0^{\ts}} \sbb^{0}+\xb^{\ts} \sbb \leq 3 \xb^{\ts} \sbb=3\dimtwo\mu.\label{eq:fin_bound}
	\end{flalign}
	The first inequality in eqn.~\eqref{eq:fin_bound} follows from the following facts. First, $(1-\eta)((\xb^{*})^{\ts}\sbb+\xb^{\ts} \sbb^{*})\ge0$ as $(\xb^{*}, \sbb^{*}) \geq \zero$ and $(\xb^{0}, \sbb^{0}) \geq \zero$. Second, as $(\xb^{*}, \sbb^{*}, \yb^{*}) \in \mathcal{S}$ (which implies $\xb^{*}\circ\sbb^{*}=\zero$), we have $(\xb^{*})^{\ts} \sbb^{*}=0$. The second inequality in eqn.~\eqref{eq:fin_bound} holds as $\xb^{*} \leq \xb^{0}$, $\sbb^{*} \leq \sbb^{0}$,  $(\xb^{*}, \sbb^{*}) \geq \zero$, and $(\xb^{0}, \sbb^{0}) \geq \zero$; combining them we get $(\xb^{0^{\ts}} \sbb^{*}+\sbb^{0^{\ts}} \xb^{*})\le2\,\xb^{0^\ts}\sbb^{0}$. Third inequality in eqn.~\eqref{eq:fin_bound} is true as we have $ \eta^{2} \xb^{0^{\ts}}+2 \eta(1-\eta) \xb^{0^{\ts}} \sbb^{0}=2 \eta\xb^{0^{\ts}} \sbb^{0}-\eta^2\xb^{0^{\ts}} \sbb^{0}\le2 \eta\xb^{0^{\ts}} \sbb^{0}$. The final inequality holds as $\eta \leq \frac{\xb^{\ts} \sbb}{\xb^{0^\ts} \sbb^{0}}$.\qed
	
	\noindent\textbf{Proof of eqn.~\eqref{eq:ineq2}.}
	The last inequality follows from eqn.~\eqref{eq:ineq}. The second to last inequality is also easy to prove as
	\begin{flalign}
	\|\Sb\xb^{0}\|_2=\sqrt{\sum_{i=1}^{s}(s_ix_i^{0})^2}\le\sqrt{\left(\sum_{i=1}^{s}s_ix_i^{0}\right)^2}=\sbb^\ts\xb^{0}\,.
	\end{flalign}
	
	\noindent To prove the first inequality in eqn.~\eqref{eq:ineq2}, we use the fact $\xb^{0}\ge\xb^{*}$ as follows:
	\begin{flalign}
	\|\Sb\xb^{0}\|_2^2-\|\Sb(\xb^{*}-\xb^{0})\|_2^2=&~\sum_{i=1}^\dimtwo(s_ix_i^{0})^2-\sum_{i=1}^\dimtwo s_i^2\left((x_i^{*})^2+(x_i^{0})^2-2x_i^{*}x_i^{0}\right)\nonumber\\
	=&~\sum_{i=1}^\dimtwo s_i^2\left(2x_i^{*}x_i^{0}-(x_i^{*})^2\right)\ge 0\nonumber\,.
	\end{flalign}\qed
	
	\noindent \textbf{Proof of eqn.~\eqref{eq:ineq3}.}
	To prove this we use a similar approach as in eqn.~\eqref{eq:ineq2}. The last inequality directly follows from eqn.~\eqref{eq:ineq}; the second to last inequality is also easy to prove as
	\begin{flalign}
	\|\Xb\sbb^{0}\|_2=\sqrt{\sum_{i=1}^{\dimtwo}(x_is_i^{0})^2}\le\sqrt{\left(\sum_{i=1}^{\dimtwo}x_is_i^{0}\right)^2}=\xb^\ts\sbb^{0}\,.
	\end{flalign}
	For the first inequality, we proceed as follows:
	\begin{flalign}
	\|\Xb(\sbb^{0}+\Ab^\ts\yb^{0}-\cbb)\|_2^2=&~\|\Xb\sbb^{0}\|_2^2+\|\Xb(\Ab^\ts\yb^{0}-\cbb)\|_2^2+2\sbb^{0^\ts}\Xb^\ts\Xb(\Ab^\ts\yb^{0}-\cbb)\nonumber\\
	=&~\|\Xb\sbb^{0}\|_2^2+\sum_{i=1}^\dimtwo x_i^2(\Ab^\ts\yb^{0}-\cbb)_i^2+2\sum_{i=1}^\dimtwo x_i^2s_i^{0}(\Ab^\ts\yb^{0}-\cbb)_i\nonumber\\
	\le&~\|\Xb\sbb^{0}\|_2^2+\sum_{i=1}^\dimtwo (x_is_i^{0})^2+2\sum_{i=1}^\dimtwo(x_is_i^{0})^2\nonumber\\
	=&~\|\Xb\sbb^{0}\|_2^2+\|\Xb\sbb^{0}\|_2^2+2\|\Xb\sbb^{0}\|_2^2=4\|\Xb\sbb^{0}\|_2^2.\label{eq:note4}
	\end{flalign}
	The inequality in eqn.~\eqref{eq:note4} follows from $x_i\ge0$, $s_i^{0}\ge0$ and $\abs{(\Ab^\ts\yb^{0}-\cbb)_i}\le s_i^{0}$ for all $i=1\ldots \dimtwo$.
\end{proof}

\noindent Our next result bounds $\|\Qb^{-\nicefrac{1}{2}}\pb\|_2$ which is instrumental in proving the final bound.
\begin{lemma}\label{thm:boundf}
	Let $(\xb^{0},\yb^{0},\sbb^{0})$ be the initial point with $(\xb^{0},\sbb^{0})>\zero$ such that $\xb^{0}\ge\xb^{*}$ and $\sbb^{0}\ge\max\{\sbb^{*},|\cbb-\Ab^\ts\yb^{0}|\}$ for some $(\xb^{*},\yb^{*},\sbb^{*})\in\mathcal{S}$. Furthermore, let $(\xb,\yb,\sbb)\in\mathcal{N}(\gamma)$ with $\rb=\eta\,\rb^{0}$ for some $0\le\eta\le1$. If the sketching matrix $\Wb\in\RR{\dimtwo}{w}$ satisfies the condition in eqn.~\eqref{eq:pdcond1}, then
	\begin{flalign}
	\|\Qb^{-\nicefrac{1}{2}}\pb\|_2\le~\sqrt{2}\left(\frac{9\dimtwo}{\sqrt{1-\gamma}}+\sigma\sqrt{\frac{\dimtwo}{1-\gamma}}+\sqrt{\dimtwo}\right)\sqrt{\mu}\nonumber\,.
	\end{flalign}
	Recall that $\rb=(\rb_p,\rb_d)=(\Ab\xb-\bb,\Ab^\ts\yb+\sbb-\cbb)$ and $\rb^{0}=(\rb_p^{0},\rb_d^{0})=(\Ab\xb^{0}-\bb,\Ab^\ts\yb^{0}+\sbb^{0}-\cbb)$\,.
\end{lemma}

\begin{proof}
Note that after correcting the approximation error of the iterative solver using $\vb$, the primal and dual residuals $\rb=(\rb_p,\rb_d)$ corresponding to an iterate $(\xb,\yb,\sbb)\in\mathcal{N}(\gamma)$ always lie on the line segment between zero and $\rb^{(0)}$. In other words, $\rb=\eta\rb^{(0)}$ always holds for some $\eta\in[0,1]$. This was formally proven in \citep[Lemma 3.3]{Mon03}.
	In order to bound $\|\Qb^{-\nicefrac{1}{2}}\pb\|_2$,  first we express $\pb$ as in eqn.~\eqref{eq:normal} and rewrite
	\begin{flalign}
	\Qb^{-\nicefrac{1}{2}}\pb
	=&~\Qb^{-\nicefrac{1}{2}}\left(-\rb_p-\sigma\mu\Ab\Sb^{-1}\one_\dimtwo+\Ab\xb-\Ab\Db^2\rb_d\right).\label{eq:recur3}
	\end{flalign}
	Then, applying the triangle inequality to $\|\Qb^{-\nicefrac{1}{2}}\pb\|_2$ in eqn.~\eqref{eq:recur3}, we get
	\begin{flalign}
	\|\Qb^{-\nicefrac{1}{2}}\pb\|_2\le \Delta_1+\Delta_2+\Delta_3+\Delta_4\label{eq:recur4}\,,
	\end{flalign}
	where
	\begin{flalign}
	\Delta_1=&~\|\Qb^{-\nicefrac{1}{2}}\rb_p\|_2\nonumber\,,\\
	\Delta_2=&~\sigma\mu\|\Qb^{-\nicefrac{1}{2}}\Ab\Db(\Xb\Sb)^{-\nicefrac{1}{2}}\one_\dimtwo\|_2\nonumber\,,\\
	\Delta_3=&~\|\Qb^{-\nicefrac{1}{2}}\Ab\Db\Db^{-1}\xb\|_2\,,\nonumber\\
	\Delta_4=&~\|\Qb^{-\nicefrac{1}{2}}\Ab\Db^2\rb_d\|_2\,.\nonumber
	\end{flalign}
	To bound $\Delta_1$, $\Delta_2$, $\Delta_3$ and $\Delta_4$ we heavily use the condition of eqn.~\eqref{eq:pdcond1}.
	
	\paragraph{Bounding $\Delta_1$.} Using $\rb_p=\eta\,\rb_p^{0}$, $\rb_p^{0}=\Ab\xb^{0}-\bb$ and $\bb=\Ab\xb^{*}$, we rewrite $\Delta_1$ as
	\begin{flalign}
	\Delta_1=&~\eta\,\|\Qb^{-\nicefrac{1}{2}}\Ab(\xb^{0}-\xb^{*})\|_2\nonumber\\
	=&~\eta\,\|\Qb^{-\nicefrac{1}{2}}\Ab\Db\Db^{-1}(\xb^{0}-\xb^{*})\|_2\nonumber\\
	\le&~\eta\,\|\Qb^{-\nicefrac{1}{2}}\Ab\Db\|_2\|\Db^{-1}(\xb^{0}-\xb^{*})\|_2\nonumber\\
	\le&~\sqrt{2}\eta\,\|\Db^{-1}(\xb^{0}-\xb^{*})\|_2\nonumber\\
	=&~\sqrt{2}\eta\,\|(\Xb\Sb)^{-\nicefrac{1}{2}}\Sb(\xb^{0}-\xb^{*})\|_2\nonumber\\
	\le&~\sqrt{2}\eta\,\|(\Xb\Sb)^{-\nicefrac{1}{2}}\|_2\,\|\Sb(\xb^{0}-\xb^{*})\|_2\,,\label{eq:del11}
	\end{flalign}
	where the above steps follow from submultiplicativity and  eqn.~\eqref{eq:pdcond1}. From eqn.~\eqref{eq:pdcond1}, note that we have $\|\Qb^{-\nicefrac{1}{2}}\Ab\Db\|_2\le\sqrt{2}$ as $\zeta\le 1$\,. Now, applying eqn.~\eqref{eq:ineq2} and $\|(\Xb\Sb)^{-\nicefrac{1}{2}}\|_2=\max_{1\le i \le \dimtwo}\frac{1}{\sqrt{x_is_i}}$, we further have
	\begin{flalign}
	\Delta_1\le&~\sqrt{2}\,\max_{1\le i \le \dimtwo}\frac{1}{\sqrt{x_is_i}}\cdot 3\dimtwo\mu\nonumber\\
	\le&~ 3\sqrt{2}\,\dimtwo\sqrt{\frac{\mu}{1-\gamma}}\label{eq:del12}\,,
	\end{flalign}
	where the last inequality follows from $(\xb,\yb,\sbb)\in\mathcal{N}(\gamma)$.
	
	\paragraph{Bounding $\Delta_2$.} Applying submultiplicativity, we get
	\begin{flalign}
	\Delta_2=&~\sigma\mu\,\|\Qb^{-\nicefrac{1}{2}}\,\Ab\Db\,(\Xb\Sb)^{-\nicefrac{1}{2}}\one_\dimtwo\|_2\nonumber\\
	\le&~\sigma\mu\,\|\Qb^{-\nicefrac{1}{2}}\,\Ab\Db\|_2\|(\Xb\Sb)^{-\nicefrac{1}{2}}\one_\dimtwo\|_2\nonumber\\
	\le&~\sqrt{2}\,\sigma\mu\,\|(\Xb\Sb)^{-\nicefrac{1}{2}}\one_\dimtwo\|_2\nonumber\\
	=&~\sqrt{2}\,\sigma\mu\,\sqrt{\sum_{i=1}^{\dimtwo}\frac{1}{x_i s_i}}
	\le~\sqrt{2}\,\sigma\mu\,\sqrt{\sum_{i=1}^{\dimtwo}\frac{1}{(1-\gamma)\mu}}\nonumber\\
	=&~\sqrt{2}\,\sigma\,\sqrt{\frac{\dimtwo\,\mu}{(1-\gamma)}}\label{eq:del2}\,,
	\end{flalign}
	where the second to last inequality follows from eqn.~\eqref{eq:pdcond1} and the last inequality holds as $(\xb,\yb,\sbb)\in\mathcal{N}(\gamma)$.
	
	\paragraph{Bounding $\Delta_3$.} Using $\Db=\Sb^{-\nicefrac{1}{2}}\Xb^{\nicefrac{1}{2}}$ and $\xb=\Xb\,\one_\dimtwo$ we get
	\begin{flalign}
	\Delta_3=&~\|\Qb^{-\nicefrac{1}{2}}\,\Ab\Db\,(\Sb^{\nicefrac{1}{2}}\Xb^{-\nicefrac{1}{2}})\,\Xb\,\one_\dimtwo\|_2\nonumber\\
	=&~\|\Qb^{-\nicefrac{1}{2}}\,\Ab\Db\,(\Sb\Xb)^{\nicefrac{1}{2}}\,\one_\dimtwo\|_2\nonumber\\
	\le&~\|\Qb^{-\nicefrac{1}{2}}\,\Ab\Db\|_2\|(\Sb\Xb)^{\nicefrac{1}{2}}\,\one_\dimtwo\|_2\nonumber\\
	\le&~\sqrt{2}\,\sqrt{\sum_{i=1}^{\dimtwo}x_i s_i}=~\sqrt{2\dimtwo\,\mu}\label{eq:del3}\,,
	\end{flalign}
	where the inequalities follow from submultiplicativity and eqn.~\eqref{eq:pdcond1}.
	
	\paragraph{Bounding $\Delta_4$.} Using $\rb_d=\eta\,\rb_d^{0}$, we have
	\begin{flalign}
	\Delta_4=&~\eta\|\Qb^{-\nicefrac{1}{2}}\,\Ab\,\Db^2\rb_d^{0}\|_2\nonumber\\
	\le&~ \eta\|\Qb^{-\nicefrac{1}{2}}\,\Ab\Db\|_2\|(\Xb\Sb)^{-\nicefrac{1}{2}}\Xb\rb_d^{0}\|_2\nonumber\\
	\le&~ \sqrt{2}\eta\,\|(\Xb\Sb)^{-\nicefrac{1}{2}}\Xb(\Ab^\ts\yb^{0}+\sbb^{0}-\cbb)\|_2\nonumber\\
	\le&~\sqrt{2}\eta\,\|(\Xb\Sb)^{-\nicefrac{1}{2}}\|_2\,\|\Xb(\Ab^\ts\yb^{0}+\sbb^{0}-\cbb)\|_2\,,\nonumber
	\end{flalign}
	where the above inequalities follow from submultiplicativity and eqn.~\eqref{eq:pdcond1}. Now, applying eqn.~\eqref{eq:ineq3} and $\|(\Xb\Sb)^{-\nicefrac{1}{2}}\|_2\le\frac{1}{\sqrt{(1-\gamma)\mu}}$, we further have
	\begin{flalign}
	\Delta_4\le~6\sqrt{2}\dimtwo\sqrt{\frac{\mu}{1-\gamma}}.\label{eq:del4}
	\end{flalign}
	
	\paragraph{Final bound.} Combining eqns.~\eqref{eq:recur4}, \eqref{eq:del12}, ,\eqref{eq:del2}, \eqref{eq:del3} and \eqref{eq:del4}, we get
	\begin{flalign}
	\|\Qb^{-\nicefrac{1}{2}}\pb\|_2\le~\sqrt{2}\left(\frac{9\dimtwo}{\sqrt{1-\gamma}}+\sigma\sqrt{\frac{\dimtwo}{1-\gamma}}+\sqrt{\dimtwo}\right)\sqrt{\mu}\,.
	\end{flalign}
	This concludes the proof of Lemma~\ref{thm:boundf}.
\end{proof}

\subsection{Determining step-size, bounding the number of iterations, and proof of Theorem~\ref{thm:1}}

Assume that the triplet $(\hat{\Delx},\hat{\Dely},\hat{\Dels})$ satisfies eqns.~\eqref{eq:addl_i}, \eqref{eq:delxhat_i}, and \eqref{eq:delshat_i}. We rewrite this system in the following alternative form:
\begin{subequations}\label{eq:iip_3}
	\begin{flalign}
	\Ab\hat{\Delta\xb}=&-\rb_p, \label{eq:iip_3_1}\\
	\Ab^\ts\hat{\Delta\yb}+\hat{\Delta\sbb}=&-\rb_d,  \label{eq:iip_3_2}\\
	\Xb\hat{\Delta\sbb}+\Sb\hat{\Delta\xb}=&-\Xb\Sb\,\one_\dimtwo+\sigma\mu\,\one_\dimtwo - \vb.  \label{eq:iip_3_3}
	\end{flalign}
\end{subequations}
Indeed, we now show how to derive eqns.~\eqref{eq:delxhat_i}, \eqref{eq:addl_i} and \eqref{eq:delshat_i} from eqn.~\eqref{eq:iip_3}. Pre-multiplying both sides of eqn.~\eqref{eq:iip_3_3} by $\Ab\Sb^{-1}$ and noting that $\Db^2=\Xb\Sb^{-1}$, we get
\begin{flalign}
&~\Ab\Db^2\hat{\Delta\sbb}+\Ab\hat{\Delx}=-\Ab\Xb\one_n+\sigma\mu\Ab\Sb^{-1}\one_n-\Ab\Sb^{-1}\vb\nonumber\\
\Rightarrow&~\Ab\Db^2\hat{\Delta\sbb}=-\Ab\xb+\rb_p+\sigma\mu\Ab\Sb^{-1}\one_n-\Ab\Sb^{-1}\vb.\label{eq:iip_31}
\end{flalign}
Eqn.~\eqref{eq:iip_31} holds as $\Ab\Xb\one_n=\Ab\xb$ and, from eqn.~\eqref{eq:iip_3_1}, $\Ab\hat{\Delx}=-\rb_p$. Next, pre-multiplying eqn.~\eqref{eq:iip_3_2} by $\Ab\Db^2$, we get
\begin{flalign}
&~\Ab\Db^2\Ab^\ts\hat{\Dely}+\Ab\Db^2\hat{\Dels}=-\Ab\Db^2\rb_d\nonumber\\
\Rightarrow &~\Ab\Db^2\Ab^\ts\hat{\Dely}=-\rb_p-\sigma\mu\Ab\Sb^{-1}\one_n+\Ab\xb-\Ab\Db^2\rb_d+\Ab\Sb^{-1}\vb=\pb+\Ab\Sb^{-1}\vb.\label{eq:iip_311}
\end{flalign}
The first equality in eqn.~\eqref{eq:iip_311} follows from eqn.~\eqref{eq:iip_31} and the definition of $\pb$. This establishes
eqn.~\eqref{eq:addl_i}. Eqn.~\eqref{eq:delshat_i} directly follows from eqn.~\eqref{eq:iip_3_2}. Finally, we get eqn.~\eqref{eq:delxhat_i} by pre-multiplying eqn.~\eqref{eq:iip_3_3} by $\Sb^{-1}$.

Next, we define each new point traversed by the algorithm as $(\xnew,\ynew, \snew)$, where
\begin{flalign}
(\xnew,\ynew,\snew) &= (\xb, \yb, \sbb) + \alpha(\hat{\Delta\xb},\hat{\Delta \yb},\hat{\Delta\sbb}), \\
\munew &= \xnew^\ts \snew/ \dimtwo, \\
\rb(\alpha) &= \rb\left(\xnew,\snew,\ynew\right).
\end{flalign}
The goal in this section is to bound the number of iterations required by Algorithm~\ref{algo:ipm_i}.
Towards that end, we bound the magnitude of the step size $\alpha$. First, we provide an upper bound on $\alpha$, which allows us to show that each new point $(\xnew,\snew,\ynew)$ traversed by the algorithm stays within the neighborhood $\neigh$.
Second, we provide a lower bound on $\alpha$, which allows us to bound the number of iterations required. We use multiple lemmas from~\citep{Mon03}, which we reproduce here, without their proofs.

First, we provide an upper bound on $\alpha$, ensuring that each new point $(\xnew,\ynew,\snew)$ traversed by the algorithm stays within the neighborhood $\neigh$.
\begin{lemma}[Lemma 3.5 of \citep{Mon03}]\label{lemmamaxalpha1}
Assume $(\hat{\Delx},\hat{\Dely},\hat{\Dels})$  satisfies eqns.~(\ref{eq:iip_3}) for some $\sig > 0$, $(\xb, \yb,\sbb) \in \neigh$ (for $\gamma \in (0,1)$), and $\|\vb\|_2 \leq \frac{\gamma \sig \mu}{4}$. Then, $(\xnew,\ynew,\snew) \in \neigh$ for every scalar $\alpha$ such that
\begin{flalign} \label{alphalemma1}
0 \leq \alpha \leq \min \left\{ 1, \frac{\gamma \sig \mu}{4 \delxdelsnorm}\right\} .
\end{flalign}
\end{lemma}
\noindent We now provide a lower bound on the values of $\alphaused$ and the corresponding $\mu(\alphaused)$; see Algorithm~\ref{algo:ipm_i}.
\begin{lemma}[Lemma 3.6 of~\citep{Mon03}]\label{lemmaminalpha1}
In each iteration of Algorithm~\ref{algo:ipm_i}, if $\|\vb\|_2\le\frac{\gamma\sigma\mu}{4}$, then the step size $\alphaused$ satisfies
\begin{flalign} \label{alphalemma2}
\alphaused \geq \min \left\{ 1,  \frac{ \min   \{ \gamma \sigma, (1 - \frac{5}{4} \sigma) \} \mu } {4  \|\hat{\Delta x} \circ \hat{\Delta s}\|_\infty}  \right\}
	\end{flalign}
	and
	\begin{flalign} \label{mulemma2}
	\mu(\alphaused)= \Big[ 1 - \frac{\alphaused}{2} (1-\frac{5}{4}\sig) \Big] \mu.
\end{flalign}
\end{lemma}
\noindent At this point, we have provided a lower bound (eqn.~(\ref{alphalemma2})) for the allowed values of the step size $\alphaused$. Next, we show that this lower bound is bounded away from zero. From eqn.~(\ref{alphalemma2}) this is equivalent to showing that $\delxdelsnorm$ is bounded.
\begin{lemma}[Lemma 3.7 of~\citep{Mon03} (slightly modified)]\label{lemmaminalpha2}
Let $(\xb^{0},\yb^{0},\sbb^{0})$ be the initial point with $(\xb^{0},\sbb^{0})>0$ and $(\xb^0,\sbb^0) \geq (\xb^{*}, \sbb^{*}) $ for some $(\xb^{*},\yb^{*},\sbb^{*}) \in \mathcal{S}$. Let $(\xb,\yb,\sbb) \in \neigh$ be such that $\rb = \eta \rb^{0}$ for some $\eta \in [0,1]$ and $\|\vb\|_2\le\frac{\gamma\sigma\mu}{4}$. Then, the search direction  $(\hat{\Delta \xb},\hat{\Delta \yb},\hat{\Delta \sbb})$ produced by Algorithm \ref{algo:ipm_i} at each iteration satisfies
\begin{flalign} \label{normsmax}
\max \{\| \Db^{-1} \hat{\Delta \xb}\|_2, \| \Db \hat{\Delta \sbb}\|_2 \} \le~\left(1+\frac{\sigma^2}{1-\gamma}-2\sigma\right)^{\nicefrac{1}{2}}\sqrt{\dimtwo\mu}
+\frac{6\dimtwo}{\sqrt{(1-\gamma)}}\sqrt{\mu}+\frac{\gamma\sigma}{4\,\sqrt{1-\gamma}}\sqrt{\mu}.
\end{flalign}
\end{lemma}
\noindent We should note here that the above lemma is slightly different from~\citep[Lemma 3.7]{Mon03}. Indeed, \citep[Lemma 3.7]{Mon03} actually proves the following bound:
\begin{flalign} \label{normsmax2}
\max \{\| \Db^{-1} \hat{\Delta \xb}\|_2, \| \Db \hat{\Delta \sbb}\|_2 \}
\le~\left(1+\frac{\sigma^2}{1-\gamma}-2\sigma\right)^{\nicefrac{1}{2}}\sqrt{\dimtwo\mu}+\frac{6\dimtwo}{\sqrt{(1-\gamma)}}\sqrt{\mu}+\frac{\gamma\sigma}{4\sqrt{n}}\sqrt{\mu}\,.
\end{flalign}
Notice that there is slight difference in the last term in the right-hand side, which does not asymptotically change the bound. The underlying reason for this difference is the fact that~\citep{Mon03} constructed the vector $\vb$ differently. In our case, we need to bound $\|(\Xb\Sb)^{-1/2}\vb\|_2$, which we do as follows:
\begin{flalign}
\| (\Xb \Sb)^{-1/2}\vb\|_2\le\| (\Xb \Sb)^{-1/2}\|_2\,\|\vb\|_2\le \frac{1}{\min_i \sqrt{x_is_i}}\,\frac{\gamma\sigma\mu}{4}\label{eq:db1}\,,
\end{flalign}
where in the above expression we use the fact that $\| (\Xb \Sb)^{-1/2}\|_2=\frac{1}{\min_i\sqrt{x_is_i}}$. Now as $(\xb,\yb,\sbb)\in\mathcal{N}(\gamma)$, we further have $x_is_i\ge (1-\gamma)\mu$ for all $i=1\ldots n$. Combining this with eqn.~\eqref{eq:db1}, we get
\begin{flalign}
\| (\Xb \Sb)^{-1/2}\vb\|_2\le\frac{\gamma\sigma\mu}{4\sqrt{(1-\gamma)\mu}}=\frac{\gamma\sigma}{4\,\sqrt{1-\gamma}}\sqrt{\mu}.\label{eq:bd2}
\end{flalign}
On the other hand,~\citep{Mon03} had a different construction of $\vb$ for which $\|(\Xb\Sb)^{-1/2}\vb\|_2=\|\tilde{\fb}^{(t)}\|_2$ holds. Therefore they had the following bound:
$$\|(\Xb\Sb)^{-1/2}\vb\|_2=\|\tilde{\fb}^{(t)}\|_2\le\frac{\gamma\sigma}{4\sqrt{n}}\sqrt{\mu}.$$
The next lemma bounds the number of iterations that Algorithm~\ref{algo:ipm_i} needs when started with an infeasible point that is sufficiently positive.
\begin{lemma}[Theorem 2.6  of \citep{Mon03}] \label{theoremOuter}
Assume that the constants $\gamma$ and $\sigma$ are such that $\max\{\gamma^{-1},(1-\gamma)^{-1},\sigma^{-1},(1-\frac{5}{4}\sigma)^{-1}\}=\Ocal(1)$. Let the initial point $(\xb^{0},\sbb^{0},\yb^{0})$ satisfy $(\xb^{0}, \sbb^{0}) \geq (\xb^{*}, \sbb^{*} )$ for some $(\xb^{*}, \sbb^{*},\yb^{*}) \in \mathcal{S}$ and $\|\vb\|_2\le\frac{\gamma\sigma\mu}{4}$. Algorithm \ref{algo:ipm_i} generates an iterate $(\xb^{k}, \sbb^{k}, \yb^{k})$ satisfying $\mu_k \leq \epsilon \mu_0$ and $\| \rb^{k}\|_2 \leq \epsilon \| \rb^{0}\|_2$ after
$\mathcal{O}(\dimtwo^2 \log{\nicefrac{1}{\epsilon}})$ iterations.
\end{lemma}
\noindent Finally, Theorem~\ref{thm:1} follows from Lemmas~\ref{lem:conouter} and~\ref{theoremOuter}.

\section{Additional notes on experiments}\label{app:experiments}

\subsection{Support Vector Machines (SVMs)}
\label{app:svm}

The classical $\ell_1$-SVM problem is as follows. We consider the task of fitting an SVM to data pairs $S = \{ (x_i, y_i)\}_{i=1}^m$, where $x_i \in \mathbb{R}^n$ and $y_i \in \{ + 1, - 1\}$. Here, $m$ is the number of training points, and $n$ is the feature dimension. The SVM problem with an $\ell_1$ regularizer has the following form:
\begin{align} \label{svm2}
\underset{{ w  }}{\operatorname{minimize}}  \quad &  \|w \|_1  \\
\text{subject to} \quad& y_i (w^T x_i + b') \geq 1, \quad i=1\ldots m.  \nonumber 
\end{align}
This problem can be written as an LP by introducing the variables $w^+$ and $w^-$, where $w = w^+ - w^-$. The objective becomes  $\sum_{j=1}^n w^+_j + w^-_j$, and we constrain $w^+_i \geq 0$ and $w^-_i \geq 0$. Note that the size of the constraint matrix in the LP becomes $m \times (2n +1)$.

\subsection{Random data}
\label{app:rand}
We generate random synthetic instances of linear programs as follows. To generate $\Ab \in \mathbb{R}^{m \times n}$, we set ${a}_{ij} \sim_{i.i.d.}U(0,1)$ with probability $p$ and ${a}_{ij} = 0$ otherwise. We then add min$\{m,n\}$  i.i.d. draws from $U(0,1)$ to the main diagonal, to ensure each row of $\Ab$ has at least one nonzero entry. We set $\bb = \Ab \xb + 0.1\zb$, where $\xb$ and $\zb$ are random vectors drawn from $N(0,1)$. Finally, we set $c \sim N(0,1)$.

\subsection{Real-world data}
\label{app:real}

We used a gene expression cancer RNA-Sequencing dataset, taken from the UCI Machine Learning repository. It is part of the RNA-Seq (HiSeq) PANCAN data set~\citep{Weinstein2013} and is a random extraction of gene expressions from patients who have different types of tumors: BRCA, KIRC, COAD, LUAD, and PRAD. We considered the binary classification task of identifying BRCA versus other types.

We also used the DrivFace dataset taken from the UCI Machine Learning repository. In the DrivFace dataset, each sample corresponds to an image of a human subject, taken while driving in real scenarios. Each image is labeled as corresponding to one of three possible gaze directions: left, straight, or right. We considered the binary classification task of identifying two different gaze directions: straight, or to either side (left or right).

\subsection{Feasible starting point}
\label{app:feasible_start}

We construct a linear program to find a primal feasible starting point $(\xb_0, \sbb_0, \yb_0)$ such that $\Ab\xb_0 = \bb$.  Without loss of generality, assume that all entries of $\bb$ are positive and let $\zb \in \R{\dimone}$.  Then, $(\xb, \zb)$ is an optimal solution to the following linear program when $\zb = \zero$ and $\Ab\xb = \bb$.
\begin{flalign*}
	\min\,\zb\,,\text{ subject to }\Ab\xb + \Ib \zb =\bb\,,\xb, \zb \ge \zero\,,
\end{flalign*}